\documentclass[letterpaper,12pt, unsignedthesis]{syrthesis}
\usepackage{indentfirst}
\usepackage[margin=2cm]{caption}
\usepackage{enumitem}
\usepackage{epsfig,amsmath,amssymb,amsfonts,amstext,amsthm,mathrsfs}
\usepackage{latexsym,graphics,epsf,epsfig,psfrag}
\usepackage{cite}
\usepackage{color}
\usepackage{epstopdf}
\usepackage{subcaption}
\usepackage{graphicx}
\usepackage{titlesec}
\usepackage{cleveref}
\usepackage{float}
\titleformat{\paragraph}
{\normalfont\normalsize\bfseries}{\theparagraph}{1em}{}
\titlespacing*{\paragraph}
{0pt}{3.25ex plus 1ex minus .2ex}{1.5ex plus .2ex}

\newtheorem{corollary}{\textbf{Corollary}}

\newtheorem{theorem}{\textbf{Theorem}}
\newtheorem{proposition}{\textbf{Proposition}}
\newtheorem{remark}{\textbf{Remark}}

\newcommand{\nn}{\nonumber}

\newcommand{\cX}{\mathcal{X}}

\newcommand{\cY}{\mathcal{Y}}

\newcommand{\cW}{\mathcal{W}}
\newcommand{\cS}{\mathcal{S}}
\newcommand{\cC}{\mathcal{C}}

\newcommand{\cB}{\mathcal{B}}

\newcommand{\cA}{\mathcal{A}}

\newcommand{\tl}{\tilde{l}}

\newcommand{\tv}{\tilde{v}}

\newcommand{\tR}{\tilde{R}}

\newcommand{\hl}{\hat{l}}

\newcommand{\hw}{\hat{w}}

\newcommand{\hv}{\hat{v}}

\DeclareMathAlphabet{\matheuf}{U}{euf}{m}{n}

\newcommand{\ignore}[1]{}

\renewcommand{\contentsname}{Table of Contents}

\title{Information Theoretic Limits of State-dependent Networks}
\author{Yunhao Sun}
\degree{Doctor of Philosophy}

\msdegree{M.S.(EE)}
\gradinst{NYU Tandon School of Engineering} 		
\msyear{2013} 					
\program{Electrical Engineering}

\ugdegree{B.E.(EE)} 					
\uginst{Beijing Institute of Technology} 		
\ugyear{2011} 					
\program{Electrical Engineering} 			
\thesistype{DISSERTATION} 					
\thesismonthyr{July}{2017} 		
\thesisadvisorA{Prof. Yingbin Liang} 	

\begin{document}
\frontmatter

\newpage
\makefinaltitle

\tableofcontents
\listoffigures

\bodychapters


\mainmatter
\chapter*{Abstract}
\addcontentsline{toc}{chapter}{Abstract}

We investigate the information theoretic limits of two types of state-dependent models in this dissertation. These models capture a wide range of wireless communication scenarios where there are interference cognition among transmitters. Hence, information theoretic studies of these models provide useful guidelines for designing new interference cancellation schemes in practical wireless networks.

In particular, we first study the two-user state-dependent Gaussian multiple access channel (MAC) with a helper. The channel is corrupted by an additive Gaussian state sequence known to neither the transmitters nor the receiver, but to a helper noncausally, which assists state cancellation at the receiver. Inner and outer bounds on the capacity region are first derived, which improve the state-of-the-art bounds given in the literature. Further comparison of these bounds yields either segments on the capacity region boundary or the full capacity region by considering various regimes of channel parameters.

We then study the two-user Gaussian state-dependent Z-interference channel (Z-IC), in which two receivers are corrupted respectively by two correlated states that are noncausally known to transmitters, but unknown to receivers. Three interference regimes are studied, and the capacity region or the sum capacity boundary is characterized either fully or partially under various channel parameters. The impact of the correlation between the states on the cancellation of state and interference as well as the achievability of the capacity is demonstrated via numerical analysis.

Finally, we extend our results on the state-dependent Z-IC to the state-dependent regular IC. As both receivers in the regular IC are interfered, more sophisticated achievable schemes are designed. For the very strong regime, the capacity region is achieved by a scheme where the two transmitters implement a cooperative dirty paper coding. For the strong but not very strong regime, the sum-rate capacity is characterized by rate splitting, layered dirty paper coding and successive cancellation. For the weak regime, the sum-rate capacity is achieved via dirty paper coding individually at two receivers
as well as treating interference as noise. Numerical investigation indicates that for the regular IC, the correlation between states impacts the achievability of the channel capacity in a different way from that of the Z-IC.


\chapter{Introduction}\label{chap:Introduction}
Since the early 1940s, when Claude E. Shannon established the maximum amount of information that can be sent over a noisy channel in his classic papers \cite{Shannon1948} and \cite{Shannon59}, information theory has driven the evolution of communication systems from one generation to another. Shannon's original work focused on the discrete memoryless channel, in which the transition probability distribution (i.e., the noise characteristics of the channel) is perfectly known to both the transmitter and the receiver. In this scenario, he proved theoretically the existence of coding and decoding schemes to achieve any rate below the channel capacity, and proved that such capacity is the reliable transmission limit via a converse argument. In particular, Shannon came up with the idea of {\em random coding} (see \cite{Csiszar98}) to show the achievability of the rate, the ingenious tool that is still widely used in information theory today and is used throughout the rest of this dissertation.

Then, Shannon's basic approach was extended by both mathematicians and engineers to more general models with respect to information sources, coding structures, and performance measures. The fundamental theorem for entropy was extended to the same generality as the ordinary ergodic theorems by McMillan \cite{Mcmillan53} and Breiman \cite{breiman1957} and the result is now known as the Shannon-McMillan-Breiman theorem (the asymptotic equipartition theorem or AEP, the ergodic theorem of information theory, and the entropy theorem). A variety of detailed proofs of the basic coding theorems and stronger versions of the theorems for memoryless, Markov, and other special cases of random processes were developed.

In \cite{Gray1990}, Robert M. Gray pointed out that there are two primary goals of information theory: The first is the development of the fundamental theoretical limit on the achievable performance when communicating a given information source over a given communication channel using optimal (but only theoretical) coding schemes from within a prescribed class. The second goal is the development of practical coding schemes, e.g., structured encoder(s) and decoder(s), which provide the performance that is reasonably good in comparison with the optimal performance given by the theory. 

During the development of practical communication systems, both of these two goals need to be fulfilled. In this dissertation, we mainly investigate the first aspect of information theory. In particular, we focus on a type of state-dependent channels, and explore the dirty paper coding as a useful tool to understand the fundamental performance limit of such a type of channels. The remainder of this chapter provides background materials and outlines the contributions of this dissertation.

\section{Motivation}

Interference management is one of important issues that determine the spectral efficiency of wireless networks. Techniques to deal with interference in all up-to-date cellular networks follow the basic principle of {\em orthogonalizing} transmissions in time, frequency, code, and space, which yields TDMA, FDMA, CDMA and more advanced OFDM technologies. However, orthogonal schemes typically do not reach the best spectral efficiency and are not information theoretically optimal in general. On the other hand, various advanced {\em non-orthogonal} interference cancellation schemes are proposed, motivated by information theoretic designs. For example, the Han-Kobayashi \cite{Han81} scheme is based on the idea that the receiver decodes interference partially, and then subtracts it from the output to reduce the interference. Such a scheme is recently exploited in a down-link non-orthogonal multiple access (NOMA) scheme for interference management in \cite{Saito13,Li14} and \cite{Xu15}. However, such successive interference cancellation requires users to share codebooks and hence can be very complex to implement in practice, especially when transmissions are not within the same network domain.


My dissertation aims at investigating a new framework for interference control in wireless networks based on the following key perspectives of the interference. In fact, interference signals in nature contain coded information sent to certain intended nodes, and hence such signals as codeword sequences are typically {\em noncausally} known by various nodes in the network. For example, an interferer clearly knows the interference signal that it causes to other users noncausally, because such interference is the codeword that this interferer transmits to its intended receivers. As another example, if the interferer is a base station, it can easily inform other base stations about its interference via the backhaul network or inform access points in the cell via wired links. The major observation here is that {\em interference cognition} (i.e., the knowledge of interference being informed to various nodes) naturally exists or can be established at very low costs in networks. Thus, nodes that possess noncausal interference information should be able to exploit it to assist cancellation of such interference.

One major advantage of such an idea is that the design of interference cancellation is handled mainly at the interferer or the helper side, which are typically powerful nodes (such as base stations and access points) in networks and can hence easily take the extra load of interference cancellation. Since the design is on one side, it is more efficient and does not require sharing codebooks as in successive  interference cancellation in the Han-Kobayashi scheme. Moreover, the design can be made transparent to nodes being interfered with. This is very useful in cognitive networks (i.e., \cite{Mitola:IPC:99,Haykin:JSAC:05,Akyi06}) and internet of things (IoT) networks \cite{Ashton09,Atzori10,Gubbi13}. With the interference to primary nodes being canceled by the interferer itself or helper nodes, the access of primary channels can be made simultaneous and transparent from primary networks.




%

In this dissertation, we explore two types of information theoretic models, which capture the impact of interference cognition in wireless networks, and set the goal of best exploiting such information for interference cancellation. These two types of models respectively represent two angles to treat the interference cognition and thus focus on two performance objectives to accomplish. We next explain the essence of these two types of models by their basic versions.
\begin{list}{$\bullet$}{\topsep=0.ex \leftmargin=0.3in \rightmargin=0.1in \itemsep =-0in}

	\item The first type of models are state-dependent channels with the state known at a helper, where neither the transmitter nor the receiver knows the information of channel interference. Instead, there is a dedicate helper who has the information of channel interference and assists the receiver to cancel the channel state. This model can be illustrated via an example (see Fig.~\ref{fig:hmode}). In this model, we consider a device-to-device (D2D) communication in a picocell located inside a macrocell of a cellular network. The D2D communication inside the picocell is corrupted by the interference $s^n$. The helper (i.e., wireless router) in the picocell, which knows the information of channel state through the wire cable, sends its signal $X_0^n$ to help the D2D communication.  Although the help signal $X_0^n$ may also cause interference to the macrocell base station, as long as the power of $X_0^n$ is much less than the power of $S^n$, there is still significant gain in throughput. In fact, our results  demonstrate that the helper can use a relatively small amount of power to completely cancel the interference that the base station causes to D2D users (e.g., the picocell users in our example) even if the interference is as large as infinite.
	\vspace{5mm} 
	\begin{figure}[H]
		\centering
		\includegraphics[width=3in]{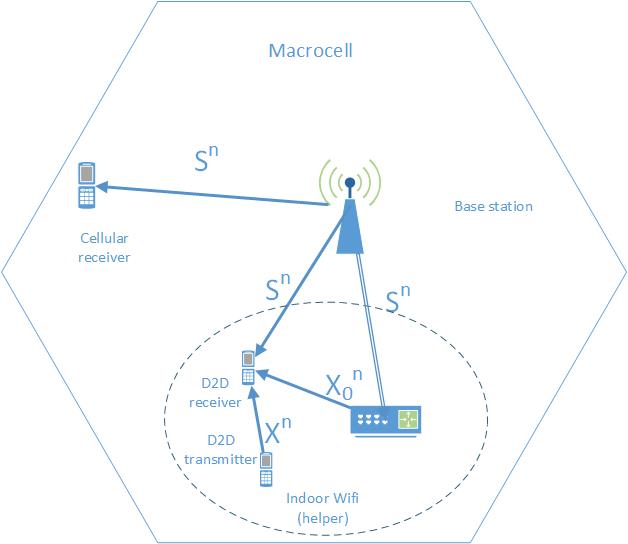}
		\caption{state-dependent channels with states known at helper} \label{fig:hmode}
	\end{figure}
	\vspace{5mm}
	 \item The second type of models are state-dependent channels with the state known at transmitters, so that the transmitters can help the receiver to cancel the state interference. This class of models can be illustrated via a simple example as follows. Consider a multi-cell network, where base station 1 sends information to a cellular user. It is typical that base station 2 in cell 2 can cause interference to the cellular user in cell 1. In this case, base station 2 can send such interference to base station 1 through the backhaul network, so that base station 1 can use its noncausal knowledge about the interference to cancel the interference efficiently.
	\vspace{5mm}
	\begin{figure}[H]
		\centering
		\includegraphics[width=3in]{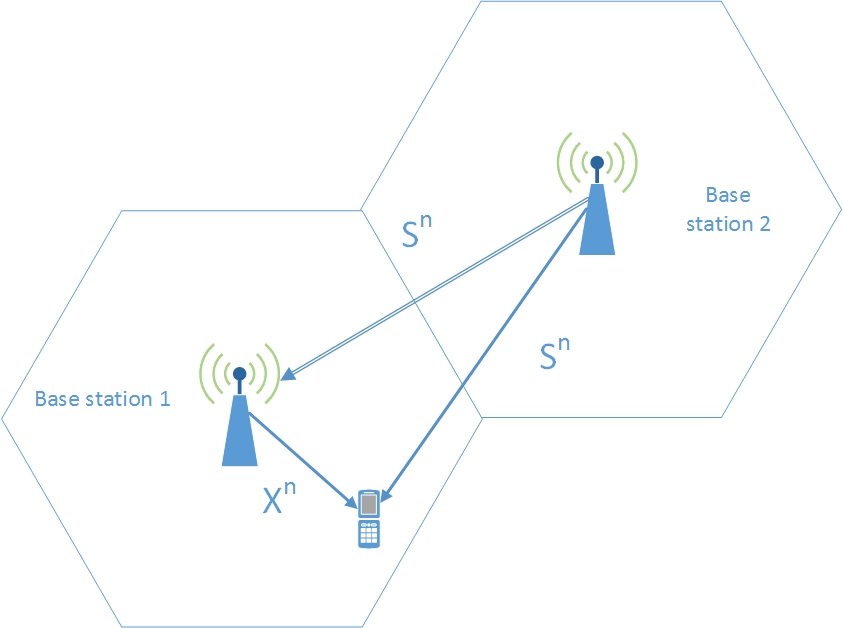}
		\caption{state-dependent channels with states known at transmitter} \label{fig:ctmode}
	\end{figure}
	\vspace{5mm}
%
%
\end{list}

The above two basic scenarios can be broadened to a wide range of network models to capture various scenarios arising in practical networks: multiple receivers can be interfered with by the same or different interfering sources, multiple helpers can cooperatively assist interference cancellation, and interference from different interfering sources can be known distributively by different helpers. All of these models have their representative characteristics which give rise to unique designs and technical challenges, and comprehensive exploration of these models will yield a new framework for interference management based on interference cognition.

%

Thus, the goal of this dissertation is to conduct extensive investigations of the two types of models in order to develop comprehensive understandings of the impact of cognitive interference on interference management in various network environments and devise a set of analysis techniques for characterizing the information theoretic performance of these models. Our studies will provide useful guidelines to significantly improve interference management technologies in practical wireless networks. In particular, we design new interference/state cancellation schemes that maximize the performance of these systems, and characterize the fundamental communication limits of the basic models. Then, we can understand the impact of different channel parameters on the achievability of the channel capacity through our numerical analysis.

\section{Related Work}
In this section, we introduce studies in the literature that are related to the results in this dissertation. The study of state-dependent channels were initiated by Shannon in \cite{Shannon1958}, in which the channel model with causal knowledge of the channel state at the transmitter was studied. In \cite{Gelf80}, the point-to-point channel with the state known noncausally at the transmitter was studied, and the capacity was obtained for the discrete memoryless channel via Gel'fand-Pinsker binning. Based on this result, in \cite{Costa83}, the capacity for the state-dependent point-to-point Gaussian channel was obtained, and it was shown that the state can be perfectly canceled as if there is no state interference. The achievable scheme was referred to as "dirty paper coding". 

Various state-dependent network models have been studied, including the state-dependent multiple-access channel (MAC) in \cite{Kim05,Lapidoth13_a,Zaidi13,Lapidoth13_b,Li13}, the state-dependent broadcast channel in \cite{Steinberg05_2,Khos09,Lapidoth11,Duan14ITW}, the state-dependent relay channel in \cite{Aref09,Khos09,Zaidi10,Zaidi11}, and the state-dependent interference channel (IC) as we discuss below.


More closely to our work on the state-dependent ICs in Chapter \ref{chap:z2state} and Chapter \ref{cha:2state} are a few studies on the state-dependent ICs as follows. A state-dependent IC model was studied in \cite{Zhang11a,Zhang11b} with two receivers corrupted by the same state, and in \cite{Ghas13} with two receivers corrupted by independent states. More recently, in\cite{Duan16IT}, both the state-dependent regular IC and Z-IC were studied, where the receivers are corrupted by the same but differently scaled state. Furthermore, in \cite{Duan13ITW,Ghas14}, a type of the state-dependent Z-IC was studied, in which only one receiver is corrupted by the state and the state information is known only to the other transmitter. In \cite{Fehri2015Z}, a type of the state-dependent Z-IC with two states was studied, where each transmitter knows only the state that corrupts its corresponding receiver. In \cite{Haji13}, a state-dependent Z-interference broadcast channel was studied, in which one transmitter has only one message for its corresponding receiver, and the other transmitter has two messages respectively for two receivers. Both receivers are corrupted by the same state, which is known to both transmitters.

In \cite{Somekh08} and \cite{Kazemi13ISIT}, a model of the cognitive state-dependent IC  was also studied, in which both transmitters (i.e., the primary and cognitive transmitters) jointly send one message to receiver 1, and the cognitive transmitter sends an additional message separately to receiver 2. The state is noncausally known at the cognitive transmitter only. In \cite{Duan12ISIT,Duan15IT}, two state-dependent cognitive IC models were studied, where one transmitter knows both messages, and the two receivers are corrupted by two states which are know to both the two transmitters.


In all the previous work on the state-dependent IC and Z-IC, the states at two receivers are either assumed to be independent, or to be the same but differently scaled, with the exception of \cite{Fehri2015Z} that allows correlation between states. However, \cite{Fehri2015Z} assumes that each transmitter knows only one state at its corresponding receiver, and hence two transmitters cannot cooperate to cancel the states. In this dissertation, we investigate the state-dependent IC and Z-IC with the two receivers being corrupted respectively by two correlated states and with both transmitters knowing both states in order for them to cooperate. The state sequences are assumed to be known at both transmitters. The main focus of this dissertation is on the Gaussian state-dependent IC and Z-IC, where the receivers are corrupted by additive interference, state, and noise. The aim is to design encoding and decoding schemes to handle interference as well as to cancel the state at the receivers. In particular, we are interested in answering the following two fundamental questions: (1) whether or under what conditions both states can be simultaneously fully canceled so that the capacity for the IC and Z-IC without state can be achieved; and (2) what is the impact that the correlation between two states make towards state cancellation and capacity achievability.

A common nature that the above models share is that for each message to be transmitted, at least one transmitter in the system knows both the message and the state, and can incorporate the state information in encoding of the message so that state interference at the corresponding receiver can be canceled However, in practice, it is often the case that transmitters that have messages intended for receivers do not know the state, whereas some third-party nodes know the state, but do not know the message. In such a mismatched case, a helper user can assist all the interfered users to cancel state, though state information cannot be exploited in encoding of messages. A number of previously studied models capture such mismatched property. In \cite{Mallik08}, the point-to-point channel with a helper was studied, in which a transmitter sends a message to a state-dependent receiver, and a helper knows the state noncausally and can help the transmission. Lattice coding was designed in \cite{Mallik08} for the helper to assist state cancellation at the receiver, and was shown to be optimal under certain channel conditions.  A number of more general models were then further studied, which include the point-to-point channel as a special case. More specifically, in \cite{Laneman08,Zaidi09}, the state-dependent MAC was studied, which can be viewed as the point-to-point model with the helper also having its own message to the receiver. Two more general state-dependent MACs were studied in \cite{Somekh08MAC} and \cite{Phil11}, which can be viewed as the MAC model in \cite{Laneman08,Zaidi09} respectively with the helper further knowing the transmitter's message and with one more state corruption known at the transmitter. In \cite{Duan14TIT}, the state-dependent Z-interference channel was studied, which can be viewed as the point-to-point model with the helper also having a message to its own receiver. In \cite{Aref09,Zaidi10}, the state-dependent relay channel was studied, which can be viewed as the point-to-point model with the helper also receiving information from the transmitter and serving as a relay. When these models reduce to the point-to-point model here, the results in \cite{Zaidi09,Phil11,Zaidi10,Duan14TIT} characterize the capacity of the Gaussian channel as the state power goes to infinity as in \cite{Mallik08}. In particular, the achievable scheme in \cite{Phil11} is based on lattice coding similar to \cite{Mallik08}, and the scheme in \cite{Zaidi09,Zaidi10,Duan14TIT} can be viewed as single-bin dirty paper coding (i.e., a special case of dirty paper coding \cite{Gelf80,Costa83} with only one bin). The channel capacity of channels with helper remained unknown until, in \cite{Yunhao16IT}, a new achievable scheme was introduced. under which the channel capacity for point-to-point channel with helper is characterized when the power of channel state is finite and the power of helper is small. 

In this dissertation, we are interested in the state-dependent MAC with a helper. Various state-dependent MAC models were studied previously, which are related but different from the MAC model with a helper studied in this dissertation. State-dependent MAC models with state causally or strictly causally known at the transmitter were studied in \cite{Kim05,Lapidoth13_a,Zaidi13,Lapidoth13_b,Li13}, whereas our model assumes that the state is noncausally known at the helper. The two-user MAC with state noncausally known at the transmitters has been previously studied in various cases. \cite{Kim04,Gelf84} studied the MAC model with state noncausally known at both transmitters, while \cite{Laneman08,Zaidi09} assumed that the state is known only to one transmitter. \cite{Somekh08MAC} studied the cognitive MAC model in which one transmitter also knows the other transmitter's message in addition to the noncausal state information. Furthermore, \cite{Phil08,Phil11} studied the model with the receiver being corrupted by two independent states and each state is known noncausally to one transmitter. In all these two-user MAC models with noncausal state information, at least one transmitter knows the state information, and can hence encode messages by incorporating the state information. Our MAC model is different in that only an additional helper knows the state information and assists to cancel the state. Our goal is to characterize the capacity region either fully or partially for such a model.


\section{Contributions and Organization of Dissertation}
The rest of the dissertation is organized as follows. In Chapter \ref{cha:helper}, we study the two-user state-dependent MAC with a helper. In this model, transmitters 1 and 2 respectively send two messages to one receiver, which is corrupted by an independent and identically distributed (i.i.d.) state sequence. The state sequence is known to neither the transmitters nor the receiver, but is known to a helper noncausally, which thus assists state interference cancellation at the receiver. Our focus is on the Gaussian channel with additive state. An outer bound on the capacity region is first derived, and an inner bound is then obtained based on a scheme that integrates direct state cancellation and single-bin dirty paper coding. By comparing the inner and outer bounds, the channel parameters are partitioned into appropriate cases, and for each case, either segments on the capacity region boundary or the full capacity region are characterized.

 In Chapter \ref{chap:z2state} and Chapter \ref{cha:2state}, we investigate the Gaussian state-dependent IC and Z-IC, in which two receivers are corrupted respectively by two different but correlated states that are noncausally known to two transmitters and but are unknown to the receivers. Three interference regimes are studied, and the capacity region or sum capacity boundary is characterized either fully or partially under various channel parameters. For the very strong regime, the capacity region is achieved by a scheme where the two transmitters implement a cooperative dirty paper coding. For the strong but not very strong regime, the sum-rate capacity is characterized by rate splitting, layered dirty paper coding and successive cancellation. For the weak regime, the sum-rate capacity is achieved via dirty paper coding individually at two receivers as well as treating interference as noise. Furthermore, the impact of the correlation between states on cancellation of state and interference as well as achievability of capacity is explored with numerical illustrations. 

This dissertation leads to the following two conference publications \cite{Yunhao16ISIT,Yunhao17ISIT}, one journal paper publication \cite{Yunhao16IT} and one journal paper to be submitted \cite{Yunhao17IT}.



\chapter{Helper-Assisted State Cancellation for Multiple Access Channels}\label{cha:helper}
In this chapter, we study the state-dependent MAC channel with a helper, where two transmitters wish to send the messages to a receiver over a state-corrupted channel, and a helper knows the state information noncausally and wishes to assist the receiver to cancel the state interference.

The rest of this chapter is organized as follows. We first describe the channel model. Then, we provide lower and upper bounds on the capacity. By analyzing the lower bounds to compare them with our upper bounds, we characterize the capacity region either fully or partially in various cases.

\section{Channel Model}\label{sec:macmodel}

\vspace{5mm}
\begin{figure}[H]
	\centering
	\includegraphics[width=5in]{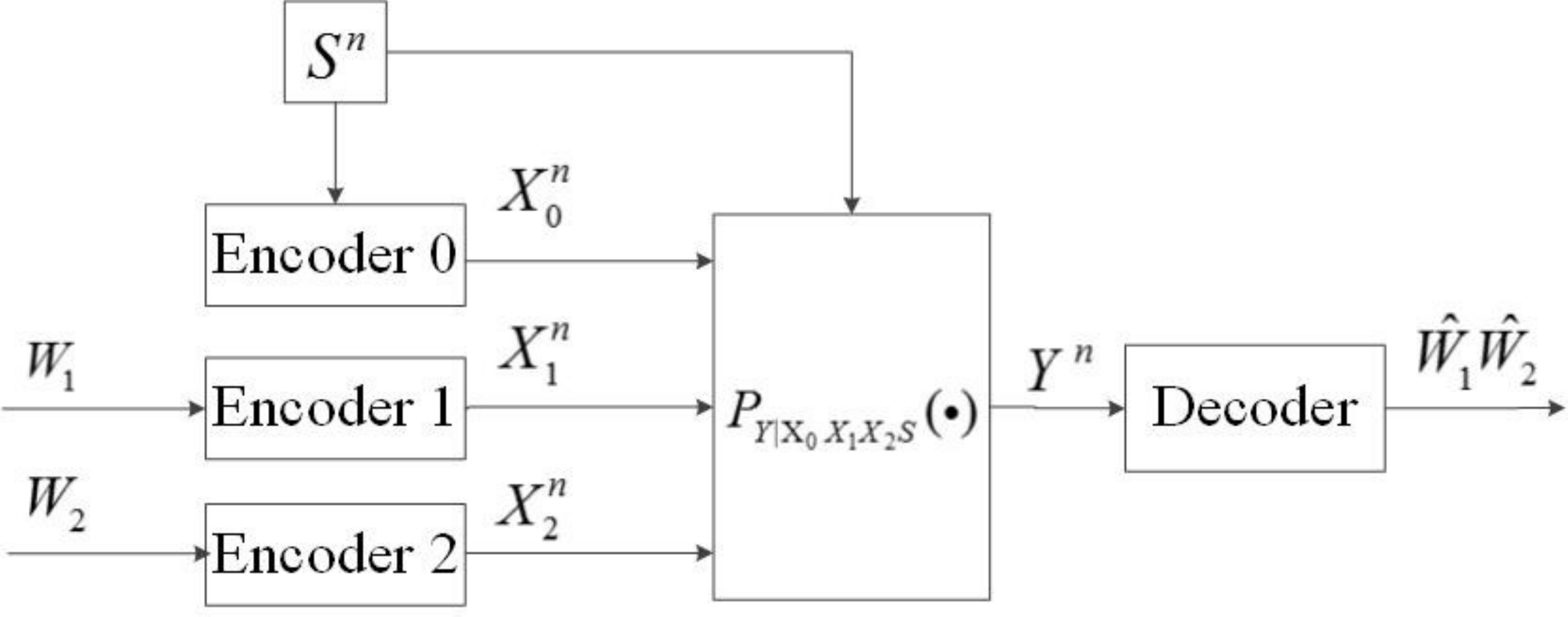}
	\caption{The state-dependent MAC with a helper} \label{fig:channelmodel_mac}
\end{figure}
\vspace{5mm}

We consider the state-dependent MAC with a helper (as shown in Fig.~\ref{fig:channelmodel_mac}), in which transmitter 1 sends a message $W_{1}$, and transmitter 2 sends a message $W_{2}$ to the receiver. The encoder $f_k: \cW \rightarrow \cX_k^n$ at transmitter $k$ maps a message $w_k\in \cW_k$ to a codeword $x_k^n\in \cX_k^n$ for $k=1,2$. The two inputs $x_1^n$ and $x_2^n$ are transmitted over the MAC to a receiver, which is corrupted by an i.i.d.\ state sequence $S^n$. The state sequence is known to neither the transmitters nor the receiver, but is known to a helper noncausally. Hence, the helper assists the receiver to cancel the state interference. The encoder $f_0:\cS^n \rightarrow \cX_0^n$ at the helper maps the state sequence $s^n\in \cS^n$ into a codeword $x_0^n\in\cX_0^n$. The channel transition probability is given by $P_{Y| X_0 X_{1} X_{2} S}$. The decoder $g:\cY^n\rightarrow (\cW_1,\cW_2)$ at the receiver maps the received sequence $y^n$ into two messages $\hw_k\in \cW_k$ for $k=1,2$.


The average probability of error for a length-$n$ code is defined as
\begin{flalign}\label{PE}
P_e^{(n)} = & \frac{1}{|\cW_1||\cW_2|}\sum_{w_1=1}^{|\cW_1|}\sum_{w_2=1}^{|\cW_2|} Pr\lbrace(\hat{w}_1, \hat{w}_2) \neq (w_1, w_2)\rbrace.
\end{flalign}
A rate pair $(R_1, R_2)$ is {\em achievable} if there exist a sequence of message sets $\cW_{k}^{(n)}$ with $|\cW_{k}^{(n)}|=2^{nR_k}$ for $k=1, 2$, and encoder-decoder tuples $(f_0^{(n)},f_1^{(n)},f_2^{(n)},g^{(n)})$ such that the average error probability $P_e^{(n)} \rightarrow 0$ as $n \to \infty$. We define the {\em capacity region} to be the closure of the set of all achievable rate pairs $(R_1, R_2)$.

We focus on the state-dependent Gaussian channel with the output at the receiver for one channel use given by
\begin{flalign}
Y&=X_0 + X_1+ X_2+S+N\label{eq:GeneralChannelModel}
\end{flalign}
where the noise variables $N \sim \mathcal{N}(0,1)$, and $S \sim \mathcal{N}(0,Q)$. Both the noise variables and the state variable are i.i.d. over channel uses. The channel inputs $X_0$, $X_1$ and $X_2$ are subject to the average power constraints $P_0$, $P_1$ and $P_2$.

Our goal is to characterize the capacity region of the Gaussian channel under various channel parameters $(P_0,P_1,P_2,Q)$.


\section{Outer and Inner Bounds}\label{sec:macbounds}

We first provide an outer bound on the capacity region as follows, in which the first terms in the "min" improve the corresponding bounds give in \cite{Duan14ISITA}.
\begin{proposition}\label{pps:Gaussian outer}
	An outer bound on the capacity region of the state-dependent Gaussian MAC with a helper consists of rate pairs $(R_1,R_2)$ satisfying:
	\begin{subequations}
		\begin{flalign}
		R_1\leqslant \min\Big\{& \frac{1}{2}\log(1+\frac{P_1}{Q+2\rho_{0S}\sqrt{P_0Q}+P_0+1})+\frac{1}{2}\log(1+P_0-\rho_{0S}^2P_0),\nonumber\\ & \frac{1}{2}\log(1+P_1)\Big\}\label{eq:r1outer}\\
		R_2\leqslant \min\Big\{& \frac{1}{2}\log(1+\frac{P_2}{Q+2\rho_{0S}\sqrt{P_0Q}+P_0+1})+\frac{1}{2}\log(1+P_0-\rho_{0S}^2P_0), \nonumber \\ & \frac{1}{2}\log(1+P_2)\Big\}\label{eq:r2outer}\\
		R_1+R_2\leqslant \min\Big\{&\frac{1}{2}\log(1+\frac{P_1+P_2}{Q+2\rho_{0S}\sqrt{P_0Q}+P_0+1})+\frac{1}{2}\log(1+P_0-\rho_{0S}^2P_0), \nonumber \\ & \frac{1}{2}\log(1+P_1+P_2)\Big\}\label{eq:r12outer}
		\end{flalign}
	\end{subequations}
	for some $\rho_{0S}$ that satisfies $-1\leqslant \rho_{0S}\leqslant 1$.
\end{proposition}
\begin{proof}
	See Section \ref{apx:OuterGaussian}.
\end{proof}

The second terms in the "min" in \eqref{eq:r1outer}-\eqref{eq:r12outer} capture the capacity region of the Gaussian MAC without state. If these bounds dominate the outer bound, then it is possible to design achievable schemes to fully cancel the state. Otherwise, if the first terms in the "min" in \eqref{eq:r1outer}-\eqref{eq:r12outer} dominate the outer bound, then the state cannot be fully canceled by any scheme, and the capacity region of the state-dependent MAC is smaller than that of the MAC without state.

We next derive an achievable region for the channel based on an achievable scheme that integrates direct state cancellation and single-bin dirty paper coding. In particular, since the helper does not know the messages, dirty paper coding naturally involves only one bin. More specifically, an auxiliary random variable (represented by $U$ in Proposition \ref{pps:DMC inner}) is generated to incorporate the state information so that the receiver decodes such variable first to cancel the state and then decode the transmitters' information. Based on such an achievable scheme, we derive the following inner bound on the capacity region.
\begin{proposition}\label{pps:DMC inner}
	For the discrete memoryless state-dependent MAC with a helper, an inner bound on the capacity region consists of rate pairs $(R_1, R_2)$ satisfying:
	\begin{subequations}
		\begin{flalign}
		R_1 \le & \min\{I(X_1;Y|X_2,U),\; I(U,X_1;Y|X_2)-I(U;S)\}\label{eq:r1innerdmc} \\
		R_2 \le & \min\{I(X_2;Y|X_1,U),\; I(U,X_2;Y|X_1)-I(U;S)\} \label{eq:r2innerdmc}\\
		R_1+R_2 \le & \min\{I(X_1,X_2;Y|U),\; I(U,X_1,X_2;Y)-I(U;S)\} \label{eq:r12innerdmc}
		\end{flalign}
	\end{subequations}
	for some distribution $P_{S}P_{U|S}P_{X_0|US}P_{X_1}P_{X_2}P_{Y|SX_0X_1X_2}$.
\end{proposition}
\begin{proof}
	See Section \ref{apx:DMC inner}.
\end{proof}

Based on the above inner bound, we derive the following inner bound for the Gaussian channel.
\begin{proposition}\label{pps:Gaussian inner}
	For the state-dependent Gaussian MAC with a helper, an inner bound on the capacity region consists of rate pairs $(R_1,R_2)$ satisfying:
	\begin{subequations}
		\begin{flalign}
		R_1 \leqslant & \min\{ f(\alpha,\beta,P_1), g(\alpha,\beta,P_1)\} \label{eq:r1innergau}\\
		R_2 \leqslant & \min\{ f(\alpha,\beta,P_2), g(\alpha,\beta,P_2)\} \label{eq:r2innergau}\\
		R_1+R_2 \leqslant & \min\{f(\alpha,\beta,P_1+P_2),g(\alpha,\beta,P_1+P_2)\} \label{eq:r12innergau}
		\end{flalign}
	\end{subequations}
	for some real constants $\alpha$ and $\beta$ satisfying $-\sqrt{\frac{P_0}{Q}}\leqslant \beta \leqslant \sqrt{\frac{P_0}{Q}}$. In the above bounds,
	\begin{flalign}
	f(\alpha,\beta,P)&= \frac{1}{2}\log\frac{P_0'(P_0'+(1+\beta)^2Q+P+1)}{P_0'Q(\alpha-1-\beta)^2+P_0'+\alpha^2Q}, \label{eq:Inner1-1mac}\\
	g(\alpha,\beta,P)&= \frac{1}{2}\log\left(1+\frac{P(P_0'+\alpha^2Q)}{P_0'Q(\alpha-1-\beta)^2+P_0'+\alpha^2Q} \right), \label{eq:Inner1-2mac}
	\end{flalign}
	where $P_0'= P_0-\beta^2Q$.
\end{proposition}
\begin{proof}
	The region follows from Proposition \ref{pps:DMC inner} by choosing the joint Gaussian distribution for random variables as follows:
	\begin{flalign*}
	&U=X_{0}'+\alpha S, \\
	&X_0=X_{0}'+\beta S, \\
	&X_{0}'\sim \mathcal{N}(0,P_{0}'), \\
	&X_1 \sim \mathcal{N}(0,P_1),\quad \\
	&X_2 \sim \mathcal{N}(0,P_2)\nn
	\end{flalign*}
	where $X_{0}',X_1,X_2,S$ are independent. The constraint on $\beta$ follows due to the power constraint on $X_0$.
\end{proof}
We note that the above construction of the input $X_0$ of the helper reflects two state cancelation schemes: the term $\beta S$ represents direct cancelation of some state power in the output of the receiver; and the variable $X_{0}'$ is used for dirty paper coding via generation of the state-correlated auxiliary variable $U$. Hence, the parameter $\beta$ controls the balance of two schemes in the integrated scheme, and can be optimized to achieve the best performance. This scheme is also equivalent to the one with $U=X_0+\alpha S$, where $X_0$ and $S$ are correlated. While such approaches have been considered in the literature (see e.g., \cite{Laneman08}), we believe that  selecting $U$ and $X_0$ successively provides a more operational meaning to the correlation structure.


\section{Capacity Characterization}\label{sec:maccapacity}

By comparing the inner and outer bounds provided in Section \ref{sec:macbounds}, we characterize the capacity region or segments on the capacity boundary in various channel cases. Our idea is to separately analyze the bounds \eqref{eq:r1innergau}-\eqref{eq:r12innergau} in the inner bound and characterize conditions on the channel parameters $(P_0,P_1,P_2,Q)$ under which these bounds respectively meet the bounds \eqref{eq:r1outer}-\eqref{eq:r12outer} in the outer bound. In such cases, the corresponding segment on the capacity region is characterized.


We first consider the bound on $R_1$ in \eqref{eq:r1innergau}. Let
\begin{flalign}\label{eq:alphabeta}
&\alpha_1=\frac{(1+\beta_1)P_0'}{P_0'+1},\\
&\beta_1=\rho_{0S}^*\sqrt{\frac{P_0}{Q}}.
\end{flalign}
Then $f(\alpha,\beta,P_1)$ takes the following form
\begin{flalign}
f(\alpha_1,\beta_1,P_1)&= \frac{1}{2}\log\left(1+\frac{P_1}{Q+2\rho_{0S}^*\sqrt{P_0Q}+P_0+1}\right)\nn\\
&+ \frac{1}{2}\log(1+P_0-\rho_{0S}^{*2}P_0)\label{eq:Inner1-1-1mac}
\end{flalign}
where $\rho_{0S}^*\in [-1,1]$ maximizes
\[\frac{1}{2}\log\left(1+\frac{P_1}{Q+2\rho_{0S}\sqrt{P_0Q}+P_0+1}\right)+ \frac{1}{2}\log(1+P_0-\rho_{0S}^2P_0).\] In fact, $\alpha_1$ is set to maximize $f(\alpha,\beta,P_1)$ for fixed $\beta$, and $\beta_1$ is set to maximize the function with $\alpha$ being replaced by $\alpha_1$.
If $f(\alpha_1,\beta_1,P_1)\leqslant g(\alpha_1,\beta_1,P_1)$, then $R_1=f(\alpha_1,\beta_1,P_1)$ is achievable, and this meets the outer bound (the first term in "min" in \eqref{eq:r1outer}). Thus, one segment of the capacity region is specified by $R_1=f(\alpha_1,\beta_1,P_1)$.

Furthermore, we set $\beta=\alpha-1$ and then obtain:
\begin{flalign}
g(\alpha,\alpha-1,P_1)&=\frac{1}{2}\log(1+P_1). \label{eq:Inner1-2-2mac}
\end{flalign}
If $g(\alpha,\alpha-1,P_1)\leqslant f(\alpha,\alpha-1,P_1)$, i.e., $P_0'^2\ge \alpha^2Q(P_1+1-P_0')$  where $P_0'=P_0-(\alpha-1)^2Q$ holds for some $\alpha \in \Omega_{\alpha}=\{\alpha:1-\sqrt{\frac{P_0}{Q}}\leq \alpha \leq 1+\sqrt{\frac{P_0}{Q}}\}$, then $R_1=\frac{1}{2}\log(1+P_1)$ is achievable, and this meets the outer bound (the second term in "min" in \eqref{eq:r1outer}). This also equals the maximum rate for $R_1$ when the channel is not corrupted by state. Thus, one segment of the capacity region is specified by $R_1=\frac{1}{2}\log(1+P_1)$.

we demonstrate our characterization of the capacity via numerical plots.

\vspace{5mm}
\begin{figure}[thb]
	\begin{center}
		\includegraphics[width=12cm]{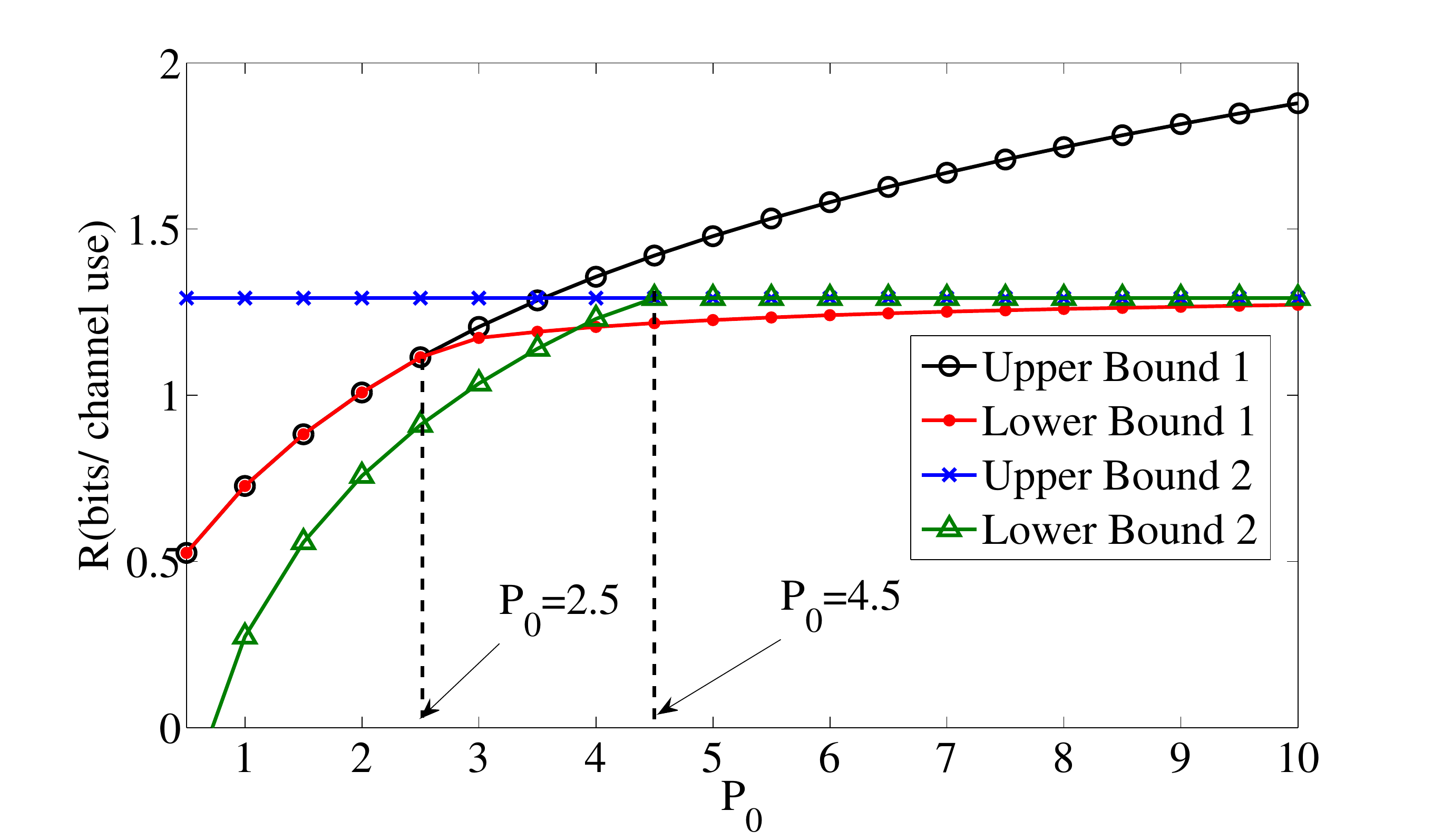}
		\vspace{5mm}
		\caption{Lower and upper bounds on the capacity for the state-dependent channel with a helper}\label{fig:Comparisons}
	\end{center}
\end{figure}
\vspace{5mm}

In Fig.~\ref{fig:Comparisons}, we fix $P=5$, and $Q=12$, and plot the lower bounds $$\min\max_{\alpha,\beta}\{f(\alpha,\beta,P_1,P_0),g(\alpha,\beta,P_1,P_0)\}$$ and the upper bounds in Proposition \ref{pps:Gaussian outer} as functions of the helper's power $P_0$. It can be seen that the lower bound $\max_{\alpha,\beta}f(\alpha,\beta,P_1,P_0)$ matches the upper bound 1 ($\frac{1}{2}\log(1+\frac{P_1+P_2}{Q+2\rho_{0S}\sqrt{P_0Q}+P_0+1})+\frac{1}{2}\log(1+P_0-\rho_{0S}^2P_0)$) when $P_0\leq 2.5$, which characterize the capacity, and the lower bound $\max_{\alpha,\beta}g(\alpha,\beta,P_1,P_0)$ matches the upper bound 2 ($\frac{1}{2}\log(1+P_1)$) when $P_0\ge 4.5 $, which corresponds to the capacity $R_1=\frac{1}{2}\log(1+P_1)$. The numerical result also suggests that when $P_0$ is small, the channel capacity is limited by the helper's power and increases as the helper's power $P_0$ increases. However, as $P_0$ becomes large enough, the channel capacity is determined only by the transmitter's power $P$, in which case the state is perfectly canceled. We further note that the channel capacity without state can even be achieved when $P_0 < Q$ (e.g., $4.5 \leq P_0 \leq 10$). This implies that for these cases, the state is fully canceled not only by state subtraction, but also by precoding the state via single-bin dirty paper coding. We finally note that a better achievable rate can be achieved by the convex envelop of the two lower bounds, which does not yield further capacity result and is not shown in Fig.~\ref{fig:Comparisons}.



\vspace{5mm}
\begin{figure}[thb]
	\centering
	\includegraphics[width=4.5in]{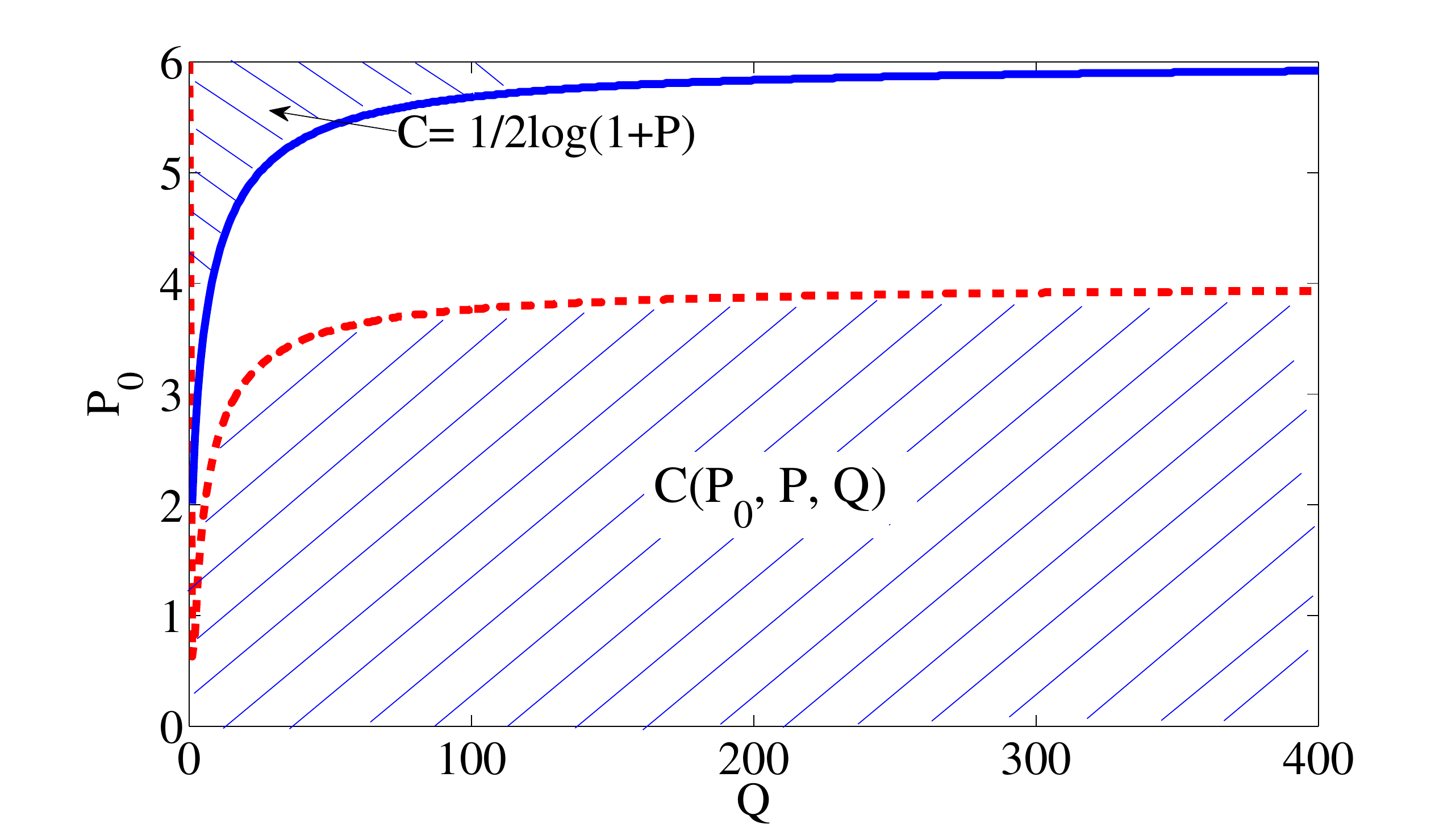}
	\caption{Ranges of parameters for which the capacity is characterized}\label{fig:ParaSets}
\end{figure}
\vspace{5mm}

In Fig.~\ref{fig:ParaSets}, we fix $P=5$, and plot the range of the channel parameters $(Q,P_0)$ for which we characterize the capacity. Each point in the figure corresponds to one parameter pair $(Q,P_0)$. The upper shaded area corresponds to channel parameters that satisfy \eqref{eq:cond2}, i.e., $P_0$ is large enough compared to $Q$, and hence the capacity of the channel without state can be achieved. The lower shaded area corresponds to channel parameters that satisfy \eqref{eq:cond1}, and hence the capacity is characterized by a function of not only $P$, but also $P_0$ and $Q$.

Similarly, following the above arguments, segments on the capacity region boundary corresponding to bounds on $R_2$ and $R_1+R_2$ can be characterized.

Summarizing the above analysis, we obtain the following characterization of segments of the capacity region boundary.

\begin{theorem}\label{th:capacity}
	The channel parameters $(P_0,P_1,P_2,Q)$ can be partitioned into the sets $\cA_1,\cB_1,\cC_1$, where
	\begin{flalign}
	& \cA_1=\{(P_0,P_1,P_2,Q): f(\alpha_1,\beta_1,P_1)\leqslant g(\alpha_1,\beta_1,P_1)\} \nonumber \\
	& \cC_1=\{(P_0,P_1,P_2,Q): P_0'^2\ge \alpha^2Q(P_1+1-P_0') \nn\\
	& \hspace{8mm} \text{ where }P_0'=P_0-(\alpha-1)^2Q,\text{ for some }\alpha \in \Omega_{\alpha}\} \nonumber \\
	& \cB_1=(\cA_1\cup \cC_1)^c. \nonumber
	\end{flalign}
	If $(P_0,P_1,P_2,Q)\in \cA_1$, then $R_1=f(\alpha_1,\beta_1,P_1)$ captures one segment of the capacity region boundary, where the state cannot be fully canceled. If $(P_0,P_1,P_2,Q)\in \cC_1$, then $R_1=\frac{1}{2}\log(1+P_1)$ captures one segment of the capacity region boundary, where the state is fully canceled. If $(P_0,P_1,P_2,Q)\in \cB_1$, $R_1$ segment of the capacity region boundary is not characterized.
	
	The channel parameters $(P_0,P_1,P_2,Q)$ can alternatively be partitioned into the sets $\cA_2,\cB_2,\cC_2$, where
	\begin{flalign}
	& \cA_2=\{(P_0,P_1,P_2,Q): f(\alpha_2,\beta_2,P_2)\leqslant g(\alpha_2,\beta_2,P_2)\} \nonumber \\
	& \cC_2=\{(P_0,P_1,P_2,Q): P_0'^2\ge \alpha^2Q(P_2+1-P_0') \nn\\
	& \hspace{8mm} \text{ where }P_0'=P_0-(\alpha-1)^2Q,\text{ for some }\alpha \in \Omega_{\alpha} \} \nonumber \\
	& \cB_2=(\cA_2\cup \cC_2)^c, \nonumber
	\end{flalign}
	where $\alpha_2,\beta_2$ are defined similarly to \eqref{eq:alphabeta} with $P_1$ being replaced by $P_2$. If $(P_0,P_1,P_2,Q)\in \cA_2$, then $R_2=f(\alpha_2,\beta_2,P_2)$ captures one segment of the capacity region boundary, where the state cannot be fully canceled. If $(P_0,P_1,P_2,Q)\in \cC_2$, then $R_2=\frac{1}{2}\log(1+P_2)$ captures one segment of the capacity region boundary, where the state is fully canceled.
	
	Furthermore, the channel parameters $(P_0,P_1,P_2,Q)$ can also be partitioned into the sets $\cA_3,\cB_3,\cC_3$, where
	\begin{flalign}
	& \cA_3=\{(P_0,P_1,P_2,Q): \nn\\
	&\quad\quad f(\alpha_3,\beta_3,P_1+P_2)\leqslant g(\alpha_3,\beta_3,P_1+P_2)\} \nonumber \\
	& \cC_3=\{(P_0,P_1,P_2,Q): P_0'^2\ge \alpha^2Q(P_1+P_2+1-P_0') \nn\\
	& \hspace{8mm} \text{ where }P_0'=P_0-(\alpha-1)^2Q,\text{ for some }\alpha \in \Omega_{\alpha} \} \nonumber \\
	& \cB_3=(\cA_3\cup \cC_3)^c, \nonumber
	\end{flalign}
	where $\alpha_3,\beta_3$ are defined similarly to \eqref{eq:alphabeta} with $P_1$ being replaced by $P_1+P_2$. If $(P_0,P_1,P_2,Q)\in \cA_3$, then $R_1+R_2=f(\alpha_3,\beta_3,P_1+P_2)$ captures one segment of the sum capacity, where the state cannot be fully canceled. If $(P_0,P_1,P_2,Q)\in \cC_3$, then $R_1+R_2=\frac{1}{2}\log(1+P_1+P_2)$ captures one segment of the sum capacity, where the state is fully canceled.
\end{theorem}

The above theorem describes three partitions of the channel parameters respectively characterizing segments on the capacity region corresponding to $R_1$, $R_2$ and $R_1+R_2$. Then intersection of three sets (with each from one partition) collectively characterizes all segments on the capacity region boundary. For example, if a given channel parameter tuple satisfies $(P_0,P_1,P_2,Q) \in (\cC_1\bigcap\cC_2\bigcap\cA_3)$, then following Theorem \ref{th:capacity}, line segments characterized by $R_1 = \frac{1}{2}\log( 1+P_1)$, $R_2 = \frac{1}{2}\log( 1+P_2)$, and $R_1+R_2=f(\alpha_3,\beta_3,P_1+P_2)$ are on the capacity region boundary. Since parameters $\alpha$ and $\beta$ that achieve these segments are not the same, the intersection of these segments are not on the capacity region boundary.

\begin{figure*}
	\begin{tabular}{cc}
		\includegraphics[width=3in]{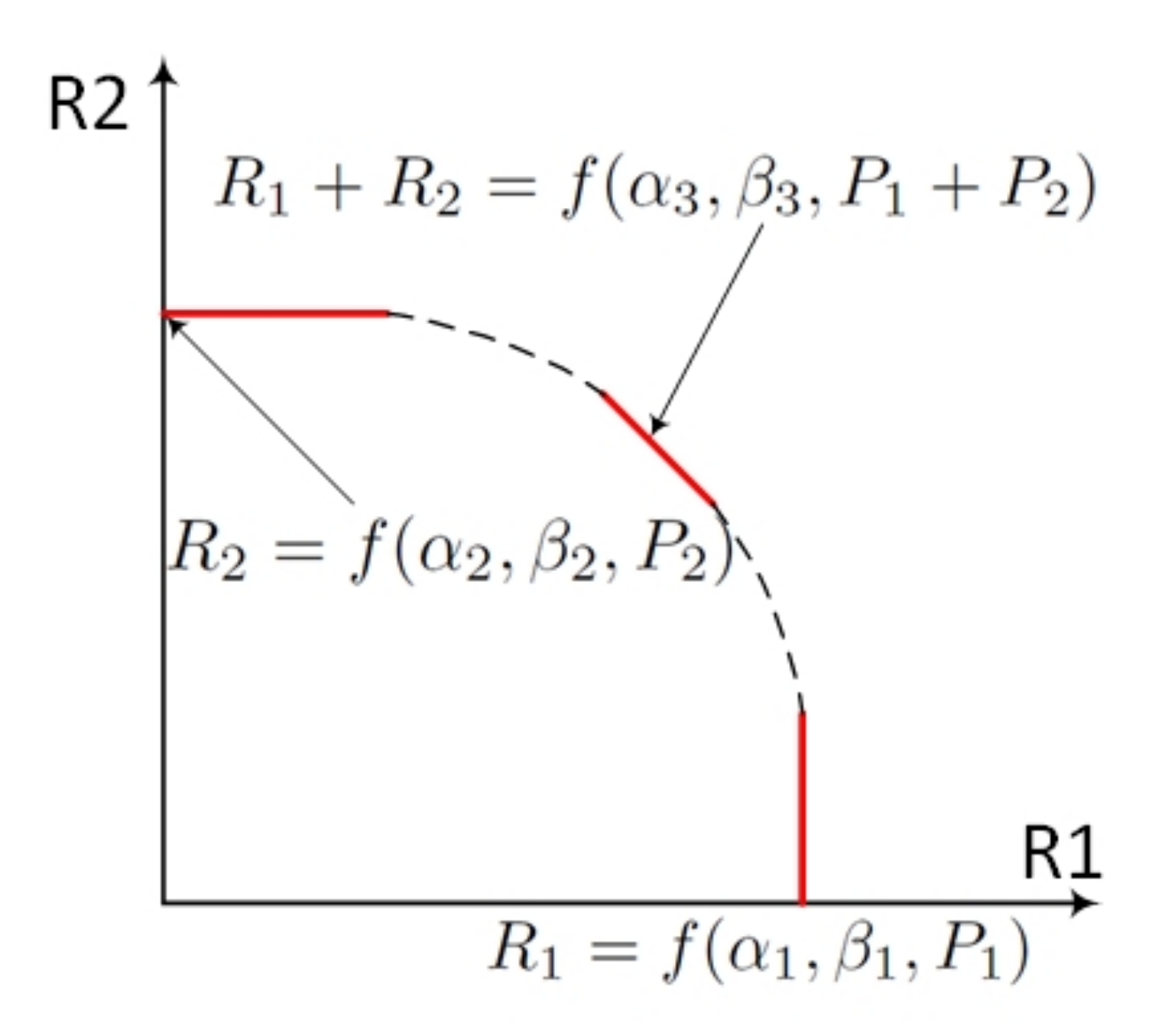}
		&\includegraphics[width=3in]{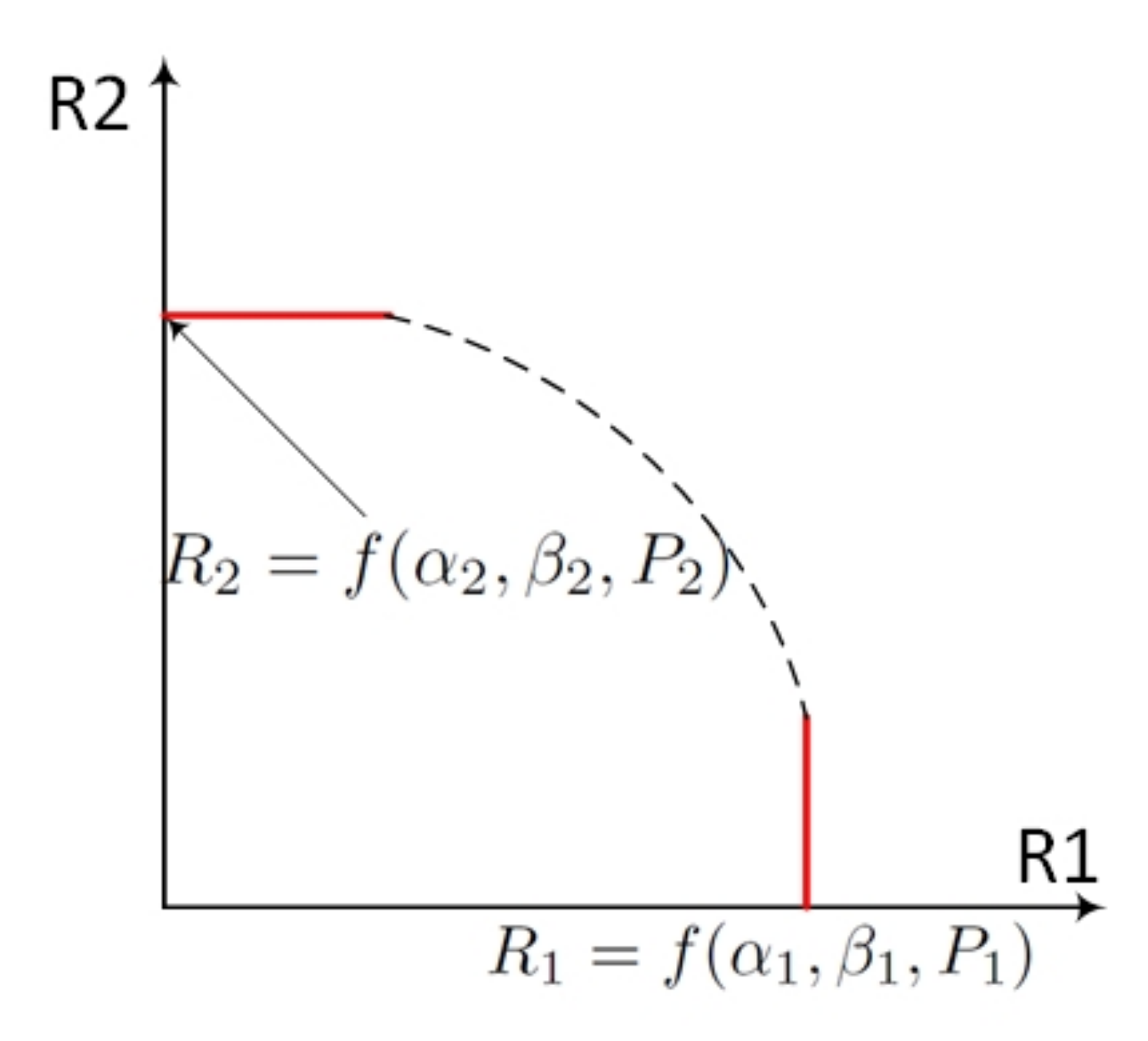}\\
		$\scriptstyle \cA_1\bigcap\cA_2\bigcap\cA_3$&$\scriptstyle\cA_1\bigcap\cA_2\bigcap\cB_3$\\
		\includegraphics[width=3in]{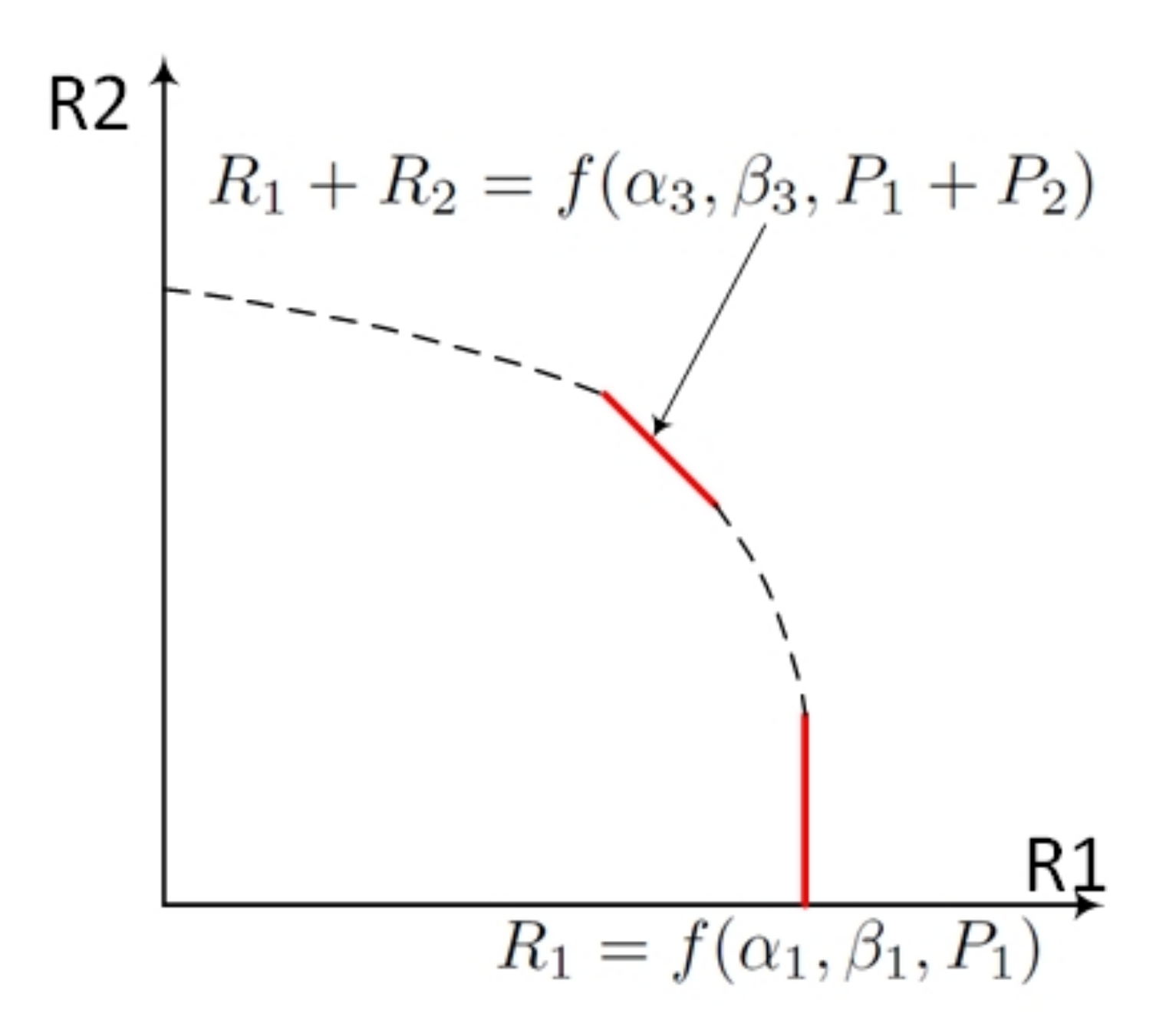}
		&\includegraphics[width=3in]{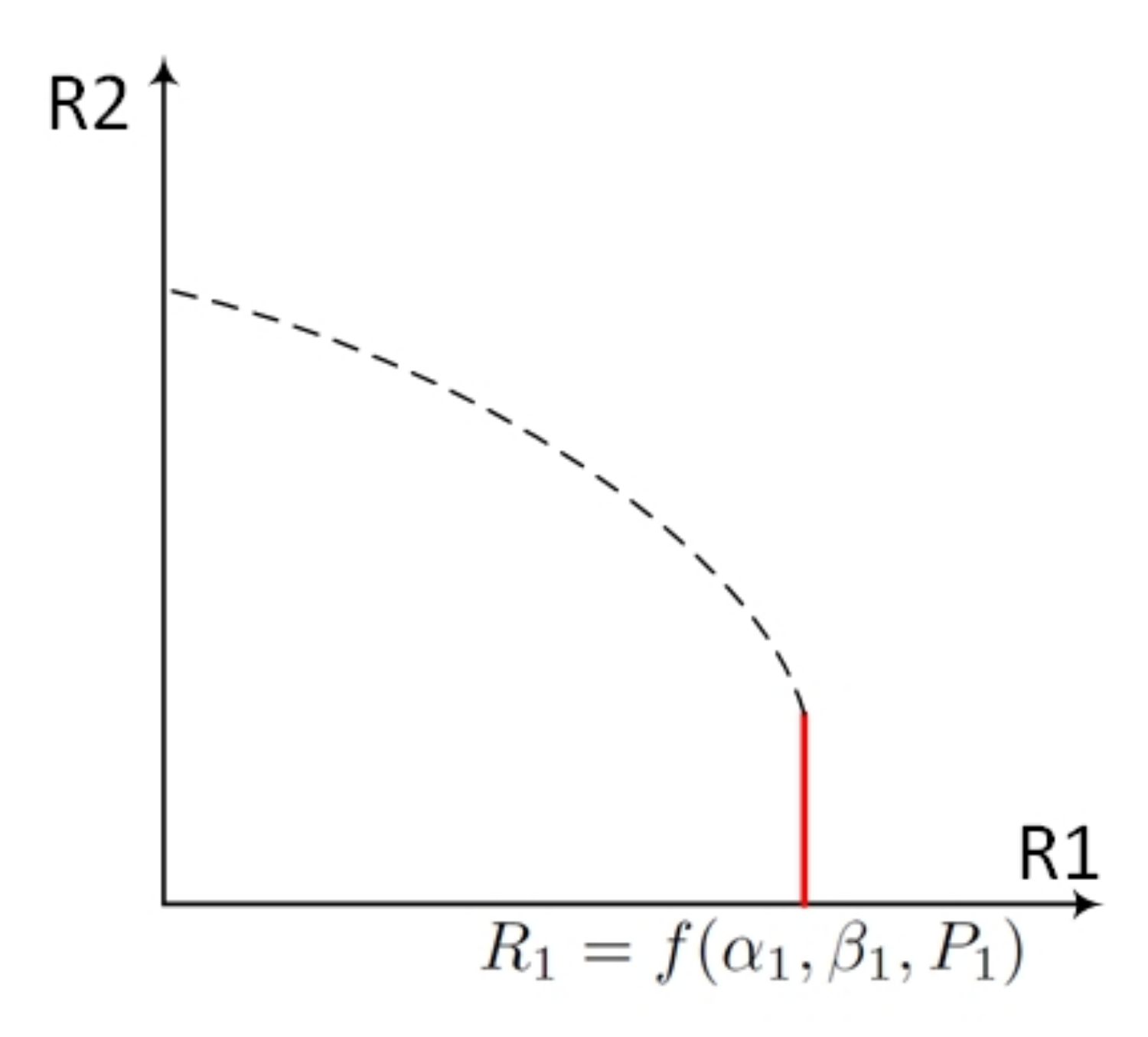}\\
		$\scriptstyle\cA_1\bigcap\cB_2\bigcap\cA_3$&$\scriptstyle\cA_1\bigcap\cB_2\bigcap\cB_3$ \\
		\includegraphics[width=3in]{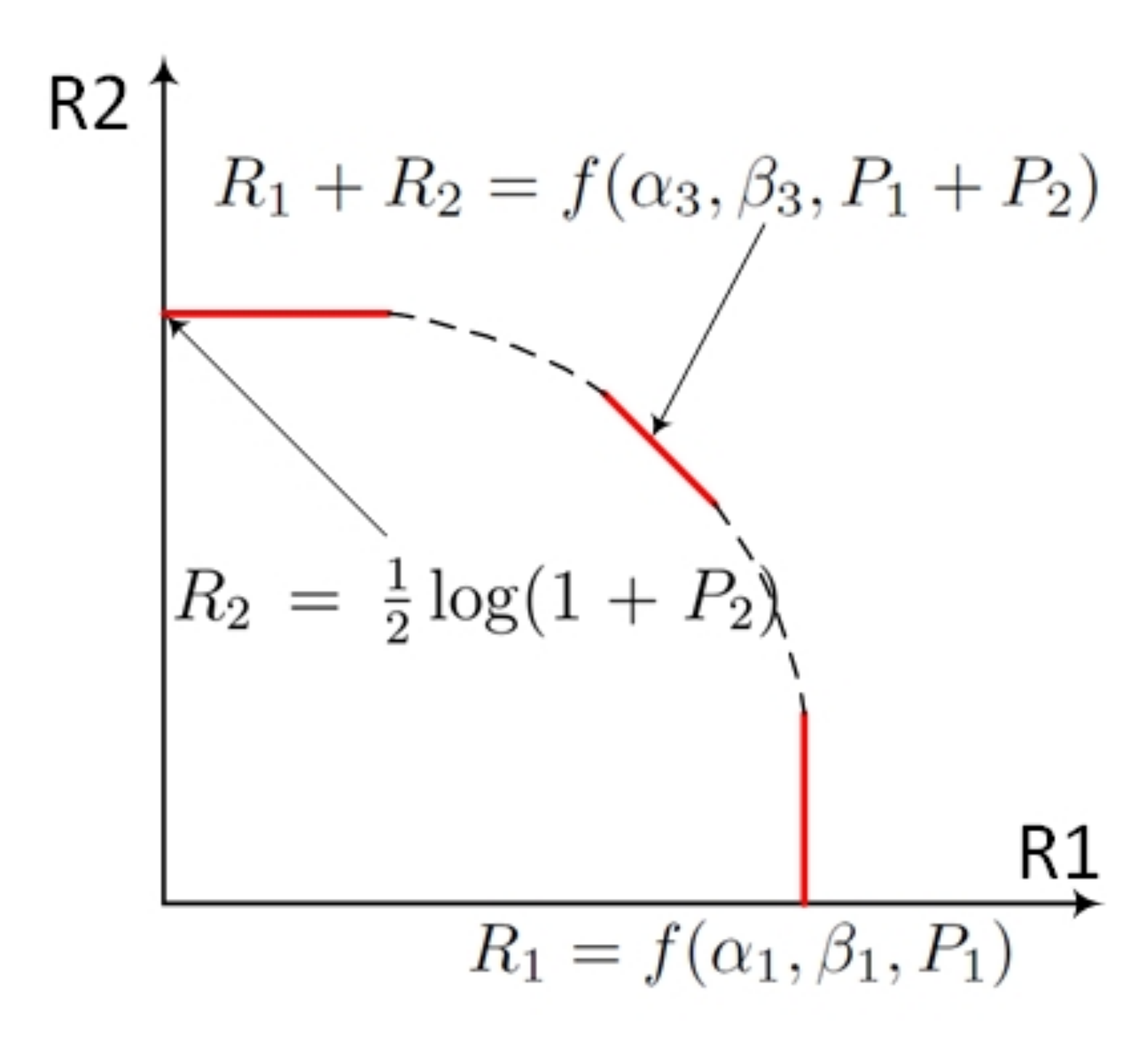}
		&\includegraphics[width=3in]{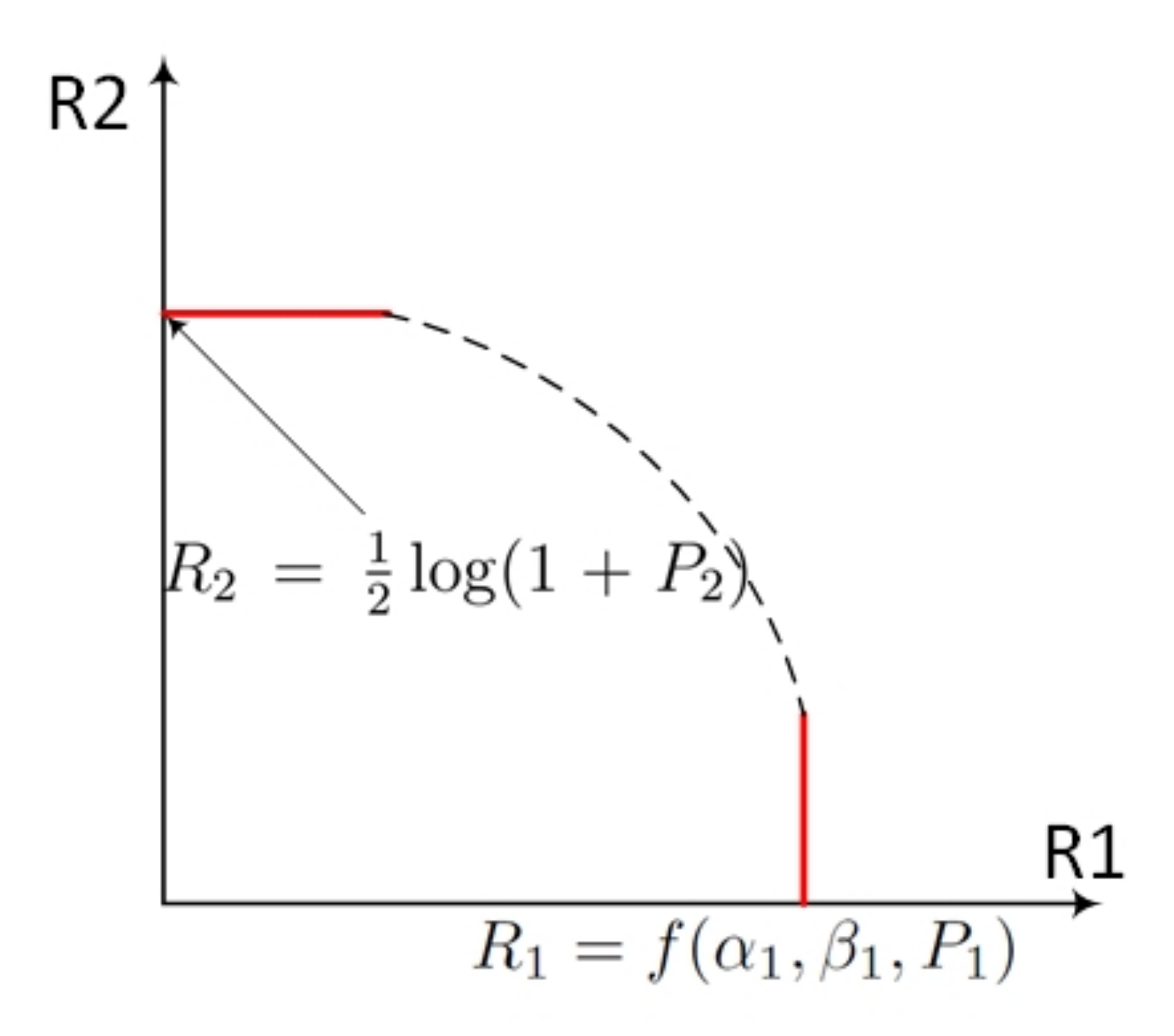}\\
		$\scriptstyle\cA_1\bigcap\cC_2\bigcap\cA_3$&$\scriptstyle\cA_1\bigcap\cC_2\bigcap\cB_3$\\
			\end{tabular}
	\end{figure*}
		\begin{figure*}
			\begin{tabular}{cc}
		\includegraphics[width=3in]{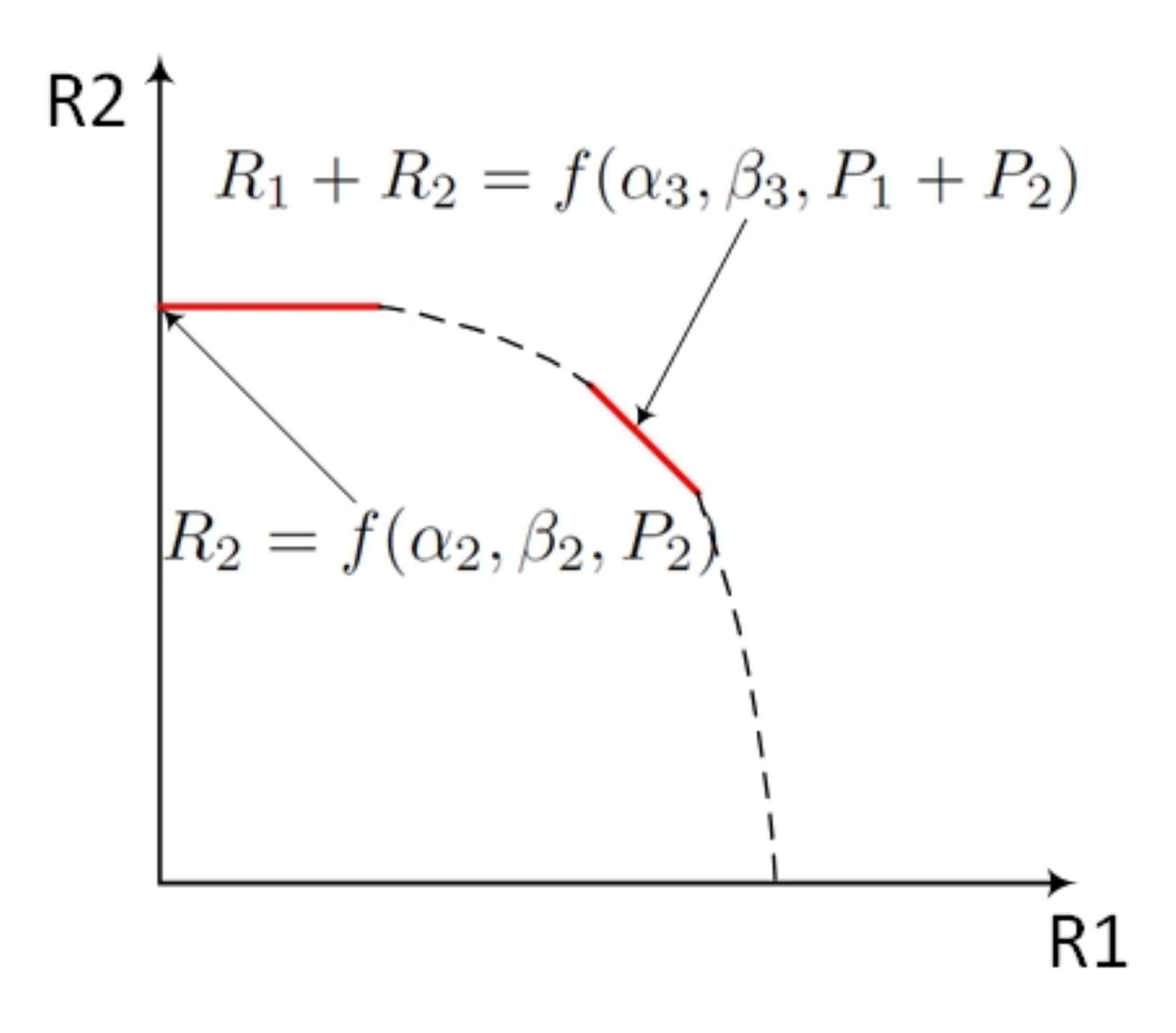}
		&\includegraphics[width=3in]{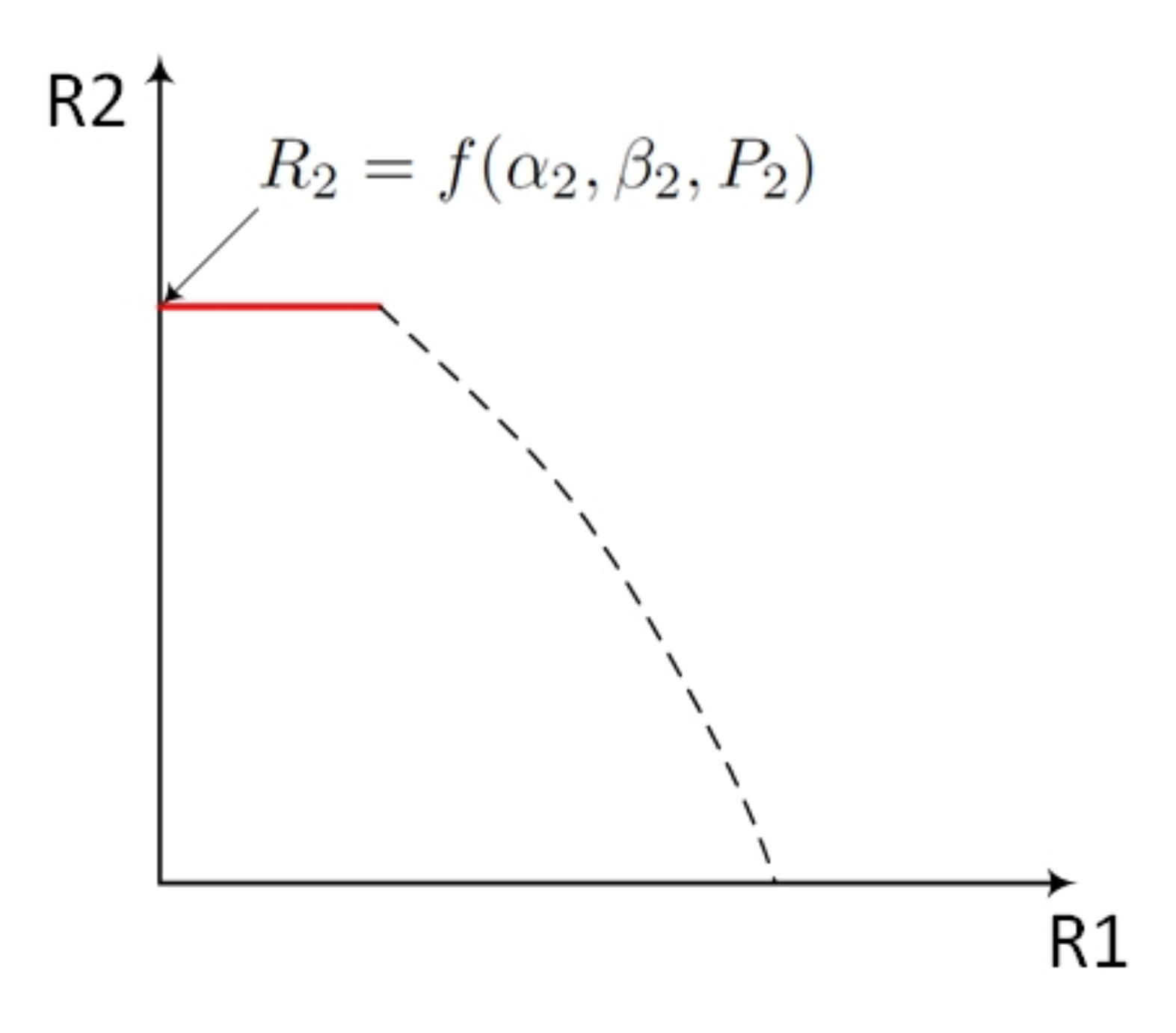} \\
		$\scriptstyle\cB_1\bigcap\cA_2\bigcap\cA_3$&$\scriptstyle\cB_1\bigcap\cA_2\bigcap\cB_3$ \\
		\includegraphics[width=3in]{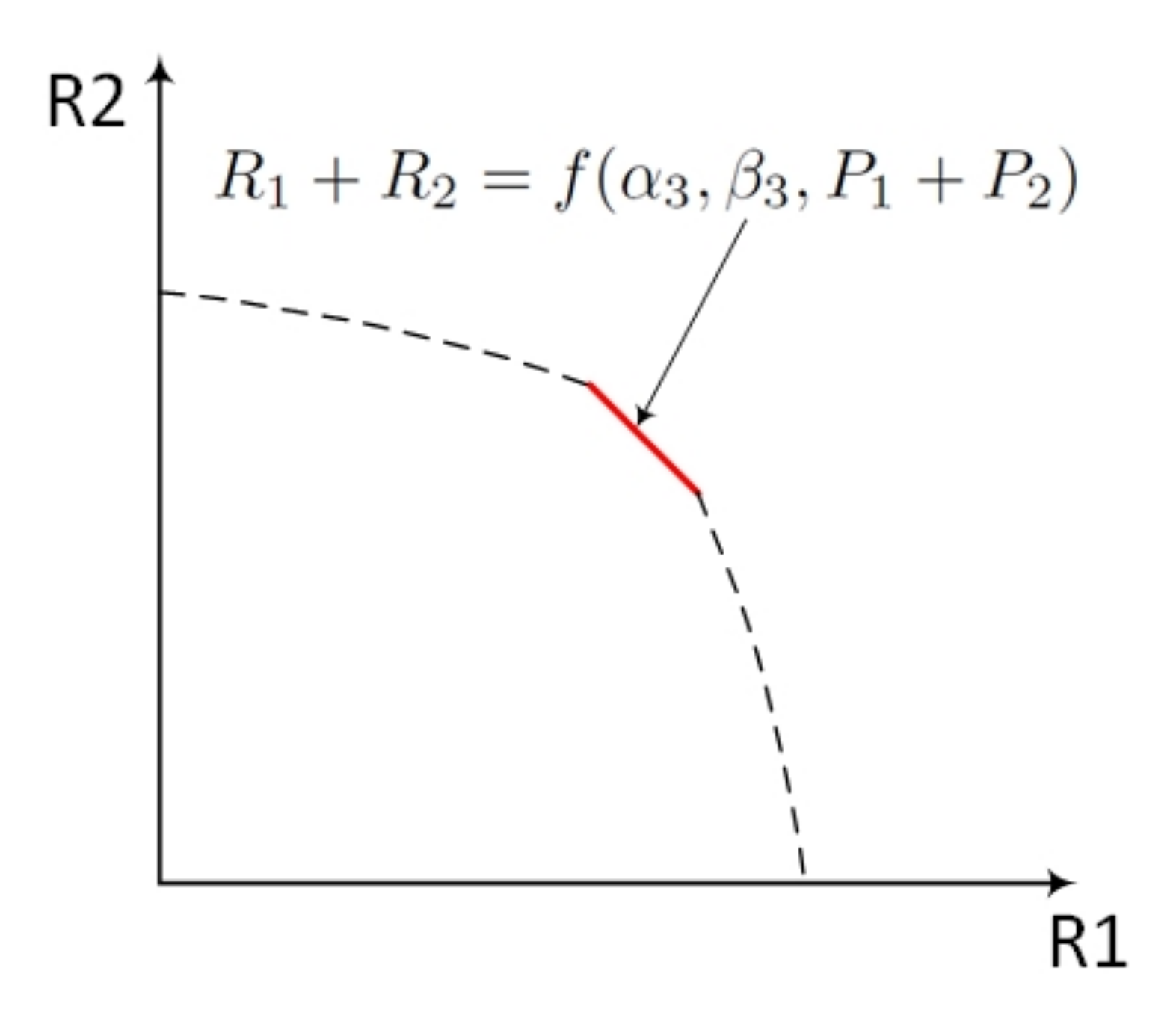}
		&\includegraphics[width=3in]{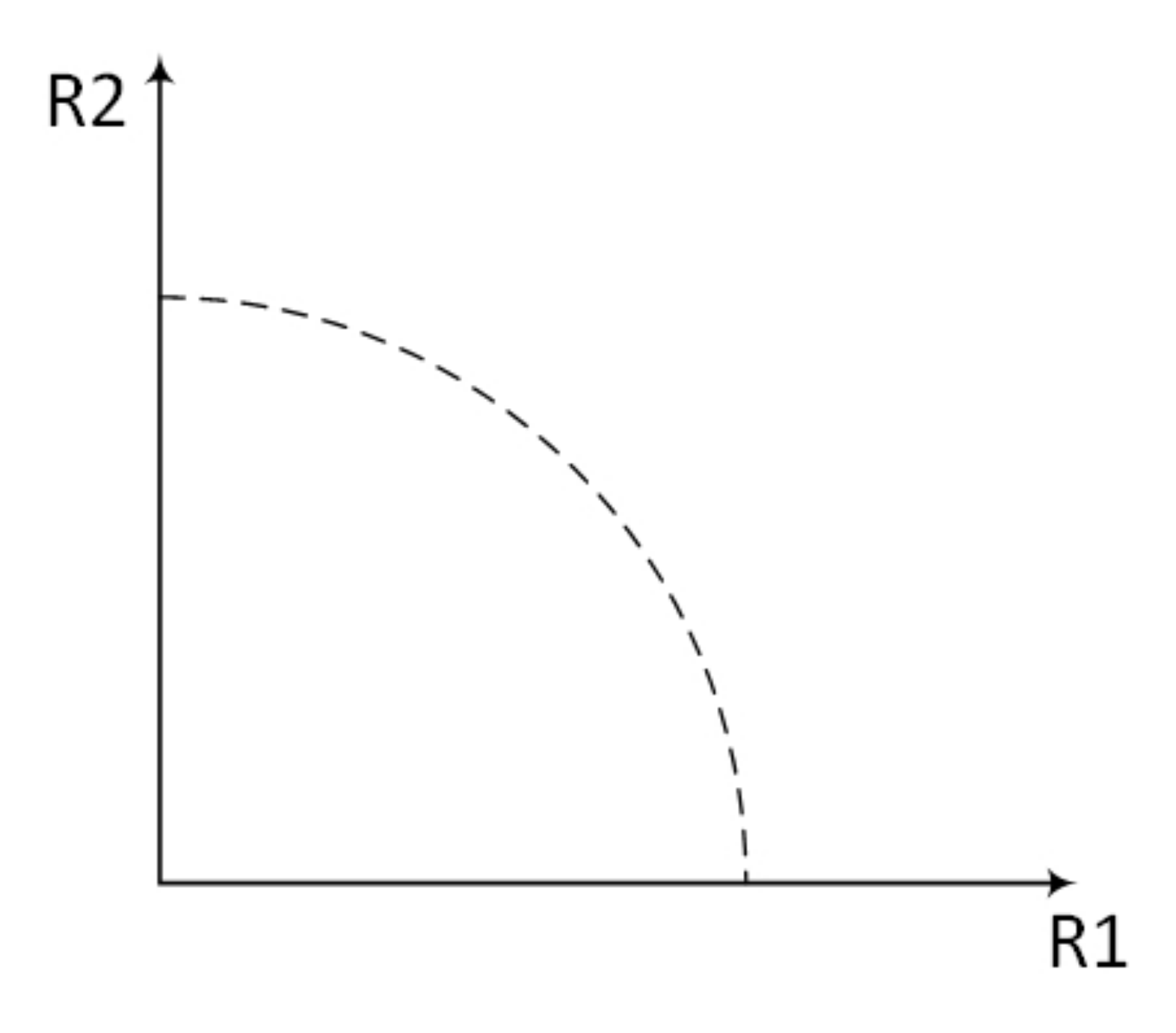}\\
		$\scriptstyle\cB_1\bigcap\cB_2\bigcap\cA_3$&$\scriptstyle\cB_1\bigcap\cB_2\bigcap\cB_3$\\
		\includegraphics[width=3in]{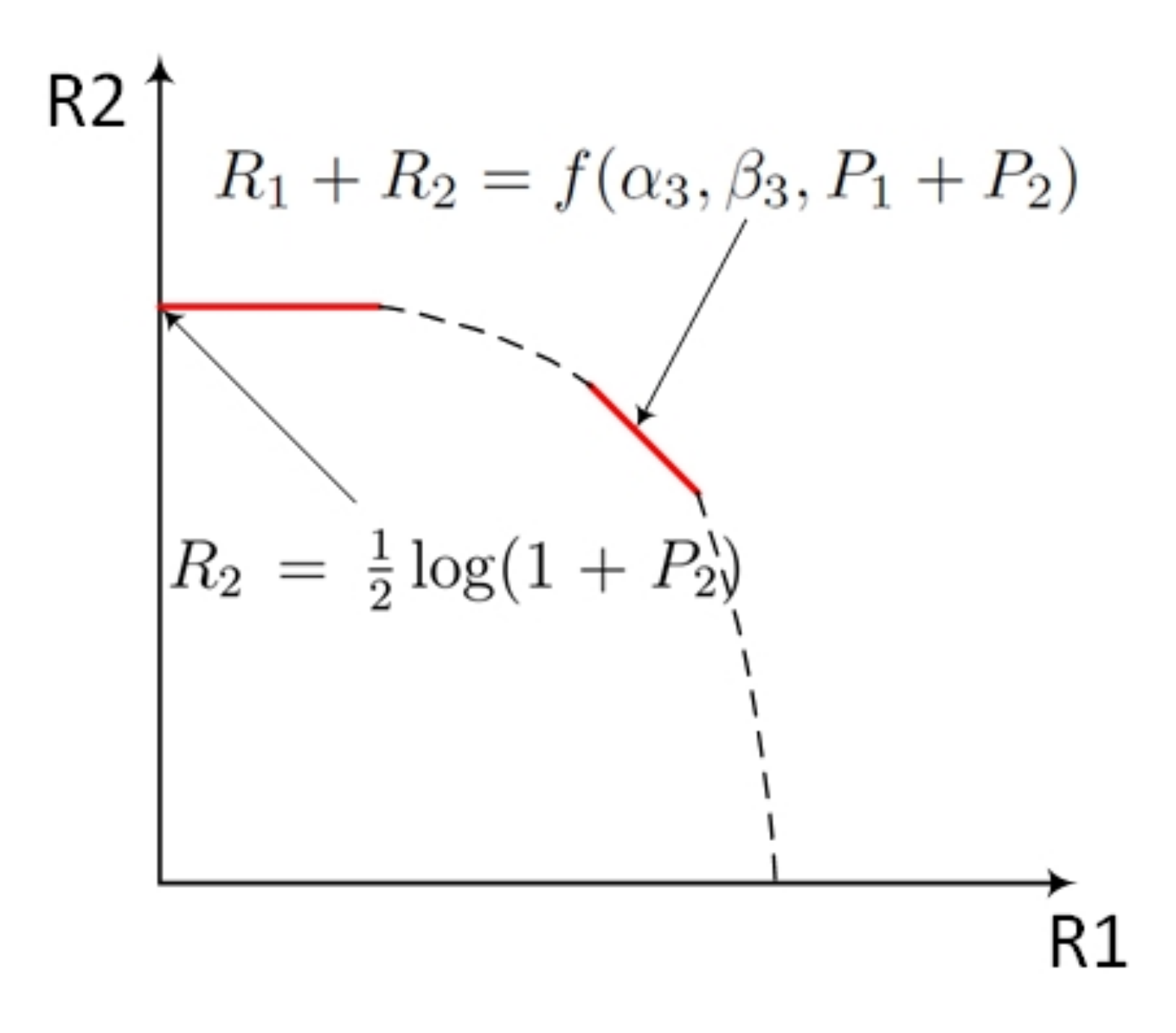}
		&\includegraphics[width=3in]{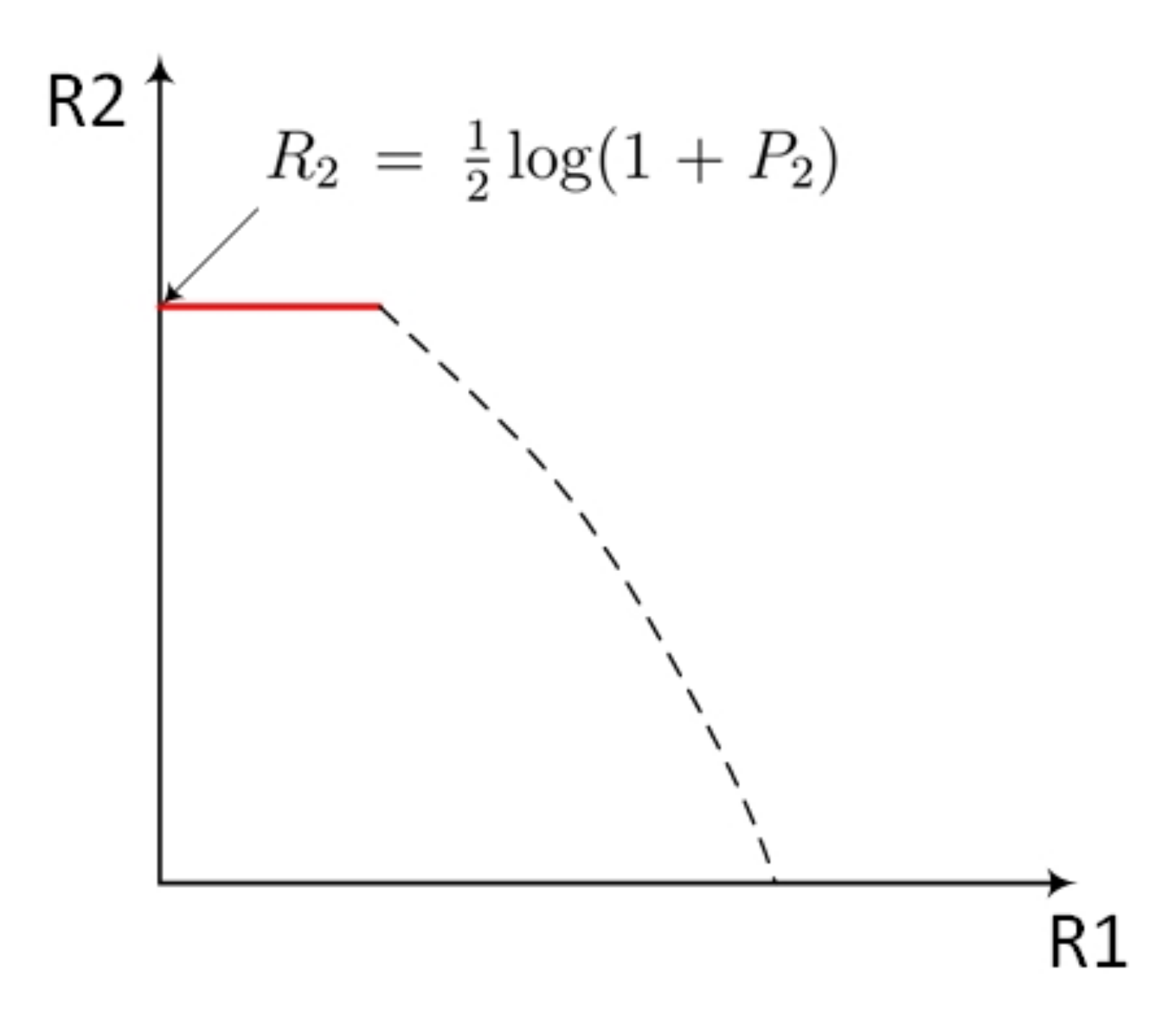}\\
		$\scriptstyle\cB_1\bigcap\cC_2\bigcap\cA_3$&$\scriptstyle\cB_1\bigcap\cC_2\bigcap\cB_3$ \\
					\end{tabular}
	\end{figure*}
\begin{figure*}
\begin{tabular}{cc}
		\includegraphics[width=3in]{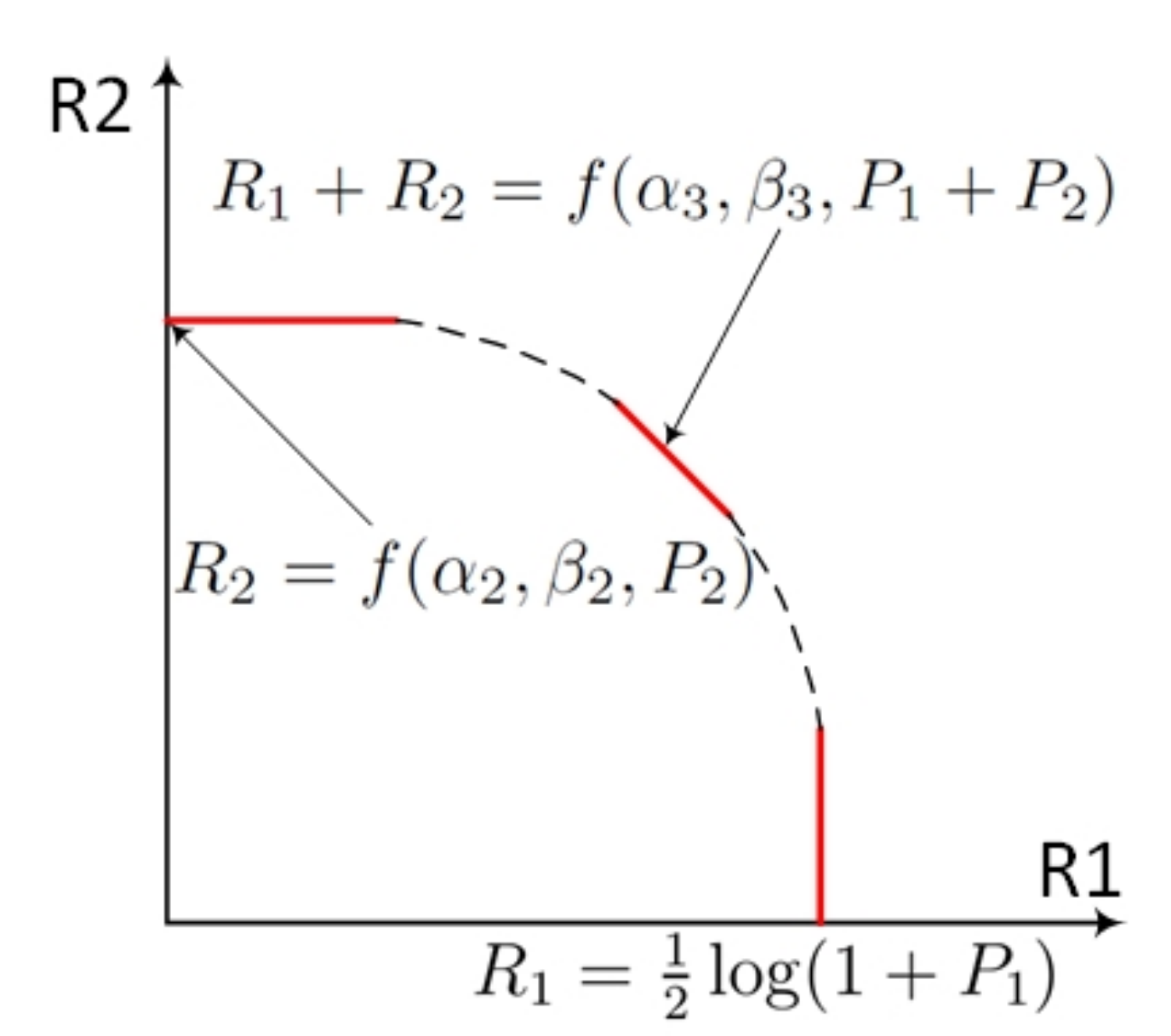}
		&\includegraphics[width=3in]{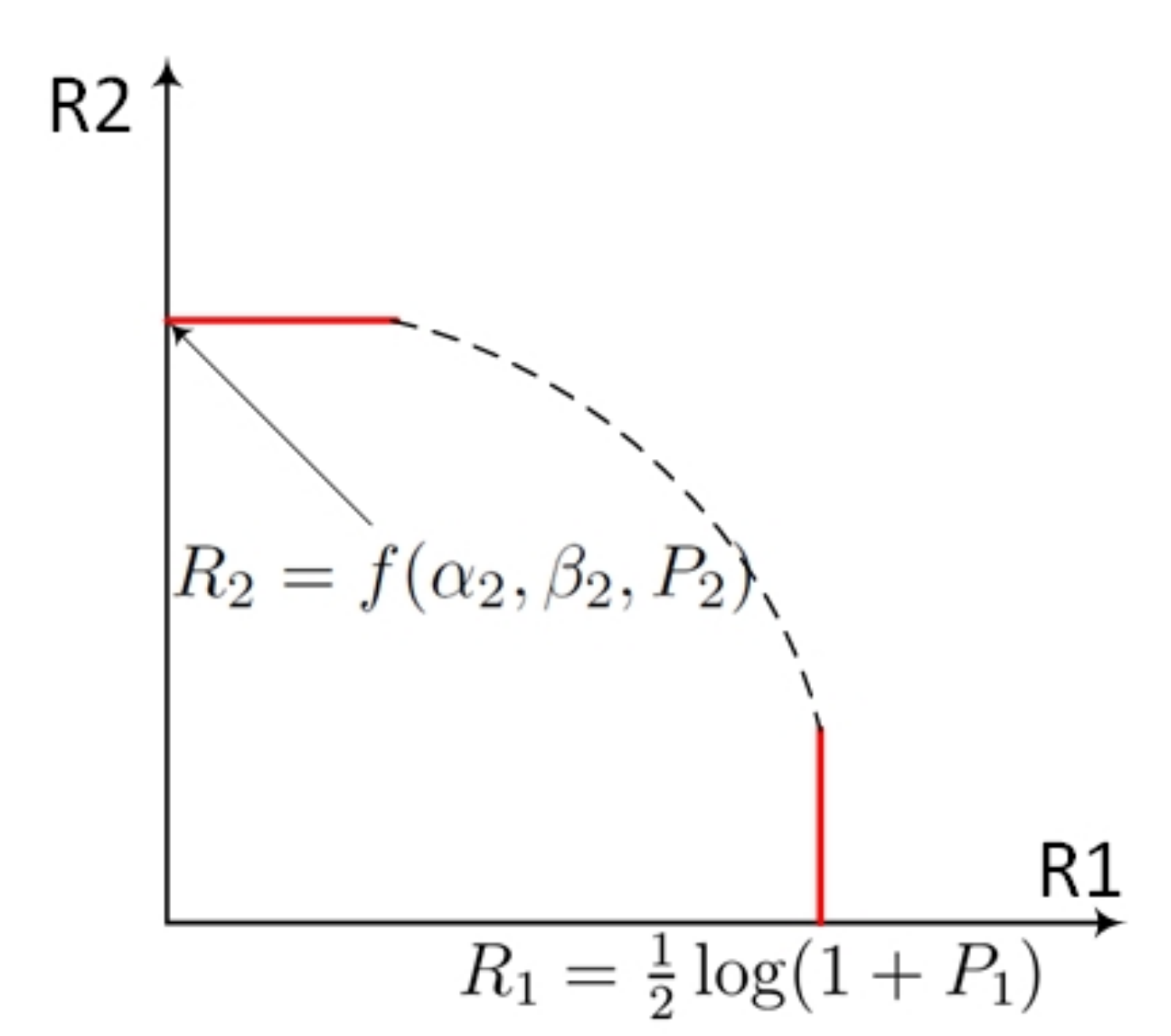}\\
		$\scriptstyle\cC_1\bigcap\cA_2\bigcap\cA_3$&$\scriptstyle\cC_1\bigcap\cA_2\bigcap\cB_3$\\
		\includegraphics[width=3in]{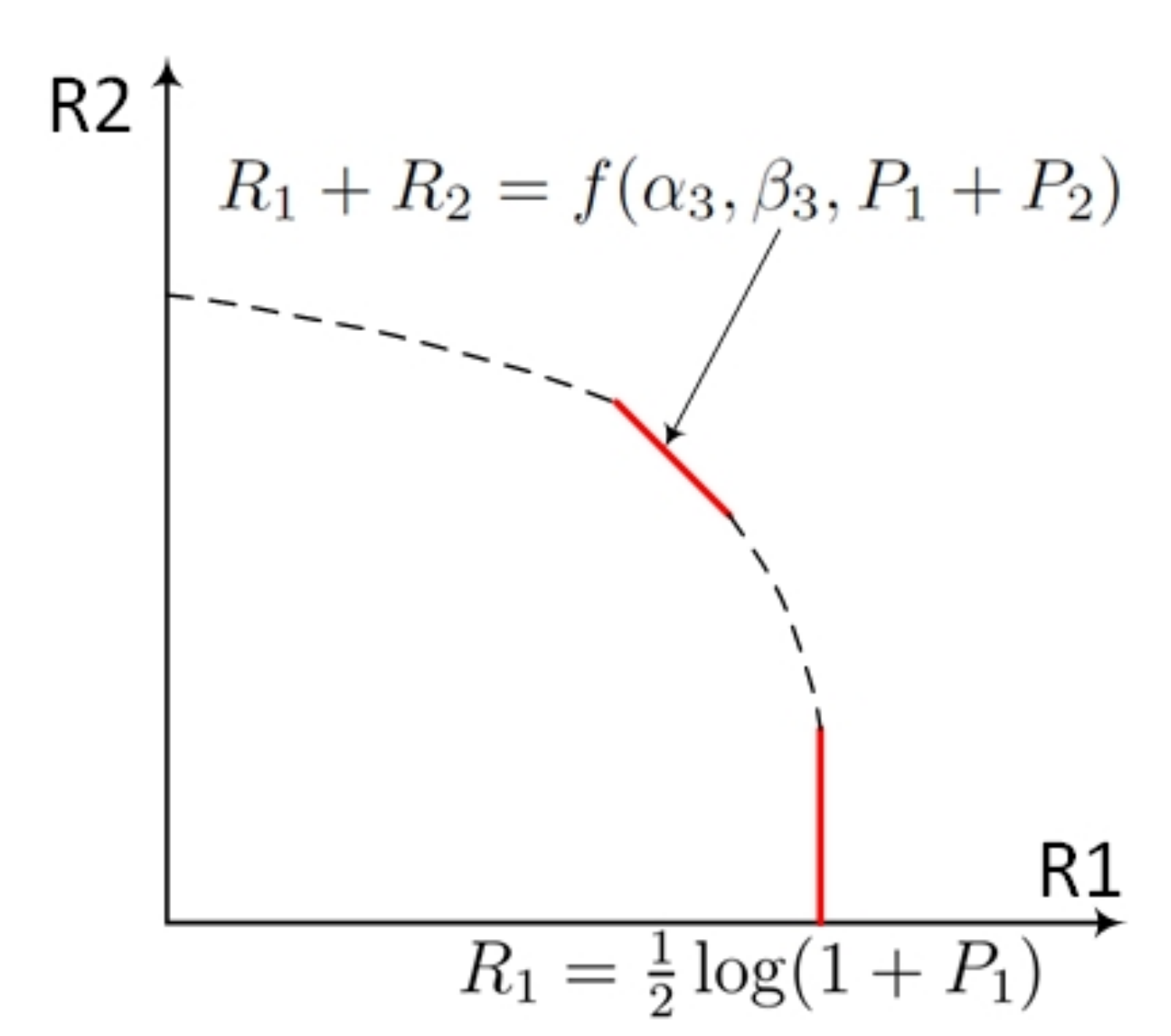}
		&\includegraphics[width=3in]{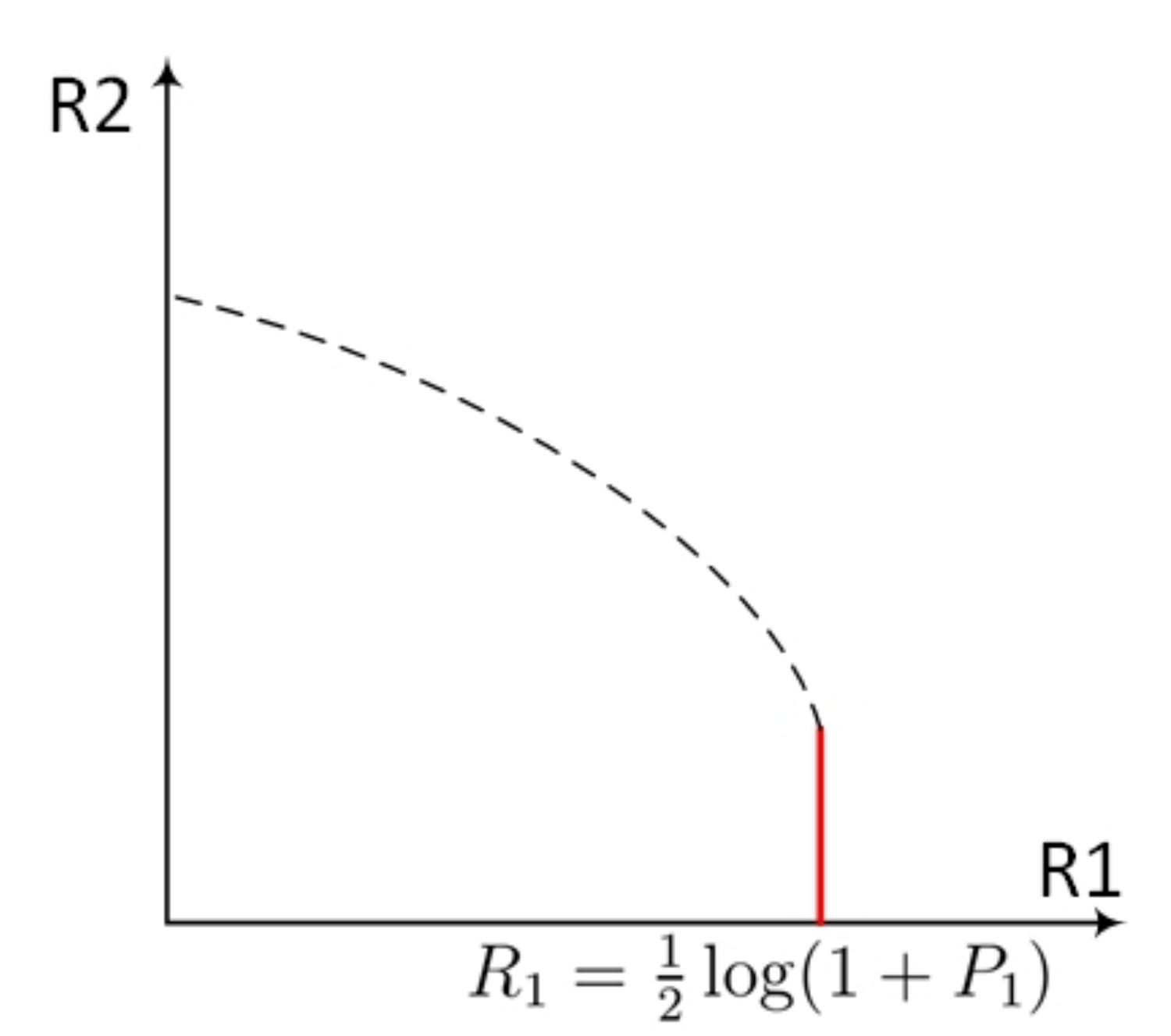}\\
		$\scriptstyle\cC_1\bigcap\cB_2\bigcap\cA_3$&$\scriptstyle\cC_1\bigcap\cB_2\bigcap\cB_3$ \\
		\includegraphics[width=3in]{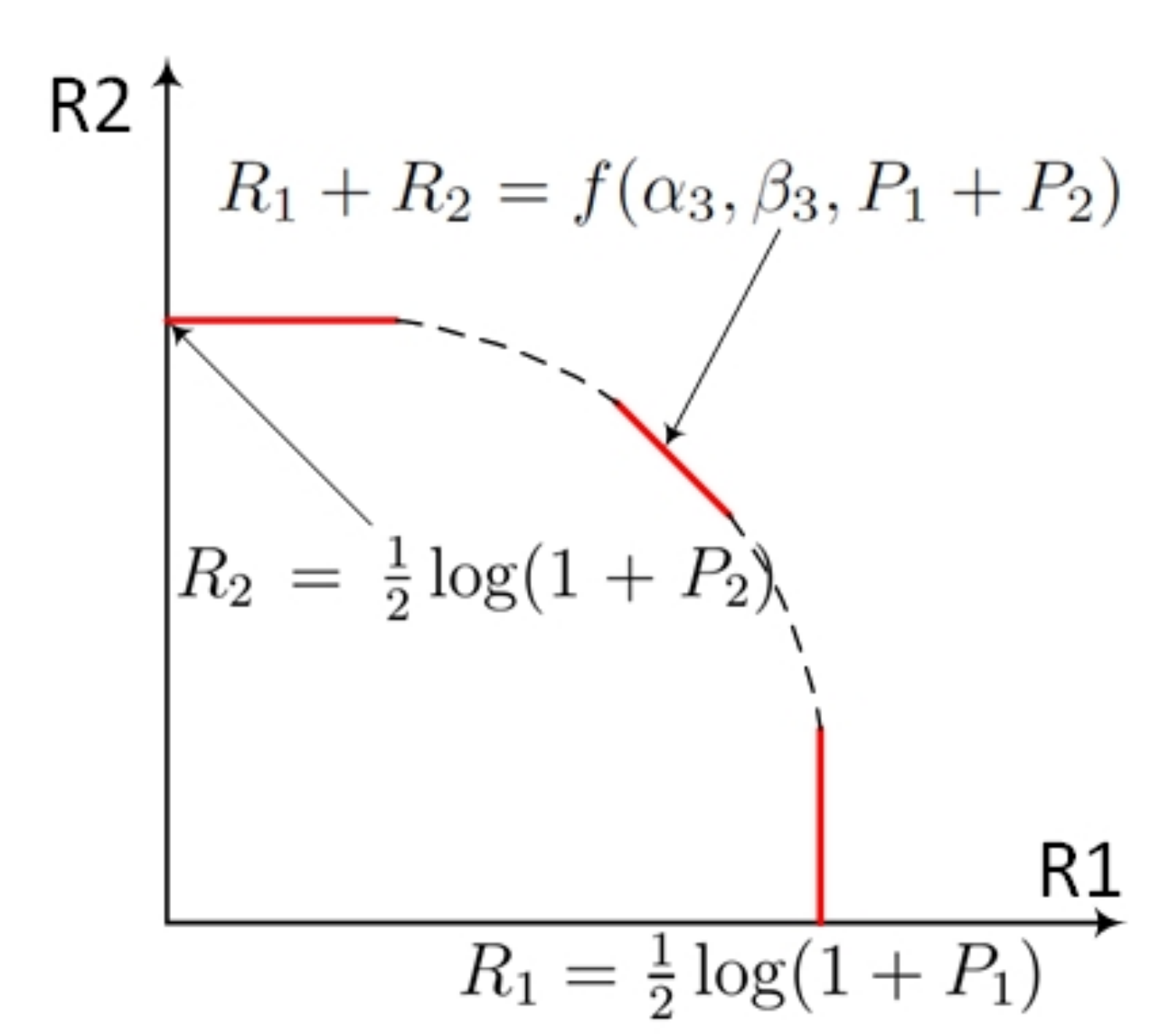}
		&\includegraphics[width=3in]{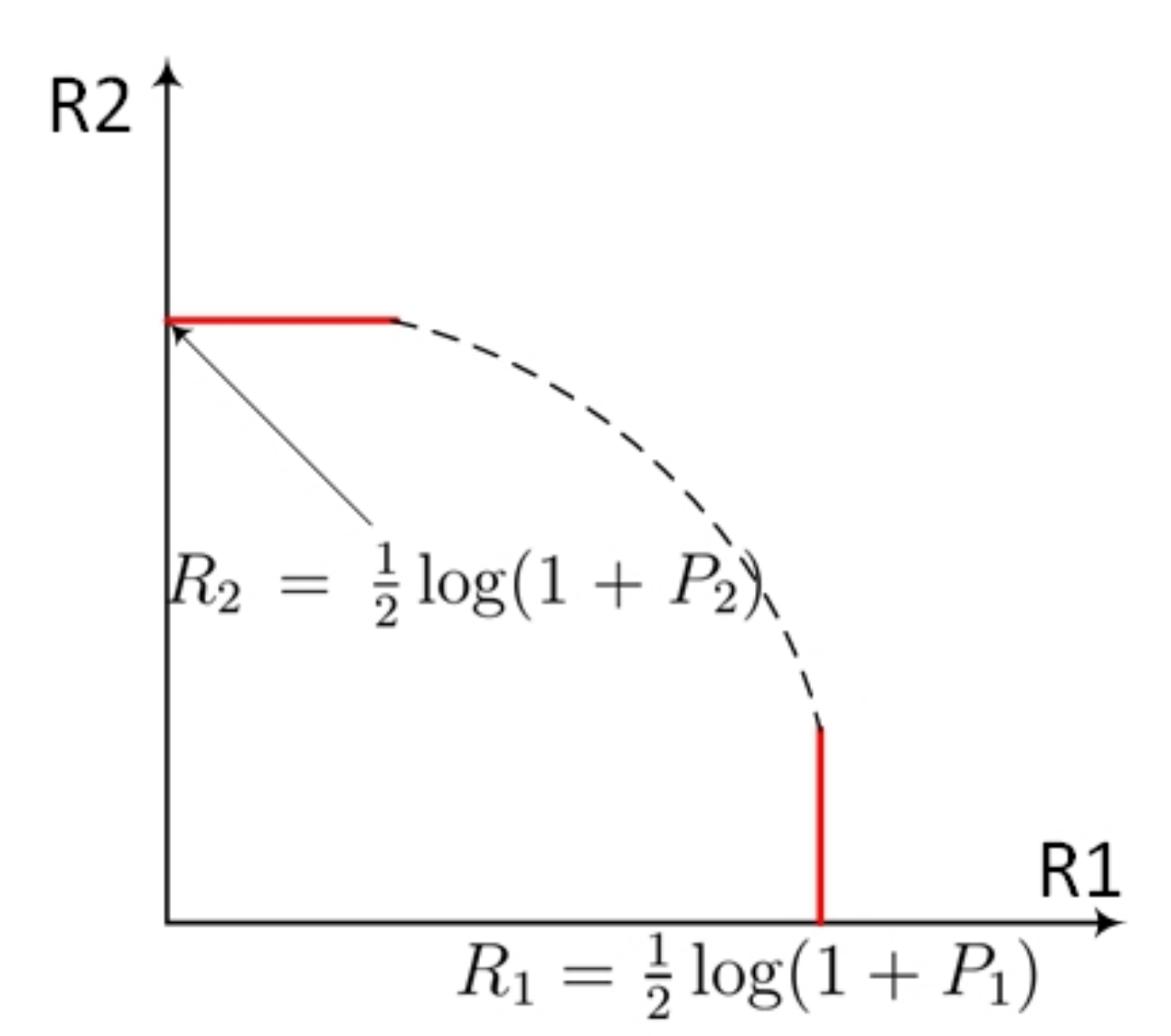}\\
		$\scriptstyle\cC_1\bigcap\cC_2\bigcap\cA_3$&$\scriptstyle\cC_1\bigcap\cC_2\bigcap\cB_3$\\
					\end{tabular}
	\end{figure*}
\begin{figure*}
\begin{tabular}{cc}
		&\includegraphics[width=3in]{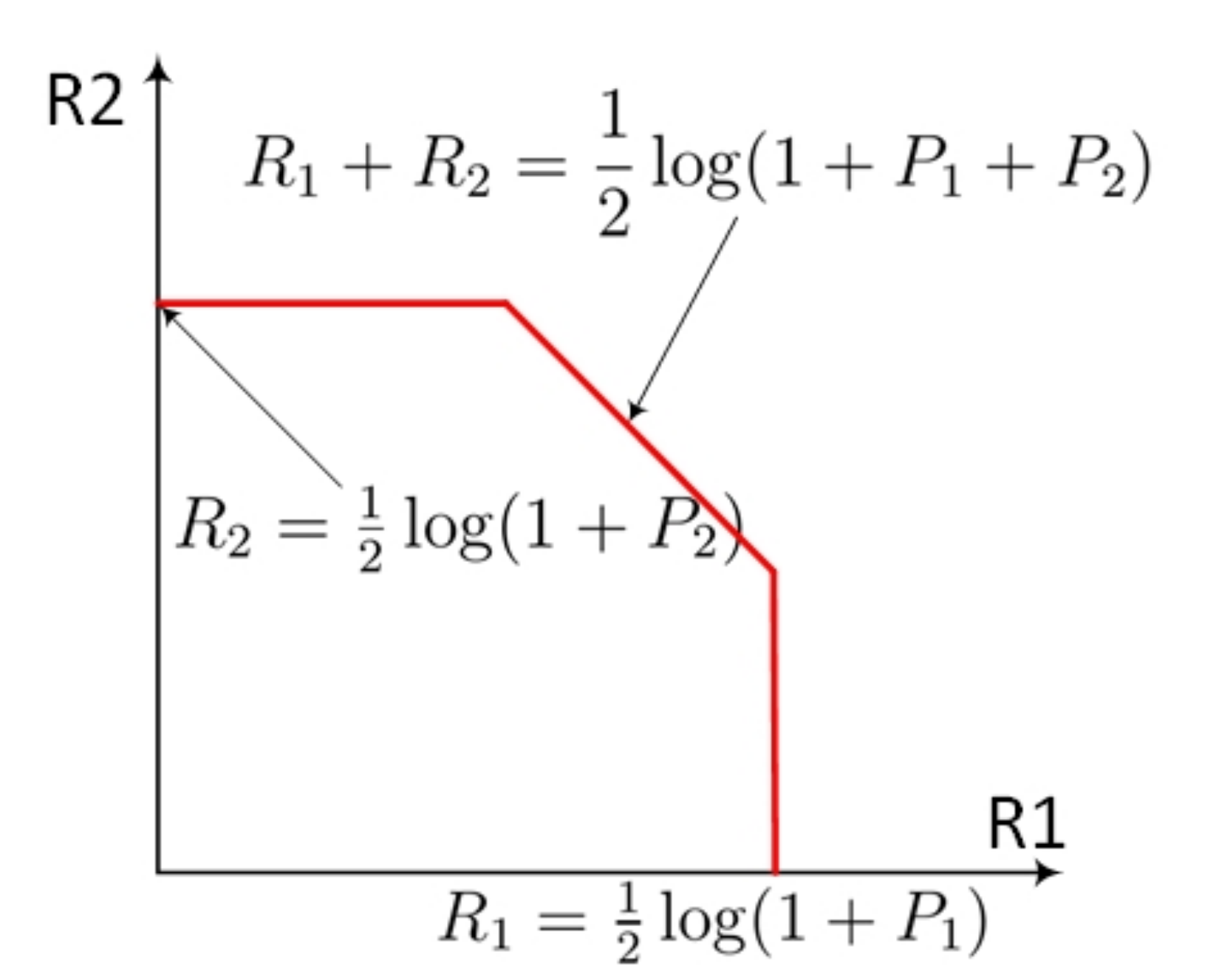} \\
		&$\scriptstyle\cC_1\bigcap\cC_2\bigcap\cC_3$
	\end{tabular}
	\centering
	\caption{Segments of the capacity region for all cases of channel parameters}\label{fig:cases}
\end{figure*}

Fig.~\ref{fig:cases} lists all possible intersections of sets that the channel parameters can belong to. In principle, there should be $3^3=27$ cases. We further note that if $(P_0,P_1,P_2,Q) \in \cC_3$, they must belong to $\cC_1$ and $\cC_2$. Hence, the total number of cases becomes $3^2\times 2+1=19$. For each case in Fig.~\ref{fig:cases}, we use the solid red lines to represent the segments on the capacity region that are characterized in Theorem \ref{th:capacity}, and we also mark the value of the capacity that each segment corresponds to as characterized in Theorem \ref{th:capacity}.

We note that for several cases, segments on the capacity region boundary are characterized to be strictly inside the capacity region of the MAC without the state, i.e., the state cannot be fully canceled. For example, for cases with $(P_0,P_1,P_2,Q) \in (\cA_1\bigcap\cA_2\bigcap\cA_3)$ and $(\cA_1\bigcap\cC_2\bigcap\cA_3)$, sum capacity segments are characterized to be smaller than the sum capacity of the MAC without state. These cases include mostly channel parameters with {\em finite} $Q$, and thus contain much larger sets of channel parameters than \cite{Duan14ISITA} that characterizes such sum capacity segment only for {\em infinite} $Q$.

We further note an interesting case (the last case in Fig.~\ref{fig:cases}), for which the capacity region is fully characterized. We state this result in the following theorem.
\begin{theorem}\label{th:region}
	If $(P_0,P_1,P_2,Q) \in (\cC_1\bigcap\cC_2\bigcap\cC_3)$, i.e.,
	\begin{flalign}\label{eq:cond}
	P_0'^2\ge \alpha^2Q(P_1+P_2+1-P_0'),
	\end{flalign}
	where $P_0'=P_0-(\alpha-1)^2Q$  for some$\alpha \in \Omega_{\alpha}$, then the capacity region of the state-dependent Gaussian MAC contains $(R_1,R_2)$ satisfying
	\begin{flalign}
	R_1 &\leqslant  \frac{1}{2}\log( 1+P_1) \nonumber \\
	R_2 &\leqslant  \frac{1}{2}\log( 1+P_2) \nonumber \\
	R_1+R_2 &\leqslant  \frac{1}{2}\log( 1+P_1+P_2) \nonumber
	\end{flalign}
	which achieves the capacity region of the Gaussian MAC without state.
\end{theorem}

Theorem \ref{th:region} implies that the state is fully canceled if the channel parameters satisfy the condition \eqref{eq:cond}. We further note two special sets of channel parameters in this case. First, if $P_0 \ge Q$, then $\alpha=0\in \Omega_{\alpha}$ and the condition clearly holds. This is not surprising because the helper has enough power to directly cancel the state. Secondly, if $P_1+P_2+1\leq P_0<Q$, then the condition holds for $\alpha=1\in \Omega_{\alpha}$ for arbitrarily large $Q$. This implies that if the helper's power is above a certain threshold, then the state can always be canceled for arbitrary state power $Q$ (even for infinite $Q$).

\newpage
\section{Technical Proofs}

\subsection{Proof of Proposition \ref{pps:Gaussian outer}}\label{apx:OuterGaussian}

The second bounds in "min" in \eqref{eq:r1outer}-\eqref{eq:r12outer} follow from the capacity of the Gaussian MAC without state. The remaining bounds arise due to capability of the helper for assisting state cancelation and are derived as follows.

Consider a $(2^{nR_1}, 2^{nR_2},n)$ code with an average error probability $P_e^{(n)}$. The probability distribution on $W_1 \times W_2 \times S^n \times X_0^n \times X_1^n\times X_2^n \times Y^n$ is given by
\begin{flalign}\label{eq:proba}
P_{W_1W_2S^nX_1^nX_2^nX_0^nY^n}=P_{W_1}P_{W_2}\left[{\prod_{i=1}^n{P_{S_i}}}\right]P_{X_1^n|W_1}P_{X_2^n|W_2}P_{X_0^n|S^n}\prod_{i=1}^n{P_{Y_{i}|X_{1i}X_{2i}X_{0i}S_i}}.
\end{flalign}

By Fano's inequality, we have
\begin{flalign}
H(W_1W_2|Y^n) & \leqslant n(R_1+R_2)P_e^{(n)}+1 = n\delta_{n} \label{eq:Fano}
\end{flalign}
where $\delta_{n} \to 0$ as $n\to +\infty$.

We first bound $R_1$ based on Fano's inequality as follows:
\begin{flalign}
nR_1 &\leqslant I(W_1;Y^n)+n\delta_{n}\nn\\
&\leqslant I(X_1^n;Y^n)+n\delta_{n}\nn\\
&= H(X_1^n)-H(X_1^n|Y^n)+n\delta_{n}\nn\\
&\overset{(a)}{\leqslant} H(X_1^n|X_2^n)-H(X_1^n|X_2^nY^n)+n\delta_{n}\nn\\
&= I(X_1^n;Y^n|X_2^n)+n\delta_{n}\nn\\
&= H(Y^n|X_2^n)-H(Y^n|X_1^nX_2^n)]+n\delta_{n}\nn\\
&= H(Y^n|X_2^n)-H(S^nY^n|X_1^nX_2^n)+H(S^n|X_1^nX_2^nY^n)+n\delta_{n}\nn\\
&= H(Y^n|X_2^n)-H(Y^n|S^nX_1^nX_2^n)-H(S^n|X_1^nX_2^n)+H(S^n|X_1^nX_2^nY^n)+n\delta_{n}\nn\\
&\leqslant H(Y^n|X_2^n)-H(Y^n|S^nX_0^nX_1^nX_2^n)-H(S^n)+H(S^n|X_1^nX_2^nY^n)+n\delta_{n}\nn\\
&\overset{(b)}{\leqslant } \sum_{i=1}^n [ H(Y_{i}|X_{2i})- H(Y_{i}|S_iX_{0i}X_{1i}X_{2i}) -H(S_i)+H(S_i|X_{1i}X_{2i}Y_i)] +n\delta_{n} \label{eq:r1}
\end{flalign}
where (a) follows because $X_1^n$ and $X_2^n$ are independent, and (b) follows because $S^n$ is an i.i.d.\ sequence.

We bound the first term in the above equation as
\begin{flalign}
\frac{1}{n}& \sum_{i=1}^n h(Y_{i}|X_{2i}) \nn \\
&\leqslant \frac{1}{2n} \sum_{i=1}^n \log 2\pi e(Var(X_{1i}+X_{0i}+S_i+N_i)) \nn \\
& = \frac{1}{2n} \sum_{i=1}^n \log 2\pi e(Var(X_{1i})+Var(X_{0i}+S_i)+Var(N_i)) \nn \\
&\leqslant \frac{1}{2n} \sum_{i=1}^n \log 2\pi e\Big(E[X_{1i}^2]+E[X_{0i}^2]+ 2E(X_{0i}S_i) + E[S_i^2] + E[N_i^2])\Big)\nn\\
&\leqslant \frac{1}{2} \log 2\pi e \Bigg(\frac{1}{n}\sum_{i=1}^n E[X_{1i}^2]+\frac{1}{n}\sum_{i=1}^n E[X_{0i}^2]+ \frac{2}{n}\sum_{i=1}^n E(X_{0i}S_i)+ \frac{1}{n}\sum_{i=1}^n E[S_i^2] \\
&\hspace{1cm}+ \frac{1}{n}\sum_{i=1}^n E[N_i^2])\Bigg)\nn\\
&\leqslant \frac{1}{2} \log 2\pi e\Bigg(P_1 + P_0 + Q +1+ \frac{2}{n} \sum_{i=1}^n E(X_{0i}S_i)\Bigg)\nn\\
&\leqslant\frac{1}{2} \log 2\pi e\Big(P_1 + P_0 + Q +1 + 2\rho_{0s}\sqrt{P_0Q}\Big) \label{eq:r1-1}
\end{flalign}
where $\rho_{0s}=\frac{1}{n\sqrt{P_0Q}} \sum_{i=1}^n E(X_{0i}S_i)$.

It is easy to obtain bounds on the second and third terms in \eqref{eq:r1} as follows.
\begin{flalign}
& \frac{1}{n}\sum_{i=1}^n H(Y_{i}|S_iX_{0i}X_{1i}X_{2i}) =\frac{1}{2} \log 2\pi e \label{eq:r1-2}\\
& \frac{1}{n}\sum_{i=1}^n H(S_i)=\frac{1}{2} \log 2\pi e Q \label{eq:r1-3}
\end{flalign}

We next bound the last term in \eqref{eq:r1} as follows.
\begin{flalign}
\frac{1}{n} \sum_{i=1}^n  h(S_i|X_{1i} X_{2i} Y_{i}) &=\frac{1}{n} \sum_{i=1}^n  h(S_i|X_{0i}+S_i+N_i)\nn\\
&\leq \frac{1}{n} \sum_{i=1}^n  h(S_i-\alpha(X_{0i}+S_i+N_i)|X_{0i}+S_i+N_i)\nn\\
&\leqslant \frac{1}{n} \sum_{i=1}^n  h(S_i-\alpha(X_{0i}+S_i+N_i)) \nn\\
&=\frac{1}{2} \log 2\pi e(\alpha^2 P_0+(1-\alpha)^2Q-2\alpha(1-\alpha)\rho_{0s}\sqrt{P_0Q}+\alpha^2)\nn\\
&=\frac{1}{2} \log 2\pi e\biggl( \frac{Q+(P_0-\rho_{0S}^2P_0) Q)}{Q+2\rho_{0S}\sqrt{P_0Q}+P_0+1}\biggl)\label{eq:r1-4}
\end{flalign}
where the last equation follows by setting $$\alpha=\frac{\rho_{0s}\sqrt{P_0Q}+Q}{1+P_0+Q+2\rho_{0s}\sqrt{P_0Q}}$$ so that $S_i-\alpha(X_{0i}+S_i+N_i)$ and $X_{0i}+S_i+N_i$ are uncorrelated.

Combining the above four bounds, we obtain the following upper bound on $R_1$.
\begin{flalign}
R_1\leqslant & \frac{1}{2} \log2\pi e(1+P_0+P_1+Q+2\rho_{0s}\sqrt{P_0Q}) -\frac{1}{2}\log 2\pi e -\frac{1}{2} \log 2\pi eQ\nn\\
\ &+\frac{1}{2} \log 2\pi e\biggl( \frac{Q+(P_0-\rho_{0S}^2P_0) Q)}{Q+2\rho_{0S}\sqrt{P_0Q}+P_0+1}\biggl)\nn\\
\leqslant & \frac{1}{2}\log(1+\frac{P_1}{Q+2\rho_{0S}\sqrt{P_0Q}+P_0+1})+ \frac{1}{2}\log(1+P_0-\rho_{0S}^2P_0)
\end{flalign}

Similarly, we can derive an upper bound for $R_2$ as
\begin{flalign}
R_2 \leqslant & \frac{1}
{2}\log(1+\frac{P_2}{Q+2\rho_{0S}\sqrt{P_0Q}+P_0+1})+ \frac{1}{2}\log(1+P_0-\rho_{0S}^2P_0).
\end{flalign}

We further bound $R_1+R_2$ following similar arguments. We highlight some important steps below.
\begin{flalign}
n(R_1+R_2) &\leqslant I(W_1W_2;Y^n)+n\delta_{n}\nn\\
&\leqslant I(X_1^nX_2^n;Y^n)+n\delta_{n}\nn\\
&= H(Y^n)-H(Y^n|X_1^nX_2^n)+n\delta_{n}\nn\\
&= H(Y^n)-H(S^nY^n|X_1^nX_2^n)+H(S^n|X_1^nX_2^nY^n)+n\delta_{n}\nn\\
&= H(Y^n)-H(Y^n|S^nX_1^nX_2^n)-H(S^n|X_1^nX_2^n)+H(S^n|X_1^nX_2^nY^n)\nn\\
&\hspace{1cm} +n\delta_{n}\nn\\
&\leqslant H(Y^n)-H(Y^n|S^nX_0^nX_1^nX_2^n)-H(S^n)+H(S^n|X_1^nX_2^nY^n)+n\delta_{n}\nn\\
&\leqslant  \sum_{i=1}^n [ H(Y_{i})- H(Y_{i}|S_iX_{0i}X_{1i}X_{2i}) -H(S_i)+H(S_i|X_{1i}X_{2i}Y_i)] +n\delta_{n}\label{eq:r12}
\end{flalign}

The first term in \eqref{eq:r12} can be bounded as follows.
\begin{flalign}
\frac{1}{n} \sum_{i=1}^n h(Y_{i})
&\leqslant \frac{1}{2n} \sum_{i=1}^n \log 2\pi e(Var(X_{1i}+X_{2i}+X_{0i}+S_i+N_i)) \nn \\
&= \frac{1}{2n} \sum_{i=1}^n \log 2\pi e(Var(X_{1i})+Var(X_{2i})+Var(X_{0i}+S_i)+Var(N_i)) \nn \\
&\leqslant \frac{1}{2n} \sum_{i=1}^n \log 2\pi e\Big(E[X_{1i}^2]+E[X_{2i}^2]+E[X_{0i}^2]+ 2E(X_{0i}S_i)+ E[S_i^2]\nn\\
&\hspace{1cm}  + E[N_i^2])\Big)\nn\\
&\leqslant \frac{1}{2} \log 2\pi e \Bigg(\frac{1}{n}\sum_{i=1}^n E[X_{1i}^2]+\frac{1}{n}\sum_{i=1}^n E[X_{2i}^2]+\frac{1}{n}\sum_{i=1}^n E[X_{0i}^2] \nn\\
&\hspace{1cm}+ \frac{2}{n}\sum_{i=1}^n E(X_{0i}S_i)+ \frac{1}{n}\sum_{i=1}^n E[S_i^2] + \frac{1}{n}\sum_{i=1}^n E[N_i^2])\Bigg)\nn\\
&\leqslant\frac{1}{2} \log 2\pi e\Big(P_1 + P_2 + Q +1 + 2\rho_{0s}\sqrt{P_0Q}\Big) \label{eq:r12-1}
\end{flalign}

Other bounds in \eqref{eq:r12} can be bounded in the way as in \eqref{eq:r1-2}, \eqref{eq:r1-3}, and \eqref{eq:r1-4}. Combining these bounds with \eqref{eq:r12-1}, we obtain the following desired upper bound on $R_1+R_2$.
\begin{flalign}
R_1+R_2 \leqslant & \frac{1}{2}\log\left(1+\frac{P_1+P_2}{Q+2\rho_{0S}\sqrt{P_0Q}+P_0+1} \right)+ \frac{1}{2}\log(1+P_0-\rho_{0S}^2P_0)
\end{flalign}

\subsection{Proof of Proposition \ref{pps:DMC inner}}\label{apx:DMC inner}
We use random codes and fix the following joint distribution:
\begin{flalign}
P_{SUX_0X_1X_2Y}=P_{SU}P_{X_0|SU}P_{X_1}P_{X_2}P_{Y|SX_0X_1X_2}.\nn
\end{flalign}
Let $T_\epsilon^n(P_{SUX_0X_1X_2Y})$ denote the strongly joint $\epsilon$-typical set based on the above distribution. For a given sequence $x^n$, let $T_\epsilon^n(P_{U|X}|x^n)$ denote the set of sequences $u^n$ such that $(u^n, x^n)$ is jointly typical based on the distribution $P_{XU}$.

\begin{enumerate}
	\item Codebook Generation:
	\begin{itemize}
		\item Generate $2^{n\tR}$ codewords $u^n(v)$ with i.i.d.\ components based on $P_U$. Index these codewords by $v = 1, \ldots, 2^{n\tR}$.
		\item Generate $2^{nR_1}$ codewords $x_1^n(w_1)$ with i.i.d.\ components based on $P_{X_1}$. Index these codewords by $w_1 =1, \ldots, 2^{nR_1}$.
		\item Generate $2^{nR_2}$ codewords $x_2^n(w_2)$ with i.i.d.\ components based on $P_{X_2}$. Index these codewords by $w_2 =1, \ldots, 2^{nR_2}$.
	\end{itemize}
	\item Encoding:
	\begin{itemize}
		\item Helper: Given $s^n$, find $\tv$, such that $(u^n(\tv), s^n)\in T^n_\epsilon(P_{SU})$. It can be shown that for large $n$, such $\tv$ exists with high probability if
		\begin{equation}
		\tR\geqslant I(S;U). \label{eq:binning}
		\end{equation}
		Then given $(u^n(\tv),s^n)$, generate $x_0^n$ with i.i.d.\ components based on $P_{X_0|SU}$ for transmission.
		\item Transmitter 1: Given $w_1$, map $w_1$ into $x_1^n(w_1)$ for transmission.
		\item Transmitter 2: Given $w_2$, map $w_2$ into $x_2^n(w_2)$ for transmission.
	\end{itemize}
	\item Decoding:
	\begin{itemize}
		\item Given $y^n$, find $(\hv,\hw_1, \hw_2)$ such that $(u^n(\hv), x^n_1(\hw_1), x^n_2(\hw_2), y^n)\in T^n_\epsilon(P_{UX_1X_2Y})$. If no or more than one $(\hw_1,\hw_2)$ can be found, declare an error($\hv$ is not necessary to be correctly decoded).
	\end{itemize}
\end{enumerate}
It can be shown that for sufficiently large n, decoding is correct with high probability if
\begin{flalign}
R_1 \leqslant & I(X_1;Y|X_2,U) \nn \\
\tR+R_1 \leqslant &I(U,X_1;Y|X_2) \nn \\
R_2 \leqslant& I(X_2;Y|X_1,U) \nn \\
\tR+R_2 \leqslant&I(U,X_2;Y|X_1)\nn \\
R_1+R_2 \leqslant& I(X_1,X_2;Y|U)\nn\\
\tR+R_1+R_2 \leqslant&I(U,X_1,X_2;Y) \nn
\end{flalign}
We note that the event that multiple $\hv$ with only single pair $(\hw_1,\hw_2)$ satisfy the above decoding requirement is not counted as an error event, because the index $v$ is not the decoding requirement. Finally, combining the above bounds with \eqref{eq:binning} yields the desired achievable region.

\chapter{State-Dependent Z-Interference Channel with Correlated States}\label{chap:z2state}

In this chapter, we study the state-dependent Z-IC channel with correlated states. This state-dependent Z-IC is different from Z-IC without state as each receiver is corrupted by a channel state and the two transmitters know the information of both of the channel states noncausally. Our focus here is on the more general scenario, where the two receivers
are corrupted by two correlated states, and our aim is to understand how the correlation
affects the design of the scheme.

The rest of chapter is organized as follows. In Section \ref{sec:zmodel}, we describe the channel model. In Section \ref{sec:zvs}, we study the model in the very strong interference regime and characterize the channel parameters under which the two receivers achieve their corresponding point-to-point channel capacity without state and interference. In Section \ref{sec:zs}, we study the model in the strong but not very strong interference regime and characterize the sum capacity boundary partially under certain channel parameters based on the joint design of rate splitting, successive cancellation, as well as dirty paper coding. In Section \ref{sec:zweak}, we study the model in the weak interference regime and characterize the sum capacity, which is achieved by the two transmitters independently designing dirty paper coding and treating interference as noise.

\section {Channel Model}\label{sec:zmodel}

\vspace{5mm}
\begin{figure}[thb]
	\centering
	\includegraphics[width=4.5in]{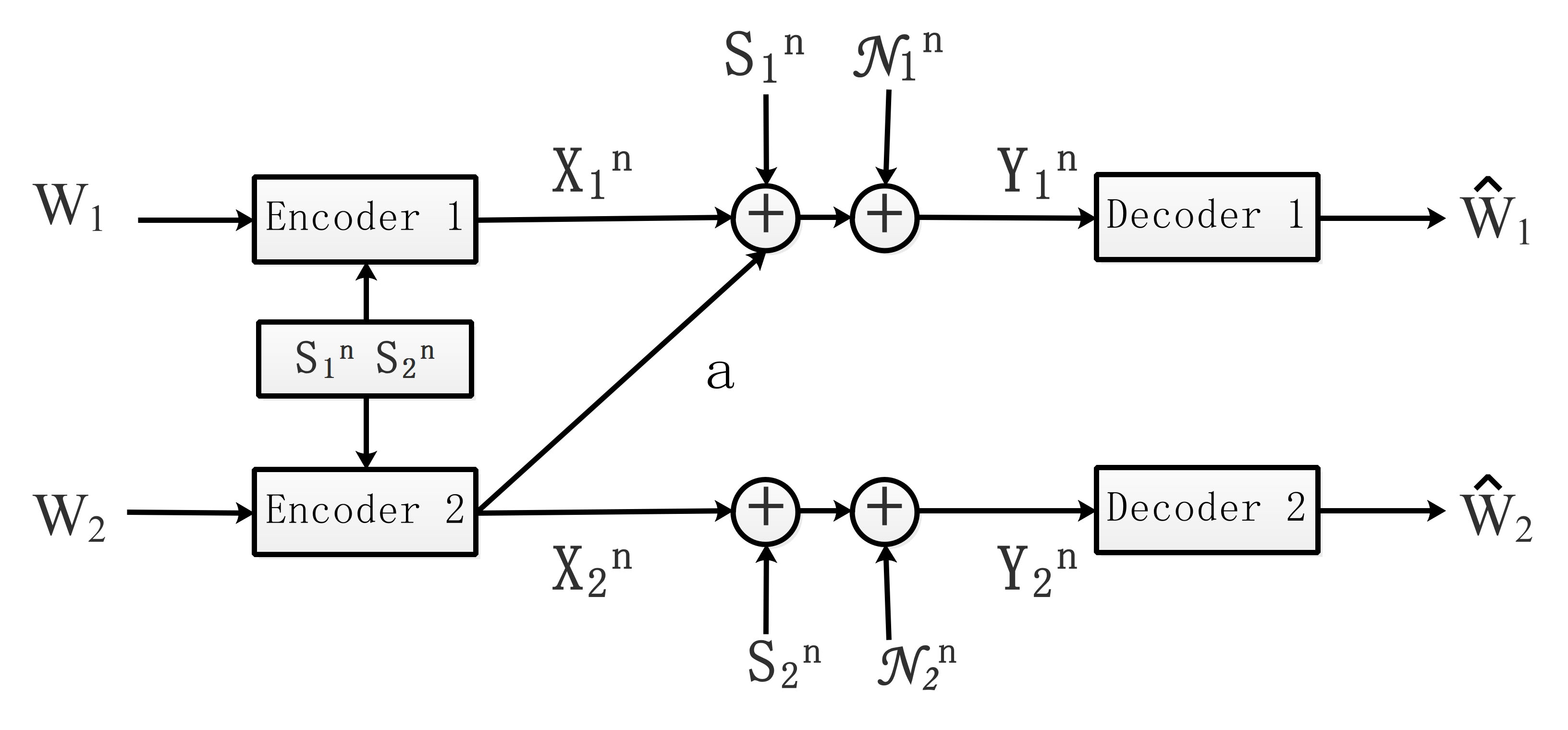}
	\caption{The state-dependent Z-IC}\label{fig:channel}
\end{figure}
\vspace{5mm}

We consider the state-dependent Z-IC (as shown in Fig.~\ref{fig:channel}), in which transmitters 1 and 2 send two messages $W_{1}$ and $W_{2}$ to two receivers 1 and 2, respectively.  Receiver 1's output is interfered by transmitter 2's input as well as a state sequence $S_1^n$, and receiver 2's output is interfered only by a state sequence $S_2^n$, which is correlated with $S_1^n$. The two state sequences $S_1^n$ and $S_2^n$ are assumed to be known {\em noncausally}  at both transmitters. The encoder $k$ at transmitter $k$ maps the message $w_k\in \cW_k=\{1,\dots,2^{nR_k}\}$ and the state sequences $s_1^n$ and $s_2^n$ to a codeword $x_k^n\in \cX_k^n$ for $k=1,2$. The two inputs $x_1^n$ and $x_2^n$ are transmitted over the memoryless Z-IC characterized by $P_{Y_1|X_1X_2S_1}$ and $P_{Y_2|X_2S_2}$. Receiver 1 is required to decode $W_1$ and receiver 2 is required to decode $W_2$. 
The average probability of error for a length-$n$ code is defined as
\begin{flalign}\nn
P_e^{(n)} = & \frac{1}{|\cW_1||\cW_2|}\sum_{w_1=1}^{|\cW_1|}\sum_{w_2=1}^{|\cW_2|} Pr\lbrace(\hat{w}_1, \hat{w}_2) \neq (w_1, w_2)\rbrace.
\end{flalign}
A rate pair $(R_1, R_2)$ is {\em achievable} if there exist a sequence of encoding and decoding schemes 
such that the average error probability $P_e^{(n)} \rightarrow 0$ as $n \to \infty$. The {\em capacity region} is defined to be the closure of the set of all achievable rate pairs.

In this dissertation, we focus on the Gaussian Z-IC with the outputs at the two receivers for one channel use given by
\begin{subequations}
	\begin{flalign}
	Y_1&=X_1+ aX_2+S_1+N_1\label{eq:zGeneralChannelModel}\\
	Y_2&=X_2+S_2+N_2
	\end{flalign}
\end{subequations}	
where $a$ is the channel gain coefficient, and $N_1$ and $N_2$ are noise variables with Gaussian distributions $N_1 \sim \mathcal{N}(0,1)$ and $N_2 \sim \mathcal{N}(0,1)$. The state variables $S_1$ and $S_2$ are jointly Gaussian with correlation coefficient $\rho$ and marginal distributions $S_1 \sim \mathcal{N}(0,Q_1)$ and $S_2\sim \mathcal{N}(0,Q_2)$. Both the noise variables and the state variables are i.i.d. over channel uses. 
The channel inputs $X_1$ and $X_2$ are subject to the average power constraints $P_1$ and $P_2$.

Our goal is to characterize channel parameters, under which the capacity of the corresponding Z-IC without the presence of the states can be achieved, and thus the capacity region of the Z-IC with the presence of state is also established. In particular, we are interested in understanding the impact of the correlation between the states on the capacity characterization.

\section {Very Strong Interference Regime}\label{sec:zvs}

In this section, we study the state-dependent Z-IC in the very strong regime, in which the channel parameters satisfy $a^2>1+P_1$. For the corresponding Z-IC without states, the capacity region contains rate pairs $(R_1,R_2)$ satisfying 
\begin{flalign}\label{cap:zVeryStrong}
&R_1\leqslant \frac{1}{2}\log(1+P_1), \;\; \\
&R_2 \leqslant  \frac{1}{2}\log(1+P_2).
\end{flalign}
In this case, the two receivers achieve the point-to-point channel capacity without interference. Furthermore, in \cite{Duan16IT}, an achievable scheme has been established to achieve the same point-to-point channel capacity when the two receivers are corrupted by the same but differently scaled state. Our focus here is on the more general scenario, where the two receivers are corrupted by two {\em correlated} states, and our aim is to understand how the correlation affects the design of the scheme.


We first design an achievable scheme to obtain an achievable rate region for the discrete memoryless Z-IC. The two transmitters encode their messages $W_1$ and $W_2$ into two auxiliary random variables $U$ and $V$, respectively, based on the Gel'fand-Pinsker binning scheme. Since receiver 2 is interference free and is corrupted by $S_2$, the auxiliary random variable $V$ is designed with regard to only $S_2$. Furthermore,
receiver 1 first decodes $V$, then uses it to cancel the interference $X_2$ and partial state interference, and finally decodes its own message $W_1$ by decoding $U$. Here, since $S_2$ is introduced to $Y_1$ when canceling $X_2$ via $V$, the auxiliary random variable $U$ is designed based on both $S_1$ and $S_2$ to fully cancel the states. Based on such a scheme, we obtain the following achievable region.
\begin{proposition}\label{pps:Z_inner}
	For the state-dependent Z-IC with the states noncausally known at both transmitters, an achievable region consists of rate pairs $(R_1,R_2)$ satisfying:
	\begin{subequations}
		\begin{flalign}
		R_1 \leqslant & I(U;VY_1) -I(S_1,S_2;U) \label{eq:zv_inner1}\\
		R_2 \leqslant &  \min\{ I(V;Y_2),I(V;Y_1)\} -I(S_2;V) \label{eq:zv_inner2}
		\end{flalign}
	\end{subequations}
	for some distribution $P_{S_1S_2}P_{U|S_1S_2}P_{X_1|US_1S_2}P_{V|S_2}P_{X_2|VS_2}$
	$P_{Y_1|S_1X_1X_2}P_{Y_2|S_2X_2}$.
\end{proposition}
\begin{proof}
	See Section \ref{apx:Z inner}	
\end{proof}
Following Proposition \ref{pps:Z_inner}, we further simplify the achievable region in the following corollary, which is in a useful form for us to characterize the capacity region for the Gaussian Z-IC. 
\begin{corollary}\label{cor:Z_inner}
	For the state-dependent Z-IC with the states noncausally known at both transmitters, if the following condition
	\begin{flalign}
	I(V;Y_2)\leqslant I(V;Y_1)\label{zvcond}
	\end{flalign} 
	is satisfied, then an achievable region consists of rate pairs $(R_1,R_2)$ satisfying:
	\begin{subequations}
		\begin{flalign}
		R_1 \leqslant & I(U;VY_1) -I(S_1,S_2;U) \\
		R_2 \leqslant & I(V;Y_2) -I(S_2;V) 
		\end{flalign}
	\end{subequations}
	for some distribution $P_{S_1S_2}P_{U|S_1S_2}P_{X_1|US_1S_2}P_{V|S_2}P_{X_2|VS_2}$
	$P_{Y_1|S_1X_1X_2}P_{Y_2|S_2X_2}$.
\end{corollary}
In Corollary \ref{cor:Z_inner}, condition \eqref{zvcond} requires that receiver 1 is more capable in decoding $V$ (and hence $W_2$) than receiver 2, which is likely to be satisfied in the very strong regime. 

We now study the Gaussian Z-IC. Since $S_1$ and $S_2$ are jointly Gaussian, $S_1$ can be expressed as $S_1=dS_2+S_1^\prime$ where $d$ is a constant representing the level of correlation, and $S_1'$ is independent from $S_2$ and $S_1'\sim \mathcal{N}(0, Q_1')$ with $Q_1=d^2Q_2+Q_1^\prime$. Thus, without loss of generality, the channel model can be expressed in the following equivalent form that is more convenient for analysis,
\begin{subequations}
	\begin{flalign}
	Y_1&=X_1+ aX_2+dS_2+S_1^\prime+N_1\\
	Y_2&=X_2+S_2+N_2.
	\end{flalign}
\end{subequations}	

Following Corollary \ref{cor:Z_inner}, we characterize the channel parameters under which both the states and interference can be fully canceled, and hence the capacity region for the Z-IC is obtained. 

	\begin{theorem}\label{thr:VSCapacity}
	For the state-dependent Gaussian Z-IC with states noncausally known at both transmitters, if the channel parameters  $(a,d,P_1,P_2,Q_1^\prime,Q_2)$ satisfy the following condition:
	\begin{equation} \frac{P_1+a^2P_2+d^2Q_2+Q_1^\prime+1}{(d+a\beta)^2Q_2P_2+(P_2+\beta^2Q_2)(P_1+Q_1^\prime+1])}\geqslant \frac{P_2+1}{P_2} \label{eq:zverystrongcond}
	\end{equation}
	where $\beta=\frac{P_2}{P_2+1}$, then the capacity region is characterized by \eqref{cap:zVeryStrong}.	
	\end {theorem}
	\begin{proof}
		Theorem \ref{thr:VSCapacity} follows from Corollary \ref{cor:Z_inner} by setting $U=X_1+\alpha_1S_2+\alpha_2S_1^\prime$, $V=X_2+\beta S_2$, where $X_1$, $X_2$, $ S_1^\prime$ and $ S_2 $ are independent Gaussian variables with mean zero and variances $P_1$, $P_2$, $ Q_1^\prime $ and $ Q_2 $, respectively. As discussed in the proof of Proposition \ref{pps:Z_inner}, V is first decoded by decoder 1. And then inspired by dirty paper coding, we design $ \alpha_1 $, $\alpha_2$ and $\beta$ for both $Y_2=X_2+S_2+N_2$ and $Y_1^{\prime}=Y_1-aV=X_1+(d-a\beta)S_2+S_1^\prime+N_1$ to fully cancel the states. Thus, the coefficients should satisfy the following conditions,
		\begin{subequations}  
			\begin{flalign}
			\frac{\alpha_1}{d-a\beta}&=\frac{P_1}{P_1+1}\label{eq:zdirty_con1}\\
			\alpha_2&=\frac{P_1}{P_1+1}\label{eq:zdirty_con2}\\
			\beta&=\frac{P_2}{P_2+1},\label{eq:zdirty_con3}
			\end{flalign}
		\end{subequations}
		which further yields
		$\alpha_1$, $\alpha_2$ and $\beta$ to satisfy
		\begin{flalign}
		\alpha_1&=\frac{P_1}{P_1+1}(d-\frac{aP_2}{P_2+1}), \quad\nn \\
		\alpha_2&=\frac{P_1}{P_1+1},\quad\nn\\
		\beta&=\frac{P_2}{P_2+1}.\ \ \ \ \  \nn
		\end{flalign}
		Substituting the above choice of the auxiliary random variables and the parameters into \eqref{zvcond} in Corollary \ref{cor:Z_inner}, we obtain the condition \eqref{eq:zverystrongcond}. Substituting those choices into the inner bound in Corollary \ref{cor:Z_inner}, we obtain the capacity region characterized by \eqref{cap:zVeryStrong}. Since such an achievable region achieves the point-to-point channel capacity for the Z-IC without the state, it can be shown to be the capacity region of the state-dependent Z-IC.
	\end{proof}
	\begin{figure}[thb]
		\centering
		\includegraphics[height=3in,width=5in]{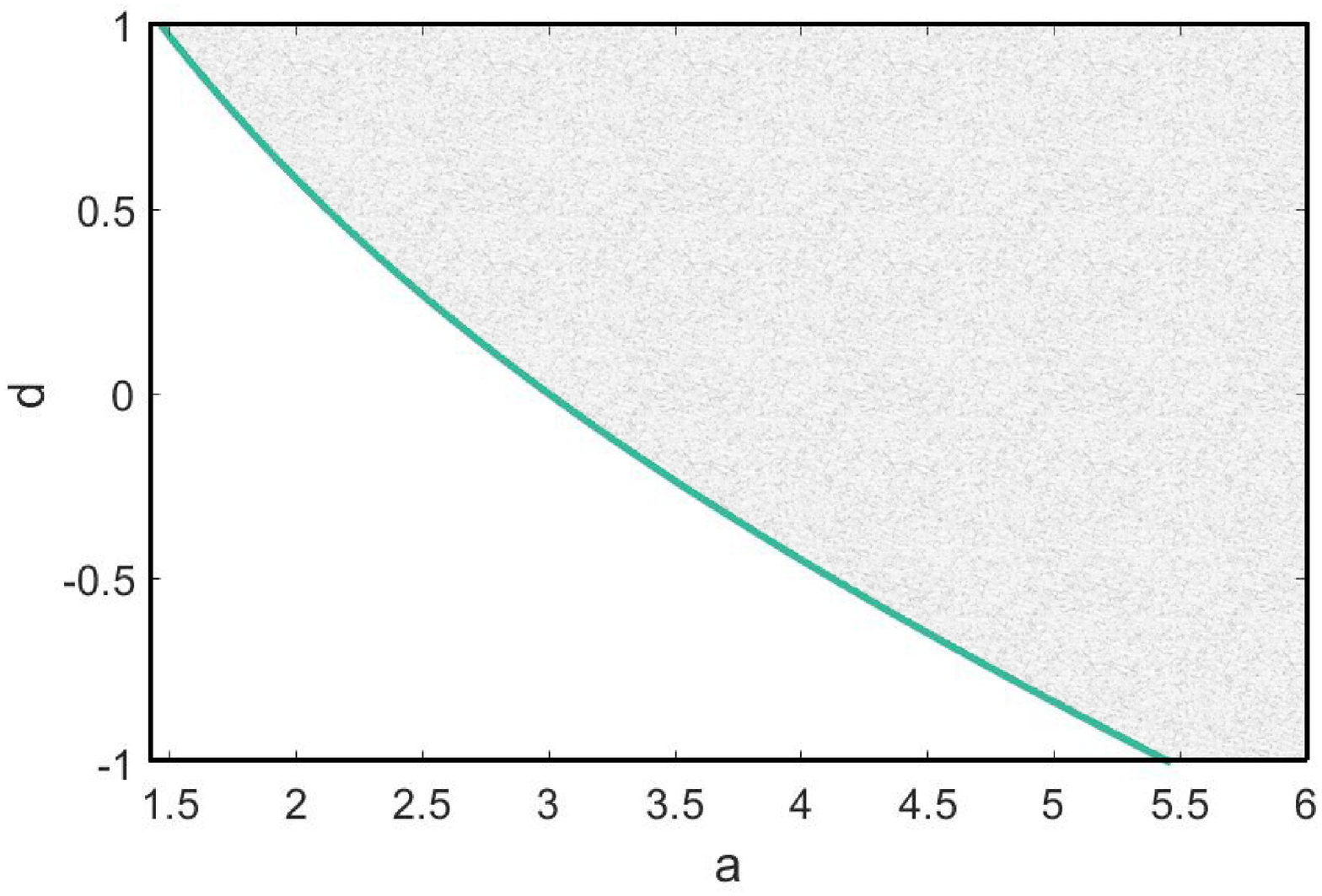}
		\caption{Characterization of channel parameters $(a,d)$ in shaded area under which the state-dependent Gaussian Z-IC achieves the capacity of the corresponding channel without states and interference in very strong regime.}\label{fig:vsac1}
	\end{figure}
	
	Based on Theorem \ref{thr:VSCapacity}, if channel parameters satisfy the condition \eqref{eq:zverystrongcond}, we can simultaneously cancel two states and the interference, and the point-to-point capacity of two receivers without state and interference can be achieved. The correlation between the two states captured by $d$ plays a very important role regarding whether the condition can be satisfied.
	In Fig.~\ref{fig:vsac1}, we set $P_1=2$, $P_2=2$, $Q_1=1$ and $Q_2=1$, and plot the range of the parameter pairs $(a,d)$ under which the channel capacity without states and interference can be achieved. These parameters fall in the shaded area above the line. It can be seen that as $d$ becomes larger (i.e., the correlation between the two states increases), the threshold on the parameter $a$ to fully cancel the interference and state becomes smaller. This suggests that more correlated states are easier to cancel together with the interference.
	\vspace{5mm}
	\begin{figure}[thb]
		\centering
		\includegraphics[height=3in,width=5in]{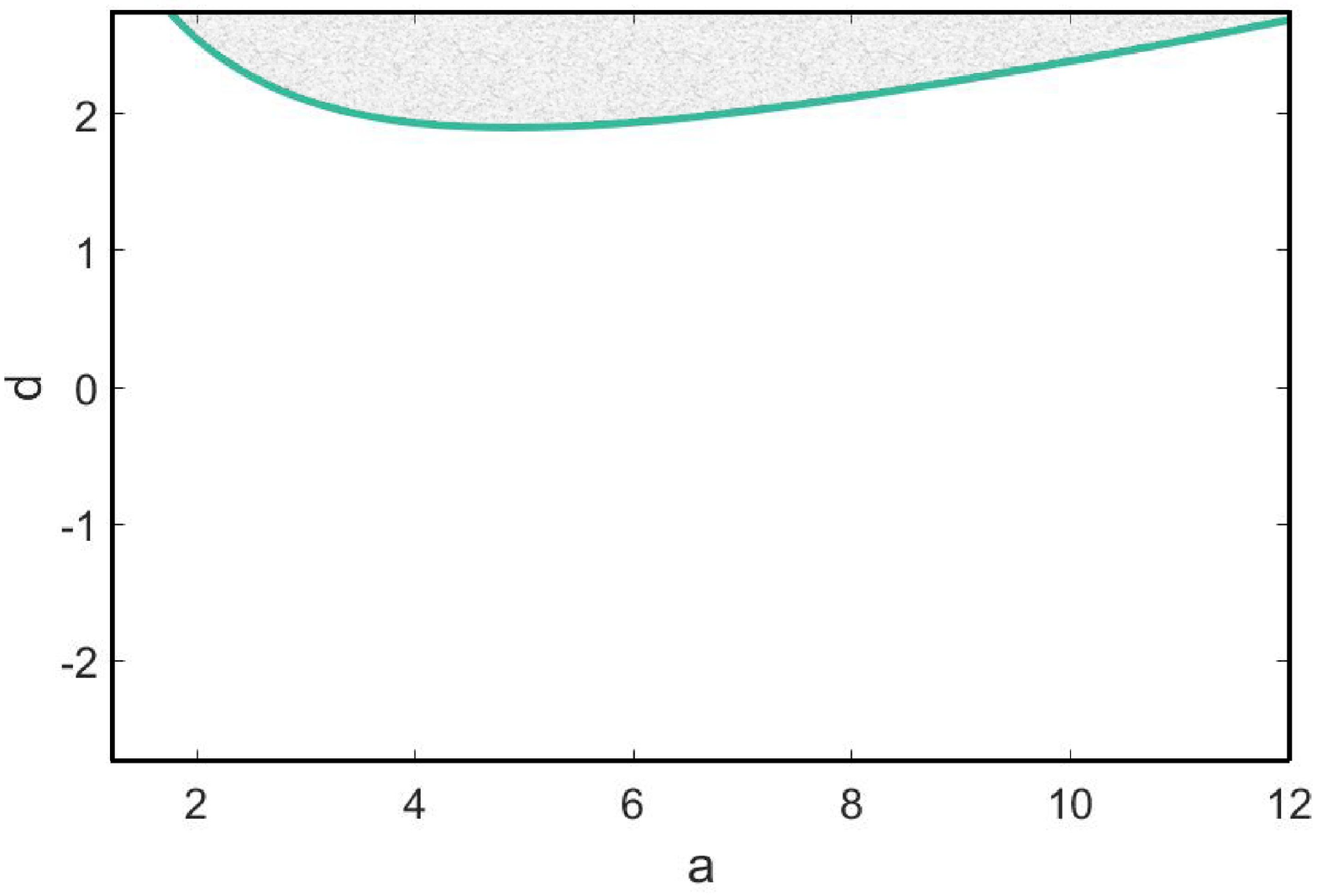}
		\caption{Characterization of channel parameters $(a,d)$ in shaded area under which the state-dependent Gaussian Z-IC achieves the capacity of the corresponding channel without states in very strong regime when $Q_2>\frac{1+P_2}{P_2}$.}\label{fig:vsac2}
	\end{figure}
	\vspace{5mm}
	
	
	Fig.~\ref{fig:vsac1} agrees with the result of the very strong IC without states in the sense that once $a$ is above a certain threshold (i.e., the interference is strong enough), then the point-to-point channel capacity without interference can be achieved. However, this is not always true for the {\em state-dependent} Z-IC. This can be seen from the condition \eqref{eq:zverystrongcond} in Theorem \ref{thr:VSCapacity}. If we let $a$ go to infinity, then the condition \eqref{eq:zverystrongcond} becomes $Q_2>\frac{1+P_2}{P_2}$, which is not always satisfied. This is because in the existence of state, $Y_1$ decodes $V$ instead of $X_2$, and the decoding rate is largest if the dirty paper coding design of $V$ (based on $S_2$ in receiver 2) also happens to be the same dirty paper coding design against $S_2$ in receiver 1. Clearly, as $a$ gets too large, $V$ is more deviated from such a favorable design, and hence the decoding rate becomes smaller, which consequently hurts the achievability of the point-to-point capacity for receiver 2. Such a phenomena can be observed in Fig.~\ref{fig:vsac2}, where the parameters $(a,d)$ under which the point-to-point channel capacity without interference and states can be achieved fall in the shaded area above the line. It can be seen that the constant $a$ cannot be too large to guarantee the achievability of the point-to-point channel capacity. Furthermore, the figure also suggests that further correlated states allow a larger range of $a$ under which the point-to-point channel capacity can be achieved. 	
	

	
	\section{Strong Interference Regime}\label{sec:zs}
	
	For the sake of technical convenience, in this section, we express $S_2$ as $S_2=cS_1+S_2'$, where $c$ is a constant representing the level of correlation, and $S_2'$ is independent from $S_1$ with $S_2^\prime\sim \mathcal{N}(0,Q_2^\prime)$ and $Q_2=c^2Q_1+Q_2^\prime$. Hence,  the channel model can be expressed in the following equivalent form,
	\begin{subequations}
		\begin{flalign}
		Y_1&=X_1+ aX_2+S_1+N_1\\
		Y_2&=X_2+cS_1+S_2^\prime+N_2.
		\end{flalign}
	\end{subequations}	
	
	It has been known that for the corresponding Z-IC without state which is strong but not very strong, i.e., $1 \leqslant a^2 < 1+P_1$,  the channel capacity contains rate pairs $(R_1,R_2)$ satisfying
	\begin{flalign}
	& R_1 +R_2 \leqslant \frac{1}{2} \log(1+P_1 +a^2P_2)\nn\\
	&R_1 \leqslant \frac{1}{2} \log{ (1+P_1)},\;\;\;R_2 \leqslant \frac{1}{2} \log{ (1+P_2)} \label{eq:StrongAchivGaussian}
	\end{flalign}
	which is illustrated as the pentagon O-A-B-E-F in Fig.~\ref{fig:scapa}.
	\vspace{5mm}
	\begin{figure}[thb]
		\centering
		\includegraphics[height=3in,width=5in]{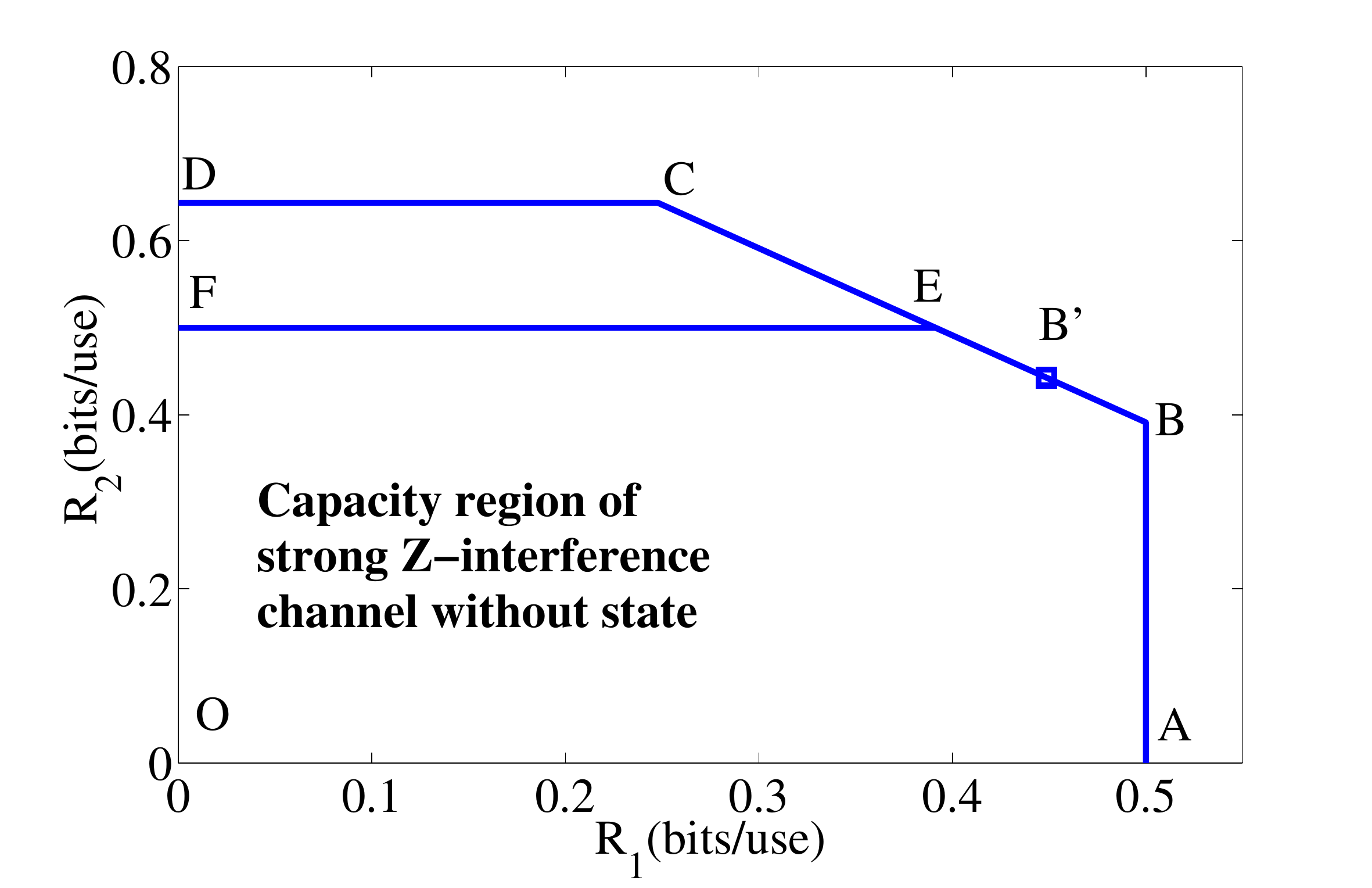}

		\caption{Capacity region of the strong Z-IC without state}\label{fig:scapa}

	\end{figure}
	\vspace{5mm}
	Our goal here is to study whether the points on the sum-capacity boundary of the Z-IC {\em without} state (i.e., the line B-E in Fig.~\ref{fig:scapa}) can be achieved for the corresponding {\em state-dependent} Z-IC. 
	Such a problem has been studied in \cite{Duan16IT} for the channel with two receivers corrupted by the same but differently scaled state. Here, we generalize such a study to the situation when the two receivers are corrupted by two correlated states.
	
	Since every point on this line can be achieved by rate splitting and successive cancellation in the case without state, for the state-dependent channel, we continue to adopt the idea of rate splitting and successive cancellation but using auxiliary random variables to incorporate dirty paper coding to further cancel state successively. More specifically, transmitter 1 splits its message $W_1$ into $W_{11}$ and $W_{12}$, and then encodes them into $U_1$ and $U_2$ respectively based on the Gel'fand-Pinsker binning scheme. Then transmitter 2 encodes its message $W_2$ into \textit{V}, based on the Gel'fand-Pinsker binning scheme.  The auxiliary random variables $U_1$, $U_2$, and $V$ are designed such that decoding of them at receiver 1 successively fully cancels the state corruption of $Y_1$ so that the sum capacity boundary (i.e., the line B-E) can be achieved if only decoding at receiver 1 is considered. Now further incorporating the decoding at receiver 2, if for any point on the line B-E, decoding of $V$ at receiver 2 does not cause further rate constraints, then such a point is achievable for the state-dependent Z-IC.  
	
	\begin{proposition}\label{pps:zstrongDMC}
		For the state-dependent Z-IC with states noncausally known at both transmitters, if the following condition is satisfied
		\begin{flalign}
		I(V;U_1Y_1)\leqslant I(V;Y_2),\label{eq:zscond}
		\end{flalign} 
		then an achievable region consists of rate pairs $(R_1,R_2)$ satisfying:
		\begin{subequations}
			\begin{flalign}
			R_1 \leqslant & I(U_1;Y_1)+I(U_2;VY_1|U_1) -I(S_1;U_1U_2)\label{zsas1}\\ 
			R_2 \leqslant &   I(V;U_1Y_1) -I(S_1;V) \label{zsas2}
			\end{flalign}
		\end{subequations}
		for some distribution $P_{S_1S_2}P_{V|S_1}P_{X_2|VS_1}P_{U_1|S_1}P_{U_2|S_1U_1}$ $P_{X_1|S_1U_1U_2}P_{Y_1|S_1X_1X_2}P_{Y_2|S_2X_2}$.
	\end{proposition}
	\begin{proof}
		See Section \ref{apx:zstrongDMC}	
	\end{proof}
	We note that although the above achievable rate region does not explicitly contain $S_2$, in fact $S_2$ implicitly affects the condition \eqref{eq:zscond} via $Y_2$. Furthermore, the correlation between $S_1$ and $S_2$ is also expected to affect the condition \eqref{eq:zscond} via $Y_2$, which is our major interest in the Gaussian case. 
	
	For the Gaussian model,  based on Proposition \ref{pps:zstrongDMC}, we characterize the condition under which any point on the sum capacity boundary of the strong Z-IC without states (e.g., point $B'$ in Fig. \ref{fig:scapa}) is achievable. Hence, such a point is on the sum capacity boundary of the state-dependent Z-IC.  
	
	\begin{theorem}\label{thr:ZSPointCapa}
		For the state-dependent Gaussian Z-IC with states noncausally known at both transmitters, if the channel parameters  $(a,c,P_1,P_2,Q_1,Q_2^\prime)$ satisfy the following condition:
		\begin{flalign}
		&\frac{a^2P_2(P_2+c^2Q_1+Q_2^\prime+1)}{(ac-\beta)^2Q_1P_2+(a^2P_2+\beta^2Q_1)(Q_2^\prime+1)}\geqslant 1+\frac{a^2P_2}{P_1^{\prime\prime}+1} \label{eq:ZSCond}
		\end{flalign}
		where $\beta=\frac{a^2P_2}{P_1+a^2P_2+1}$, then the following point (on the line B-E)
		\begin{flalign}
		&R_1=\frac{1}{2}\log\left(1+\frac{P_1'}{a^2P_2+P_1''+1}\right) + \frac{1}{2}\log\left(1+P_1''\right)\nn\\
		&R_2=\frac{1}{2}\log\left(1+\frac{a^2P_2}{P_1''+1}\right) \label{eq:InnerPointsStrong}
		\end{flalign}
		where $P_1'=P_1-P_1''$, is on the sum-capacity boundary. 
		\end {theorem}
		\begin{proof}
			Theorem \ref{thr:ZSPointCapa} follows from Proposition \ref{pps:zstrongDMC} by choosing the auxiliary random variables $U_1$, $U_2$ and $V$ as in the statement of the theorem. In particular, $U_1$ is first decoded by receiver 1, and is designed to cancel the state in $Y_1$ treating all other variables as noise. Then, $V$ is decoded by receiver 1, and is designed to cancel the state in $Y_1^{\prime}=Y_1-U_1=X_1''+aX_2+(c-\alpha_1)S_1+N_1$. Finally, $U_2$ is designed to cancel the state in $Y_1''=Y_1'-V=X_1''+(c-\alpha_1-\beta)S_1+N_1$. In order to satisfy the state cancellation requirements, $\alpha_1$, $\alpha_2$ and $\beta$ should satisfy
			\begin{flalign}
			&\alpha_1=\frac{P_1'}{P_1+a^2P_2+1}, \label{eq:sdirty_con1}\quad \quad\\
			&				\frac{\alpha_2}{1-\alpha_1}=\frac{P_1^{\prime\prime}}{P_1^{\prime\prime}+1}, \label{eq:sdirty_con2}\\
			&				\frac{\beta}{1-\alpha_1}=\frac{a^2P_2}{P_1^{\prime\prime}+a^2P_2+1},\label{eq:sdirty_con3}
			\end{flalign}
			which yields \eqref{eq:z_condition}. Substituting these choices of the random variables and the coefficients into Proposition \ref{pps:strong_inner},  \eqref{rate:pps_stronginner} becomes
			
			\begin{equation}
			\begin{aligned}
			R_{1}&\leqslant \min\left\{ I(U_1;Y_2)-I(U_1;S_1),\frac{1}{2}\log\left(1+\frac{P_1'}{a^2P_2+P_1''+1}\right)\right\}\\
			&+\min\left\{I(U_2;VY_2|U_1)-I(U_2;S_1|U_1), \frac{1}{2}\log\left(1+P_1''\right)\right\}\\
			R_{2}&\leqslant \min\left\{I(V;Y_2|U_1)-I(V;S_1), \frac{1}{2}\log\left(1+\frac{a^2P_2}{P_1''+1}\right)\right\}.
			\end{aligned}
			\end{equation}
			Hence, if the condition \eqref{cond:thr_stronginner} is satisfied, the points on the sum capacity boundary \eqref{rate:capIC} can be achieved.
		\end{proof} 
		
		Theorem \ref{thr:ZSPointCapa} provides the condition of channel parameters under which a certain given point is on the sum-capacity boundary of the capacity region. We next characterize a line segment on the sum-capacity boundary for a given set of channel parameters. 
		\begin{corollary}\label{cor:capacitySegment}
			For the state-dependent Z-IC with states noncausally known at both transmitters, if a point on the line $B-E$ in Fig.~\ref{fig:scapa} is on the sum-capacity boundary for a given set of channel parameters, then the segment between this point and point $B$ on the line $B-E$ is on the sum-capacity boundary for the same set of channel parameters.
		\end{corollary}
	
		\begin{figure}[thb]
			\centering
			\includegraphics[height=3in,width=5in]{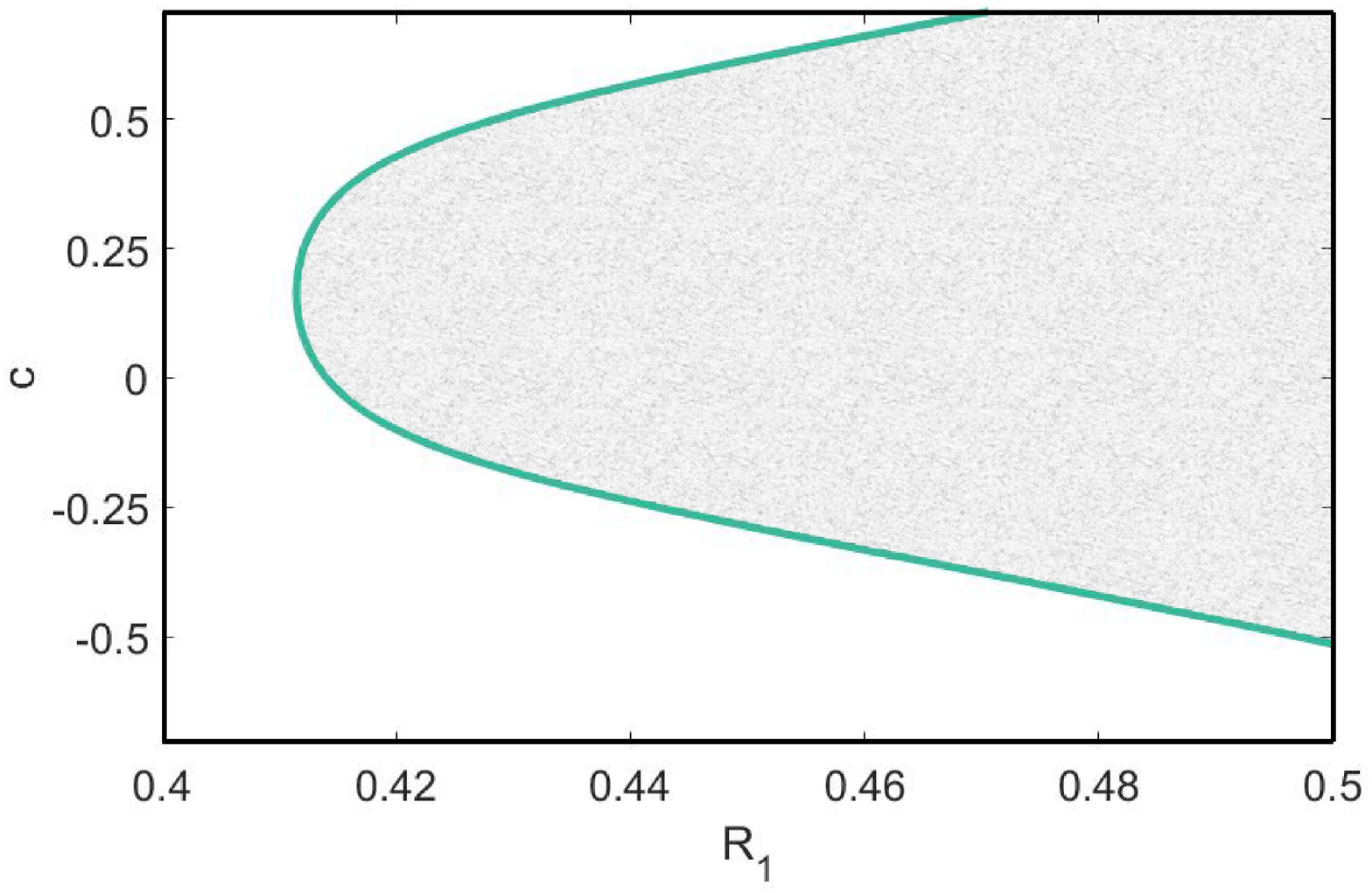}
			\caption{Ranges of $c$ under which points on sum-capacity boundary of the strong Z-IC without states can be achieved by the state-dependent Z-IC.}\label{fig:zsac}
		\end{figure}
		
		
		In order to numerically illustrate Theorem \ref{thr:ZSPointCapa}, we first note that each point on the sum-capacity boundary (i.e., the line B-E in  Fig. \ref{fig:scapa}) can be expressed as $(R_1,R_2)=(R_1, \frac{1}{2}\log(P_1+a^2P_2+1)-R_1)$. We now set $P_1=1$, $P_2=1$, $Q_1=2$, $Q_2=1$ and $a=1.2$, and hence $R_1 \in [ \frac{1}{2}\log(1.72), 0.5]$ parameterizes all points from point E to point B in Fig. \ref{fig:scapa}. In Fig.~\ref{fig:zsac}, we plot the ranges of $c$ under which points parameterized by $R_1$ on the sum capacity boundary of the strong Z-IC without states can be achieved by the state-dependent channel following Theorem \ref{thr:ZSPointCapa}. It can be seen that as the correlation between the two states (represented by $c$) increases, initially more points on the sum-capacity boundary are achieved and then  less points are achieved as $c$ is above a certain threshold. Thus, higher correlation does not guarantee more capability of achieving the sum-capacity boundary. This is because in our scheme $U_1$, $U_2$ and $V$ are specially designed for $Y_1$ based on dirty paper coding. At receiver 2, such design of $V$ initially approximates better the dirty paper coding design for $Y_2$ as $c$ becomes large, but then becomes worse as $c$ continues to increase, and hence decoding of $V$ at receiver 2 initially gets better and then becomes less capable, which consequently determines variation of achievability of the sum-capacity boundary. 
		
		
		\section{Weak Interference Regime}\label{sec:zweak}
		
		It has been shown in \cite{Sason04} that for the weak Gaussian Z-IC without state, i.e., $a^2\leqslant1$, the sum-capacity can be achieved by treating interference as noise at the interfered receiver. For the state-dependent Z-IC, if the two transmitters independently design dirty paper coding to cancel the state at their corresponding receivers, then the interference-free receiver achieves the capacity of the channel without state, and the interfered receiver (i.e., receiver 1) achieves the same rate as the channel without state by decoding its message treating the interference as noise. Thus, we obtain the following theorem.
		\begin{theorem}
			For the state-dependent Z-IC with states noncausally known at both transmitters, if $a^2\leqslant 1$,  the sum-capacity is given by
			\begin{flalign}\nn
			C_{\text{sum}} = \frac{1}{2}\log\left(1+\frac{P_1}{a^2P_2+1}\right)+ \frac{1}{2}\log\left(1+P_2\right).
			\end{flalign}
		\end{theorem}
		It can be seen that the sum-capacity achieving scheme does not depend on the correlation of the states, and hence, in the weak regime, the sum-capacity is not affected by the correlation of the states.
		

\newpage
\section{Technical Proofs}
			\subsection{Proof of Proposition \ref{pps:Z_inner}}\label{apx:Z inner}
We use random codes and fix the following joint distribution:
$$P_{S_1S_2UVX_1X_2Y_1Y_2}=P_{S_1S_2}P_{U|S_1S_2}P_{X_1|US_1S_2}P_{V|S_1S_2}P_{X_2|VS_1S_2}P_{Y_1Y_2|S_1S_2X_1X_2}$$

\begin{enumerate}
	\item Codebook Generation:
	\begin{itemize}
		\item Generate $2^{n(R_1+\tR_1)}$ codewords $U^n(w_1,l_1)$ with i.i.d.\ components based on $P_U$. Index these codewords by $w_1=1, \ldots, 2^{nR_1}, l_1 = 1, \ldots, 2^{n\tR_1}$.
		\item Generate $2^{n(R_2+\tR_2)}$ codewords $V^n(w_2,l_2)$ with i.i.d.\ components based on $P_V$. Index these codewords by $w_2=1, \ldots, 2^{nR_2}, l_2 = 1, \ldots, 2^{n\tR_2}$.
	\end{itemize}
	\item Encoding:
	\begin{itemize}
		
		\item Transmitter 1: Given $(s^n_1,s^n_2)$ and $w_1$, choose a $u^n(w_1,\tl_1)$ such that $$(u^n(w_1,\tl_1),s_1^n,s_2^n) \in T^n_\epsilon(P_{S_1S_2U})$$
		Otherwise, set $\tl_1=1$. It can be shown that for large n, such $u^n$ exists with high probability if 
		\begin{equation}
		\tR_1>I(U;S_1S_2). \label{eq:pps2-1}
		\end{equation}
		Then Generate $x^n_1$ with i.i.d. component based on $P_{X_1|US_1S_2}$ for transmission.
		
		\item Transmitter 2: Given $(s^n_1,s^n_2)$ and $w_2$, choose a $v^n(w_2,\tl_2)$ such that $$(v^n(w_2,\tl_2),s_1^n,s_2^n) \in T^n_\epsilon(P_{S_1S_2V})$$
		Otherwise, set $\tl_2=1$. It can be shown that for large n, such $v^n$ exists with high probability if 	
		\begin{equation}
		\tR_2>I(V;S_1S_2).
		\end{equation}
		Then Generate $x^n_2$ with i.i.d. component based on $P_{X_2|VS_1S_2}$ for transmission
	\end{itemize}
	
	\item Decoding:
	\begin{itemize}
		\item Decoder 1: Given $y^n_1$, find $(\hw_2, \hl_2)$ such that $$(v^n(\hw_2,\hl_2),y^n_1) \in T^n_\epsilon(P_{VY_1}).$$  If no or more than one such pair $(\hw_2, \hl_2)$ can be found, declare an error. It is easy to show that for sufficiently large n, we can correctly find such pair with high probability if 
		\begin{equation}
		R_2+\tR_2\leqslant I(V;Y_1).
		\end{equation}
		After decoding $v^n$, find the unique pair $(\hw_1, \hl_1)$ such that $$(u^n(\hw_1, \hl_1),v^n(w_2,\hl_2),y^n_1) \in T^n_\epsilon(P_{VUY_1})$$
		If no or more than one such pairs can be found, declare an error.  It is easy to show that for sufficiently large n, we can correctly find such pair with high probability if 
		\begin{equation}
		R_1+\tR_1\leqslant I(U;VY_1)
		\end{equation}
		
		\item Decoder 2: Given $y^n_2$, find $(\hw_2, \hl_2)$ such that $$(v^n(w_2,\hl_1),y^n_2) \in T^n_\epsilon(P_{UY_2}).$$ If no or more than one such pair $(\hw_2, \hl_2)$ can be found, declare an error. It is easy to show that for sufficiently large n, we can correctly find such pair with high probability if 
		\begin{equation}
		R_2+\tR_2\leqslant I(V;Y_2)\label{eq:pps2-2}
		\end{equation}
	\end{itemize}
\end{enumerate}
Proposition \ref{pps:Z_inner} is thus proved by combining \eqref{eq:pps2-1}-\eqref{eq:pps2-2}

\subsection{Proof of Proposition \ref{pps:zstrongDMC}}\label{apx:zstrongDMC}
We use random codes and fix the following joint distribution:
$$P_{S_1S_2U_1U_2VX_1X_2Y_1Y_2}=P_{S_1S_2}P_{V|S_1}P_{X_2|VS_1}P_{U_1|S_1}P_{U_2|S_1U_1}P_{X_1|U_1U_2S_1}P_{Y_1|S_1X_1X_2}P_{Y_2|S_2X_2}$$

\begin{enumerate}
	\item Codebook Generation:
	\begin{itemize}
		\item Generate $2^{n(R_{11}+\tR_{11})}$ codewords $U_1^n(w_{11},l_{11})$ with i.i.d.\ components based on $P_{U_1}$. Index these codewords by $w_{11}=1, \ldots, 2^{nR_{11}}, l_{11} = 1, \ldots, 2^{n\tR_{11}}$.
		\item For each $u_1^n(w_{11},l_{11})$, generate $2^{n(R_{12}+\tR_{12})}$ codewords $U_2^n(w_{11},l_{11},w_{12},l_{12})$ with i.i.d.\ components based on $P_{U_2|U_1}$. Index these codewords by $w_{12}=1, \ldots, 2^{nR_{12}}, l_{12} = 1, \ldots, 2^{n\tR_{12}}$.
		\item Generate $2^{n(R_2+\tR_2)}$ codewords $V^n(w_2,l_2)$ with i.i.d.\ components based on $P_V$. Index these codewords by $w_2=1, \ldots, 2^{nR_2}, v = 1, \ldots, 2^{n\tR_2}$.
	\end{itemize}
	\item Encoding:
	\begin{itemize}
		\item Transmitter 1: Given $s^n_1$ and $w_{11}$, choose a $u_1^n(w_{11},\tl_{11})$ such that $$(u^n(w_{11},\tl_{11}),s_1^n) \in T^n_\epsilon(P_{S_1U_{11}}).$$
		Otherwise, set $\tl_{11}=1$. It can be shown that for large n, such $u_1^n$ exists with high probability if 
		\begin{equation}
		\tR_{11}>I(U_1;S_1). \label{eq:pps3-1}
		\end{equation}
		Given $w_{11}$, $\tl_{11}$, $w_{12}$ and $s_1^n$, choose a $u_2^n(w_{11},\tl_{11},w_{12},\tl_{12})$ such that 
		$$(u_1^n(w_{11},\tl_{11}),u_2^n(w_{11},\tl_{11},w_{12},\tl_{12}),s_1^n) \in T^n_\epsilon(P_{S_1U_1U_2})$$
		Otherwise, set $\tl_{12}=1$. It can be shown that for large $n$, such $u_2^n$ exists with high probability if
		\begin{equation}
		\tR_{12}>I(U_2;S_1|U_1). \label{eq:pps3-2}
		\end{equation}
		Given $u_1^n(w_{11},\tl_{11})$, $u_2^n(w_{11},\tl_{11},w_{12},\tl_{12})$ and $s_1^n$, generate $x_1^n$ with i.i.d. components based on $P_{X_1|S_1U_1U_2}$
		
		\item Transmitter 2: Given $s^n_1$ and $w_2$, choose a $v^n(w_2,\tl_2)$ such that $$(v^n(w_2,\tl_2),s_1^n) \in T^n_\epsilon(P_{S_1V})$$
		Otherwise, set $\tl_2=1$. It can be shown that for large n, such $v^n$ exists with high probability if 	  
		\begin{equation}
		\tR_2>I(V;S_1).
		\end{equation}
		Then Generate $x^n_2$ with i.i.d. component based on $P_{X_2|VS_1}$ for transmission
	\end{itemize}
	
	\item Decoding:
	\begin{itemize}
		\item Decoder 1: Given $y^n_1$, find $(\hw_{11}, \hl_{11})$ such that $$(u_1^n(\hw_{11}, \hl_{11}),y^n_1) \in T^n_\epsilon(P_{U_1Y_1}).$$  If no or more than one such pair $(\hw_{11}, \hl_{11})$ can be found, declare an error. It is easy to show that for sufficiently large n, we can correctly find such pair with high probability if 
		\begin{equation}
		R_{11}+\tR_{11}\leqslant I(U_1;Y_1).
		\end{equation}
		After decoding $u_1^n$, find the unique pair $(\hw_{2}, \hl_{2})$ such that $$(u_1^n(\hw_{11}, \hl_{11}),v^n(w_2,\hl_2),y^n_1) \in T^n_\epsilon(P_{VU_1Y_1})$$
		If no or more than one such pairs can be found, declare an error.  It is easy to show that for sufficiently large n, we can correctly find such pair with high probability if 
		\begin{equation}
		R_2+\tR_2\leqslant I(V;Y_1|U_1)
		\end{equation}
		
		After successively decoding $v^n$, find the unique tuple $(w_{11},\tl_{11},w_{12},\tl_{12})$ such that
		$$((u_1^n(\hw_{11}, \hl_{11}),v^n(w_2,\hl_2),u_2^n(w_{11},\tl_{11},w_{12},\tl_{12}),y^n_1) \in T^n_\epsilon(P_{VU_1U_2Y_1}))$$
		If no or more than one such pairs can be found, declare an error.  It is easy to show that for sufficiently large n, we can correctly find such pair with high probability if 
		\begin{equation}
		R_{12}+\tR_{12}\leqslant I(U_2;VY_1|U_1)
		\end{equation}

		\item Decoder 2:  Given $y^n_2$, find $(\hw_{11}, \hl_{11})$ such that $$(u_1^n(\hw_{11}, \hl_{11}),y^n_2) \in T^n_\epsilon(P_{U_1Y_1}).$$  If no or more than one such pair $(\hw_{11}, \hl_{11})$ can be found, declare an error. It is easy to show that for sufficiently large n, we can correctly find such pair with high probability if 
		\begin{equation}
		R_{11}+\tR_{11}\leqslant I(U_1;Y_2).
		\end{equation}
		After decoding $u_1^n$, find the unique pair $(\hw_{2}, \hl_{2})$ such that $$(u_1^n(\hw_{11}, \hl_{11}),v^n(w_2,\hl_2),y^n_2) \in T^n_\epsilon(P_{VU_1Y_2})$$
		If no or more than one such pairs can be found, declare an error.  It is easy to show that for sufficiently large n, we can correctly find such pair with high probability if 
		\begin{equation}
		R_2+\tR_2\leqslant I(V;Y_2|U_1)
		\end{equation}
		
		After successively decoding $v^n$, find the unique tuple $(w_{11},\tl_{11},w_{12},\tl_{12})$ such that
		$$((u_1^n(\hw_{11}, \hl_{11}),v^n(w_2,\hl_2),u_2^n(w_{11},\tl_{11},w_{12},\tl_{12}),y^n_1) \in T^n_\epsilon(P_{VU_1U_2Y_2}))$$
		If no or more than one such pairs can be found, declare an error.  It is easy to show that for sufficiently large n, we can correctly find such pair with high probability if 
		\begin{equation}
		R_{12}+\tR_{12}\leqslant I(U_2;VY_2|U_1)
		\end{equation}
	\end{itemize}
	The corresponding achievable region is thus characterized by
	\begin{flalign}
	R_{11}&\leqslant \min\{I(U_1;Y_1), I(U_1;Y_2)\}-I(U_1;S_1)\label{eq:pps3-3}\\
	R_{12}&\leqslant \min\{I(U_2;VY_1|U_1), I(U_2;VY_2|U_1)\}-I(U_2;S_1|U_1)\label{eq:pps3-4}\\
	R_{2}&\leqslant \min\{I(V;Y_1|U_1), I(V;Y_2|U_1)\}-I(V;S_1)\label{eq:pps3-5}
	\end{flalign}
\end{enumerate}
Proposition \ref{pps:strong_inner} is thus proved by combining \cref{eq:pps3-1,eq:pps3-2} with \cref{eq:pps3-3,eq:pps3-4,eq:pps3-5}.

\chapter{State-Dependent Interference Channel with Correlated States}\label{cha:2state}
In this chapter, we study the state-dependent regular IC channel with correlated states. This state-dependent IC is different from the regular IC without state as each receiver is corrupted by a channel state and the two transmitters know the information of both channel states noncausally. These two states are correlated, so that this model generalizes the models with independent states and with the same but differently scaled states. Comparing to the model studied in Chapter \ref{chap:z2state}, we need to develop a more sophisticated scheme to cancel the channel states. Furthermore, this model brings us more insights on the impact of the correlation between the states on the capacity region.

The rest of chapter is organized as follows. First, we describe the channel model. Second, we study the model in the very strong interference regime and characterize the channel parameters under which the two receivers achieve their corresponding point-to-point channel capacity without state and interference. Then, we study the model in strong but not very strong interference regime and characterize the sum capacity boundary partially under certain channel parameters based on joint design of rate splitting, successive cancellation, as well as dirty paper coding. Finally, we study the model in the weak interference regime and characterize the sum capacity, which is achieved by the two transmitters independently designing dirty paper coding and treating interference as noise.

\section {Channel Model}\label{sec:icmodel}

\begin{figure}[thb]
	\centering
	\includegraphics[width=4in]{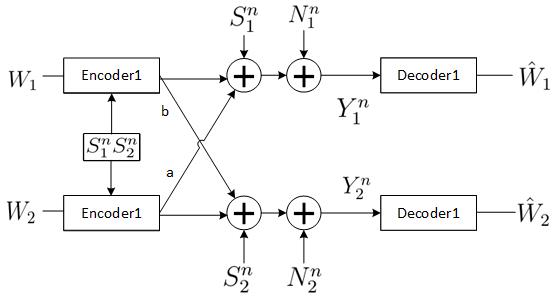}
	\caption{The state-dependent IC}\label{fig:inferencechannel}
\end{figure}

We consider the state-dependent IC (as shown in Fig.~\ref{fig:inferencechannel}), in which transmitters 1 and 2 send messages $W_{1}$ and $W_{2}$ respectively to the receivers 1 and 2. For $k=1,2$, encoder $k$ maps the message $w_k\in \cW_k$ to a codeword $x_k^n\in \cX_k^n$. The two inputs $x_1^n$ and $x_2^n$ are then transmitted over the IC to the receivers, which are corrupted by two correlated state sequences $S_1^n$ and $S_2^n$, respectively. The state sequences are known to both the transmitters noncausally, but are unknown at the receivers. Encoders 1 and 2 want to map their messages as well as the state sequences' information into codewords $x_1^n\in\cX_1^n$ and $x_2^n\in\cX_2^n$. The channel transition probability is given by $P_{Y_1Y_2|S_1S_2X_1X_2}$. The decoders at the receivers map the received sequences $y_1^n$ and $y_2^n$ into corresponding messages $\hw_k\in \cW_k$ for $k=1,2$.

The average probability of error for a length-$n$ code is defined as
\begin{flalign}\label{PE}
P_e^{(n)} = & \frac{1}{|\cW_1||\cW_2|}\sum_{w_1=1}^{|\cW_1|}\sum_{w_2=1}^{|\cW_2|} Pr\lbrace(\hat{w}_1, \hat{w}_2) \neq (w_1, w_2)\rbrace.
\end{flalign}
A rate pair $(R_1, R_2)$ is {\em achievable} if there exist a sequence of message sets $\cW_{k}^{(n)}$ with $|\cW_{k}^{(n)}|=2^{nR_k}$ for $k=1, 2$, such that the average error probability $P_e^{(n)} \rightarrow 0$ as $n \to \infty$. The {\em capacity region} is defined to be the closure of the set of all achievable rate pairs $(R_1, R_2)$.

In this dissertation, we study the Gaussian channel with the outputs at the two receivers for one channel use given by
\begin{subequations}
	\begin{flalign}
	Y_1&=X_1+ aX_2+S_1+N_1,\\
	Y_2&=bX_1+X_2+S_2+N_2
	\end{flalign}
\end{subequations}	
where $a$ and $b$ are the channel gain coefficients, and $N_1$ and $N_2$ are noise variables with Gaussian distributions $N_1 \sim \mathcal{N}(0,1)$ and $N_2 \sim \mathcal{N}(0,1)$. The state variables $S_1$ and $S_2$ are jointly Gaussian with the correlation coefficient $\rho$ and the marginal distributions $S_1 \sim \mathcal{N}(0,Q_1)$ and $S_2\sim \mathcal{N}(0,Q_2)$. Both the noise variables and the state variables are i.i.d. over the channel uses. The channel inputs $X_1$ and $X_2$ are subject to the average power constraints $P_1$ and $P_2$.

Our goal is to characterize channel parameters, under which the capacity of the corresponding IC without the presence of the state can be achieved, and thus the capacity region of the IC with the presence of state is also established. In particular, we are interested in understanding the impact of the correlation between the states $S_1$ and $S_2$ on the capacity characterization.

\section {Very Strong Interference Regime}\label{sec:icvs}

In this section, we study the impact of the correlation between states on the characterization of the capacity in the very strong regime, where the channel parameters satisfy
\begin{subequations}\nn
	\begin{flalign}
	P_1+a^2P_2+1&>(1+P_1)(1+P_2),\\
	b^2P_1+P_2+1&>(1+P_1)(1+P_2).
	\end{flalign}
\end{subequations}	
For the corresponding IC without states, the capacity region contains rate pairs ($R_1,R_2$) satisfying:
\begin{equation}\label{cap:VeryStrong}
\begin{aligned}
R_1 \leqslant & \frac{1}{2}\log(1+P_1),\\
R_2 \leqslant & \frac{1}{2}\log(1+P_2).
\end{aligned}
\end{equation}	
In this case, the two receivers achieve the point-to-point channel capacity without interference. Furthermore, in \cite{Duan16IT}, an achievable scheme has been established to achieve the same point-to-point channel capacity when the two receivers are corrupted by the same but differently scaled state. Our focus here is on the more general scenario, where the two receivers are corrupted by two {\em correlated} states, and our aim is to understand how the correlation affects the design of the scheme.


We first design an achievable scheme to obtain an achievable rate region for the discrete memoryless IC. The two transmitters encode their messages $W_1$ and $W_2$ into two auxiliary random variables $U$ and $V$, respectively, based on the Gel'fand-Pinsker binning scheme. Since the channel satisfies the very strong interference condition, it is easier for a receiver to decode the information of the interference. Thus either receiver first decodes the auxiliary random variable corresponding to the message intended for the other receiver, and then decodes its own message by decoding the auxiliary random variable for itself. For instance, receiver 1 first decodes $V$, then uses it to cancel the interference $X_2$ and partial state interference, and finally decodes its own message $W_1$ by decoding $U$. Differently from \cite{Duan16IT}, two auxiliary random variables $U$ and $V$ are designed not with regard to one state, but with regard to two correlated states. This requires a joint design for $U$ and $V$ to fully cancel the states. Based on such a scheme, we obtain the following achievable region.

\begin{proposition}\label{pps:IC inner}
	For the state-dependent IC with states noncausally known at both transmitters, the achievable region consists of rate pairs $(R_1,R_2)$ satisfying:
	\begin{subequations}
		\begin{flalign}
		R_1 \leqslant & \min\{ I(U;VY_1),I(U;Y_2)\}-I(S_1S_2;U),\label{eq:pps1-1} \\
		R_2 \leqslant & \min\{ I(V;UY_2),I(V;Y_1)\}-I(S_1S_2;V)\label{eq:pps1-2}
		\end{flalign}
	\end{subequations}
	for some distribution $P_{S_1S_2}P_{U|S_1S_2}P_{X_1|US_1S_2}P_{V|S_1S_2}P_{X_2|VS_1S_2}P_{Y_1Y_2|S_1S_2X_1X_2}$, where $U$ and $V$ are auxiliary random variables.
\end{proposition}
\begin{proof}
	See Section \ref{apx:IC inner}.
\end{proof}

We now study the Gaussian IC. For the sake of technical convenience, we express the Gaussian channel in Section \ref{sec:icvs} in a different form. Since $S_1$ and $S_2$ are jointly Gaussian, $S_1$ can be expressed as $S_1=dS_2+S_1^\prime$ where $d$ is a constant representing the level of correlation, and $S_1'$ is independent from $S_2$ and $S_1'\sim \mathcal{N}(0, Q_1')$ with $Q_1=d^2Q_2+Q_1^\prime$. Thus, without loss of generality, the channel model can be expressed in the following equivalent form that is more convenient for analysis.
\begin{subequations}
	\begin{flalign}
	Y_1&=X_1+ aX_2+dS_2+S_1^\prime+N_1,\\
	Y_2&=bX_1+X_2+S_2+N_2.
	\end{flalign}
\end{subequations}	

Following Proposition \ref{pps:IC inner}, we characterize the condition under which both the state and interference can be fully canceled, and hence the capacity region for the state-dependent Gaussian IC in the very strong regime is obtained.

\begin{theorem}\label{thr:IC inner}
	For the state-dependent Gaussian IC with state noncausally known at both transmitters,the capacity region for  is the same as the point-to-point channel capacity for both receivers, If the channel parameters  $(a,b,d,P_1,P_2,Q_1^\prime,Q_2)$ satisfy the following conditions:
	\begin{subequations}
		\begin{flalign}
		\frac{1}{2}\log(1+P_1)\leqslant& h(X_1)-h(U,Y_2)+h(Y_2)\label{eq:cond1}\\
		\frac{1}{2}\log(1+P_2)\leqslant& h(X_2)-h(V,Y_1)+h(Y_1)\label{eq:cond2}
		\end{flalign}
	\end{subequations}
	where the auxiliary random variables are designed as $U=X_1+\alpha_1S_1^\prime+\alpha_2S_2$ and $V=X_2+\beta_1S_1^\prime+\beta_2S_2$. Here, $X_1$,$X_2$, $ S_1^\prime $ and $ S_2 $ are independent Gaussian variables withe mean zero and variances $P_1$,$P_2$, $ Q_1 $ and $ Q_2 $, respectively. The parameters $\alpha_1$,$\alpha_2$,$\beta_1$ and $\beta_2$ are set as
	\begin{equation}\label{eq:s_variablesetting}
	\begin{aligned}
	\alpha_1&=\frac{P_1(1+P_2)}{(P_1+1)(P_2+1)-abP_1P_2},\\ \alpha_2&=\frac{P_1(d+dP_2-aP_2)}{(P_1+1)(P_2+1)-abP_1P_2},\\
	\beta_1&=\frac{bP_1P_2}{(P_1+1)(P_2+1)-abP_1P_2},\\ \beta_2&=\frac{P_2(P_1+1-bdP_1)}{(P_1+1)(P_2+1)-abP_1P_2}.
	\end{aligned}	
	\end{equation}
	\end {theorem}
	
	\begin{proof}
		The proof mainly follows Proposition \ref{pps:IC inner}. As discussed in the proof of Proposition \ref{pps:IC inner}, $V$ is first decoded by decoder 1 and $U$ is first decoded by decoder 2. And then  receiver 2 subtracts $U$ to cancel $X_1$ and obtain  $Y_2^{\prime}=Y_2-bU=X_2-b\alpha_1S_1^\prime+(1-b\alpha_2)S_2+N_2$, and receiver 1 subtracts $V$ to cancel $X_2$ and obtain  $Y_1^{\prime}=Y_1-aV=X_1+(1-a\beta_1)S_1^\prime+(d-a\beta_2)S_2+N_1$. In order to fully cancel the channel states for $Y_1^{\prime}$ and $Y_2^{\prime}$, based on the dirty paper coding scheme in \cite{Costa83}, we further require the coefficients to satisfy the following conditions,
		\begin{subequations}
			\begin{flalign}
			\frac{\alpha_1}{1-a\beta_1}&=\frac{\alpha_2}{d-a\beta_2}\label{eq:dirty_con1}\\
			\frac{\alpha_1}{1-a\beta_1}&=\frac{P_1}{P_1+1}\label{eq:dirty_con2}\\
			\frac{\beta_1}{-b\alpha_1}&=\frac{\beta_2}{1-b\alpha_2}\label{eq:dirty_con3}\\
			\frac{\beta_1}{-b\alpha_1}&=\frac{P_2}{P_2+1}\label{eq:dirty_con4}
			\end{flalign}
		\end{subequations}		
		which yield $\alpha_1$,$\alpha_2$,$\beta_1$ and $\beta_2$ in \eqref{eq:s_variablesetting}.
		
		By plugging these parameters into \eqref{eq:pps1-1}, we obtain $$I(U;VY_1)-I(S_1,S_2;U)= \frac{1}{2}\log(1+P_1),$$
		which yields $$R_1 \leqslant \min\{I(U;Y_2)-I(S_1,S_2;U), \frac{1}{2}\log(1+P_1)\}.$$

		Similarly, \eqref{eq:pps1-2} yields $$R_2 \leqslant \min\{I(V;Y_1)-I(S_1,S_2;V), \frac{1}{2}\log(1+P_2)\}.$$
		In order to achieve the channel capacity of the point-to-point channel as shown in \eqref{cap:VeryStrong} for both receivers, the following conditions should be satisfied:
		\begin{subequations}
			\begin{flalign}
			\frac{1}{2}\log(1+P_1)&\leqslant I(U;Y_2)-I(S_1,S_2;U)\label{eq:thr1_con1}\\
			\frac{1}{2}\log(1+P_2)&\leqslant I(V;Y_1)-I(S_1,S_2;V) \label{eq:thr1_con2}.
			\end{flalign}
		\end{subequations}	
	\end{proof}	
	We note that the conditions in Theorem \ref{thr:IC inner} represent the comparison between the ability of receivers to decode messages in different decoding steps. For instance, in condition \eqref{eq:cond1} the right-hand side term represents how much receiver 2 can decode $U$ in the first step of decoding in order to cancel the interference, and the left-hand side term represents the rate at which receiver 1 can decode $U$ in the second step of decoding, where we can use the dirty paper coding scheme to fully cancel the states and achieve the capacity. Hence, achieving the point-to-point channel capacity requires the second step to be bottleneck.
	
	We next study the impact of the channel parameters and state correlation on the achievablility of the point-to-point capacity. In particular, we illustrate how the interference gains $(a,b)$ affect the conditions (\ref{eq:cond1}) and (\ref{eq:cond2}).

	%
	\begin{figure}[H]
	\centering
	\begin{tabular}{cc} 
		\includegraphics[width=3in]{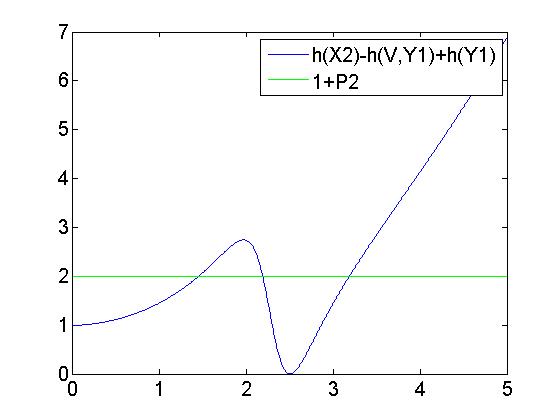}
		&\includegraphics[width=3in]{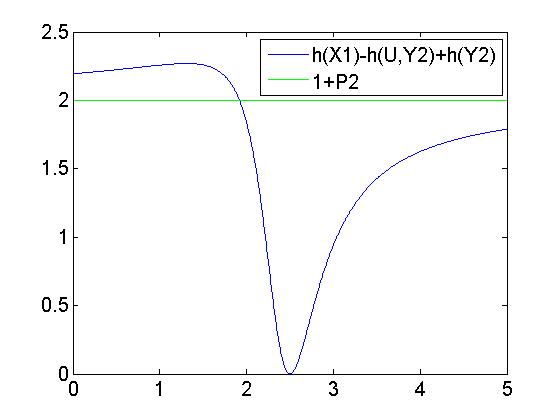}\\
		$h(X_1)-h(U,Y_2)+h(Y_2)$ versus b&	$h(X_2)-h(V,Y_1)+h(Y_1)$ versus b\\
	\end{tabular}
\vspace{0.1cm}
		\large $d=0.99$\\
		\begin{tabular}{cc} 
		\includegraphics[width=3in]{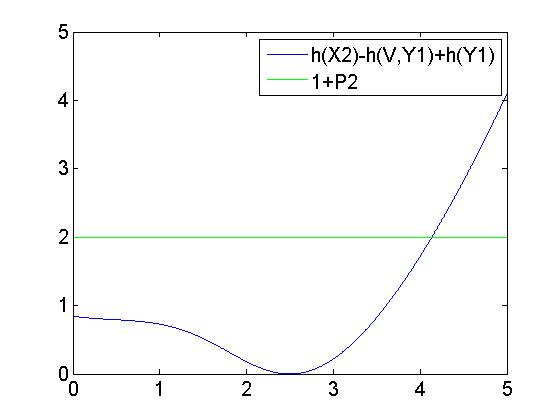}
		&\includegraphics[width=3in]{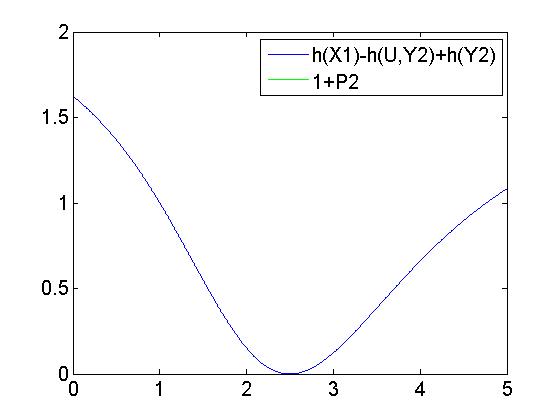}\\
			$h(X_1)-h(U,Y_2)+h(Y_2)$ versus b&	$h(X_2)-h(V,Y_1)+h(Y_1)$ versus b\\
	\end{tabular}
\vspace{0.1cm}
$d=0.5$\\
\begin{tabular}{cc} 
	\includegraphics[width=3in]{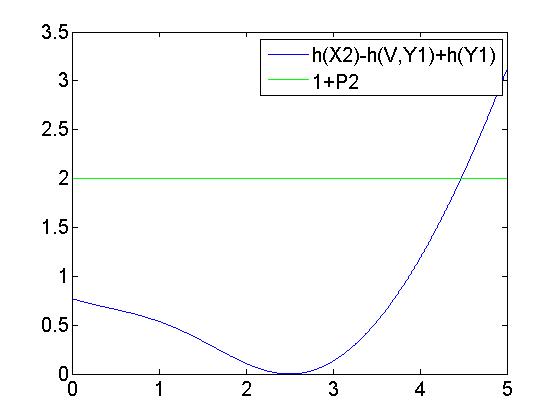}
	&\includegraphics[width=3in]{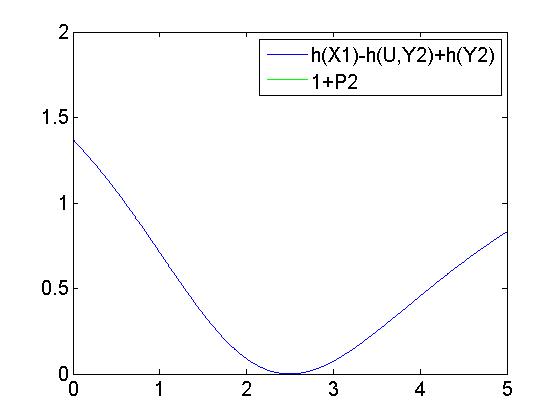}\\
		$h(X_1)-h(U,Y_2)+h(Y_2)$ versus b&	$h(X_2)-h(V,Y_1)+h(Y_1)$ versus b\\
	\end{tabular}
\vspace{0.1cm}
	$d=0.1$\\
	\caption{ Conditions \eqref{eq:cond1} and \eqref{eq:cond2} changing with b}\label{fig:b_beta}
\end{figure}
	
	In Fig. \ref{fig:b_beta}, we set $Q_1=Q_2=0.9, P_1=1,P_2=1$ and $a=1.6$, and plot the right side terms in  (\ref{eq:cond1}) and  (\ref{eq:cond2}), $h(X_1)-h(U,Y_2)+h(Y_2)$ and $h(X_2)-h(V,Y_1)+h(Y_1)$, versus the channel parameters $ b $ for three different values of $ d $. Taking the first row of Fig. \ref{fig:b_beta} as an example, it is clear that $\frac{1}{2}\log(1+P_2)$ is a straight line, and $h(X_1)-h(U,Y_2)+h(Y_2)$ is not a monotone function with respect to $ b $. The condition (\ref{eq:cond1}) is satisfied only when $h(X_1)-h(U,Y_2)+h(Y_2)$ is above the straight line $\frac{1}{2}\log(1+P_2)$. When the parameter $d=0.99$, there are two regions over which the condition (\ref{eq:cond1}) is satisfied. But if $d=0.5$, there is only one region over which the condition (\ref{eq:cond1}) is satisfied. For $d=0.1$, the condition (\ref{eq:cond1}) is not satisfied for any $b$. Similarly, the second row in Fig. \ref{fig:b_beta} illustrates the regions of $b$ over which the condition (\ref{eq:cond1}) is satisfied for the corresponding values of $ d $. Then the intersection of the region of $ b $ in the first and second rows of Fig. \ref{fig:b_beta} fully determines the ranges of $ b $ over which the point-to-point channel capacity can be achieved.
	
	\begin{figure}[H]
		\caption{Different d}\label{fig:c}
		\centering
			\includegraphics[width=3in]{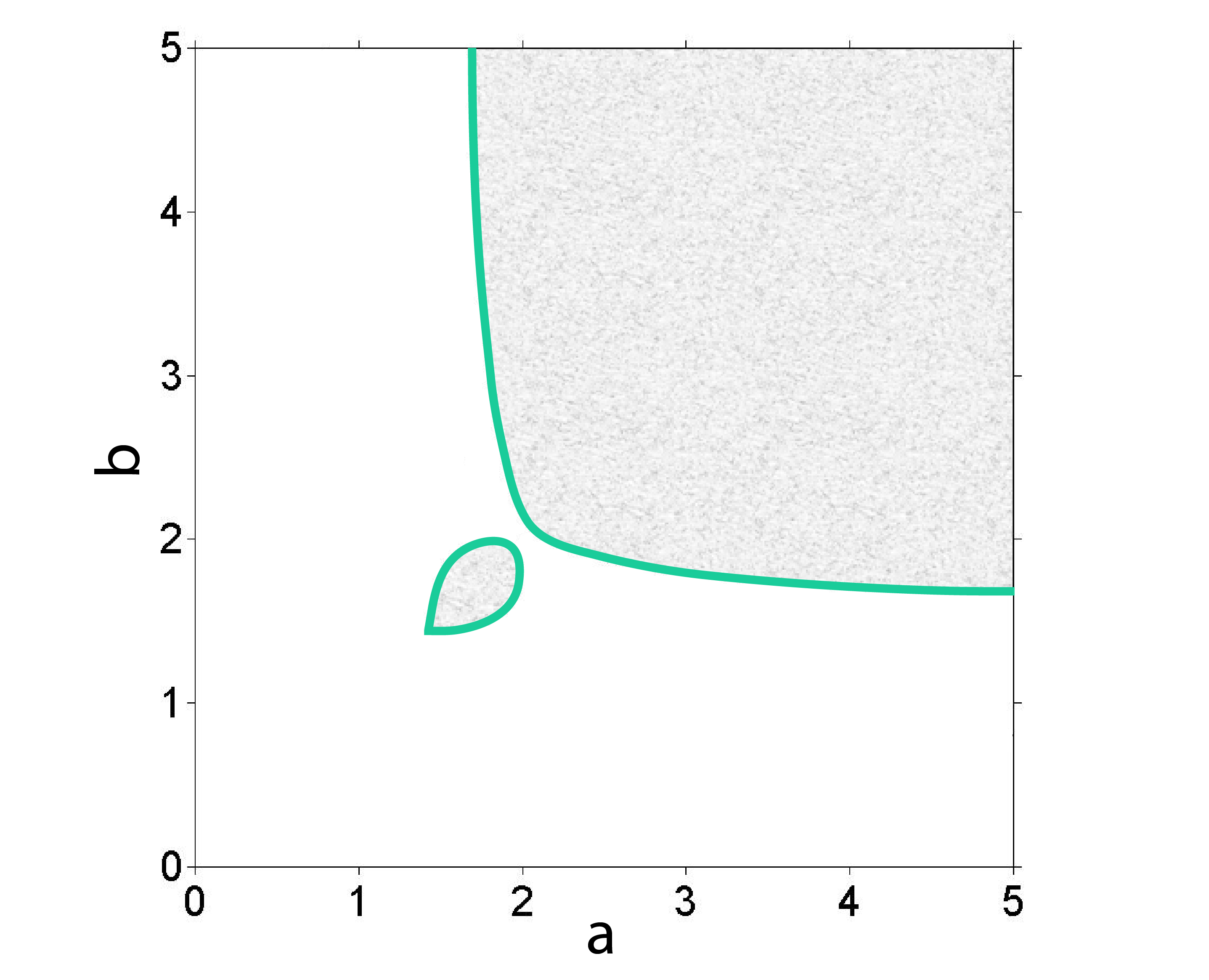}\\
			$d=0.99$\\
			\includegraphics[width=3in]{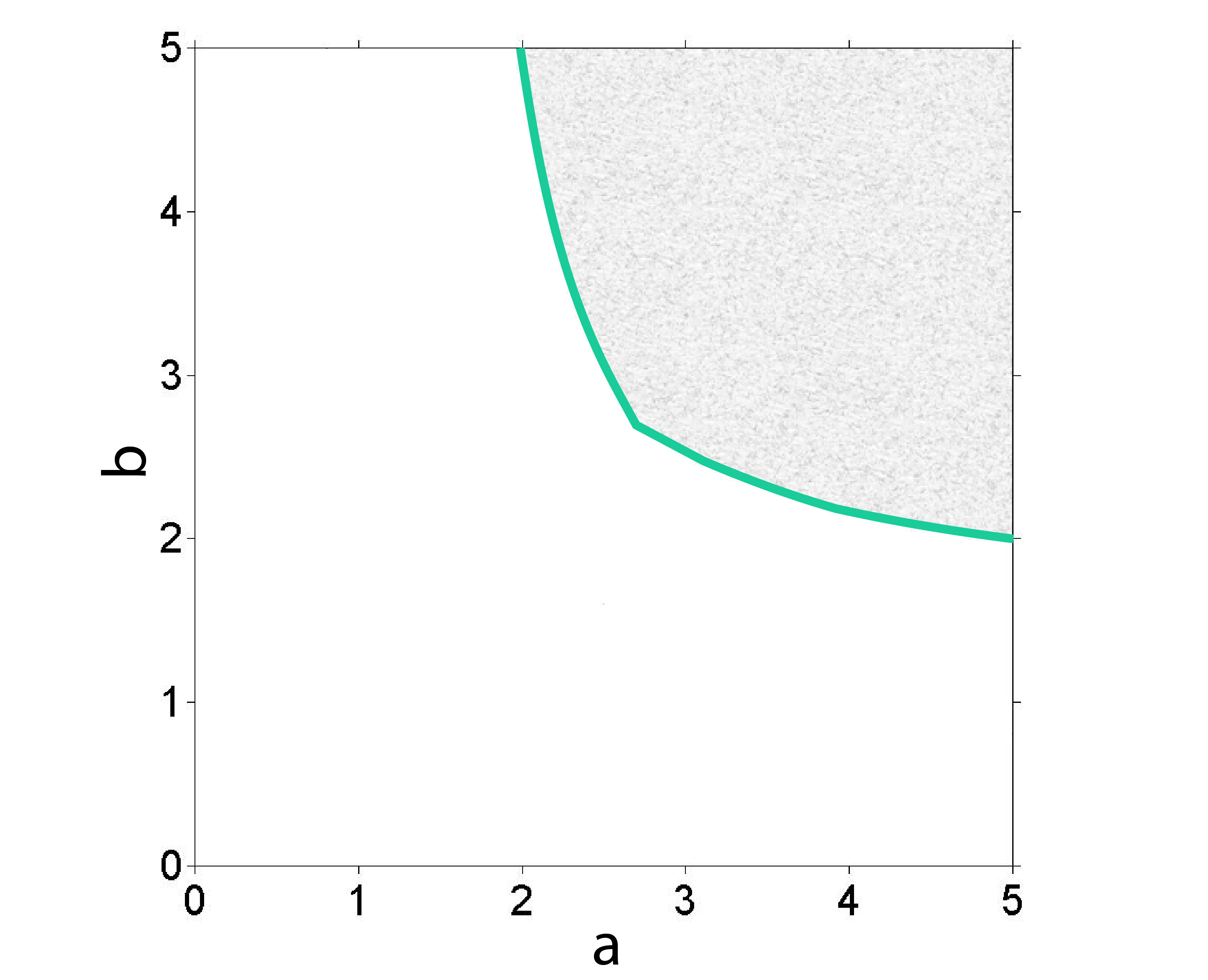}\\
			$d=0.5$\\
			\includegraphics[width=3in]{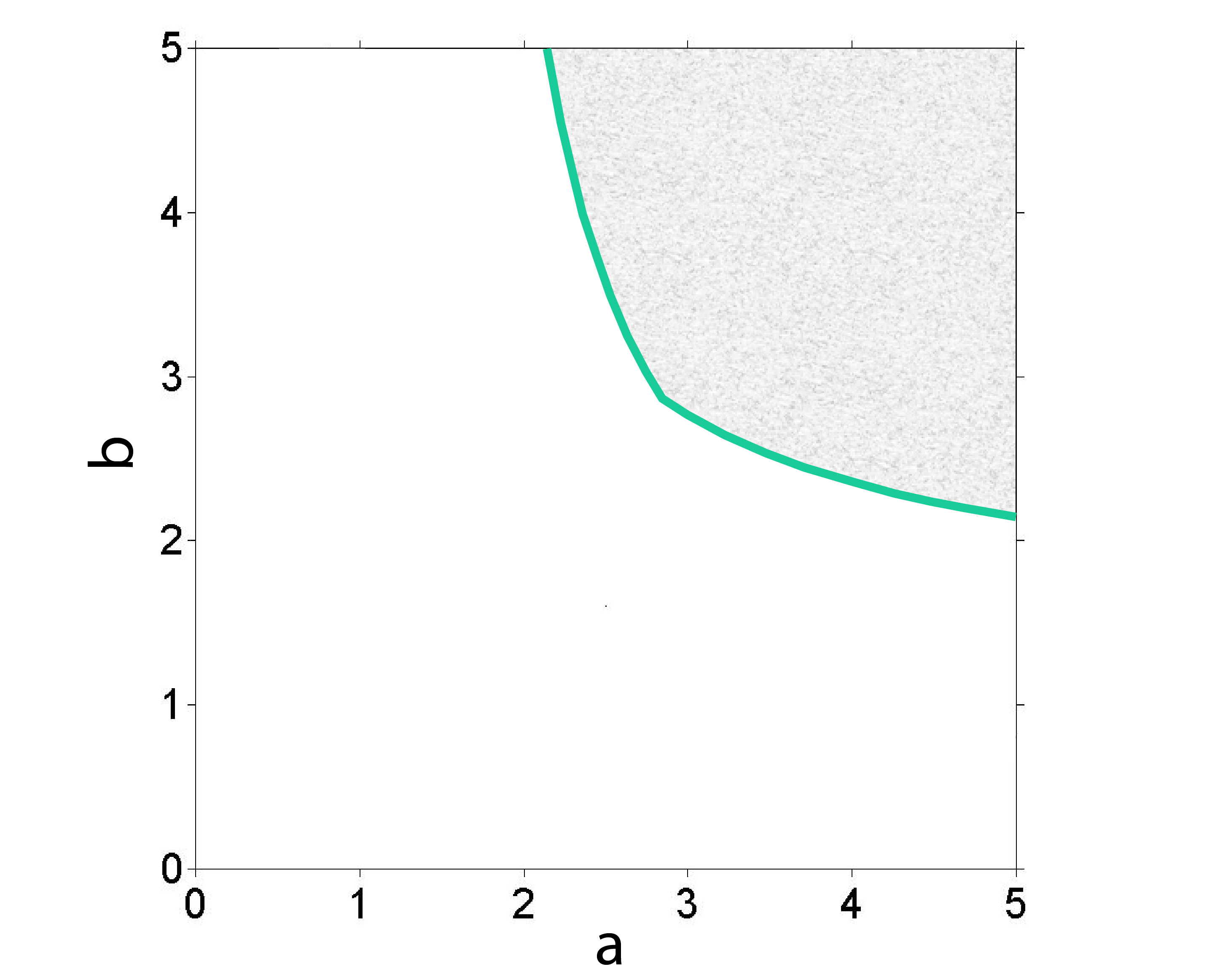}\\
			$d=0.1$
	\end{figure}	
	The range of the parameters $(a,b)$ such that the point-to-point channel capacity is obtained is shown in Fig. \ref{fig:c}. The x-axis and y-axis correspond to the parameters a and b, respectively. For these figures, if we fix $a=1.6$, then the ranges of b is consistent with those in Fig. \ref{fig:c} where both (\ref{eq:cond1})and (\ref{eq:cond2}) are satisfied.

	Fig. \ref{fig:c} also illustrates the impact of the correlation $d$ between the states $S_1$ and $S_2$ on the achievability of channel capacity. It is clear that as $d$ increases, i.e., the two states are more correlated, the range of $(a,b)$ over which the point-to-point channel capacity is achieved gets larger. This confirms the intuition that more correlated states are easier to be fully canceled.
	
		
		\section{Strong Interference Regime}\label{sec:ics}
		In this section, we study the state dependent IC in the strong regime, which excludes the very strong interference regime that has been studied in Section \ref{sec:icvs}. It has been shown in \cite{Sato81} that for the corresponding IC without state, which is strong but not very strong, i.e., the channel parameters satisfy
		\begin{flalign}
		&a\geqslant 1, \ \ \ \ b\geqslant 1,\\\nn
		&\min\{P_1+a^2P_2+1,b^2P_1+P_2+1\}\leqslant(1+P_1)(1+P_2),
		\end{flalign}
		where, without loss of generality, we assume that $P_1+a^2P_2+1 \leqslant b^2P_1+P_2+1$. The capacity region contains rate pairs  $(R_1,R_2)$ satisfying
		Hence, only one sum-rate bound is left, and the capacity region for strong IC without states contain rate pair $(R_1,R_2)$, which was characterized in \cite{Sato81}, satisfying 				
		\begin{flalign}\nn
		&	R_1 \leqslant \frac{1}{2}\log(1+P_1) ,\ \ \ \	R_2 \leqslant \frac{1}{2}\log(1+P_2),\\
		&R_1+R_2 \leqslant \frac{1}{2}\log(P_1+a^2P_2+1).\label{eq:strong_cap}
		\end{flalign}
		
		\vspace{5mm}
		\begin{figure}[H]
			\centering
			\includegraphics[width=4in]{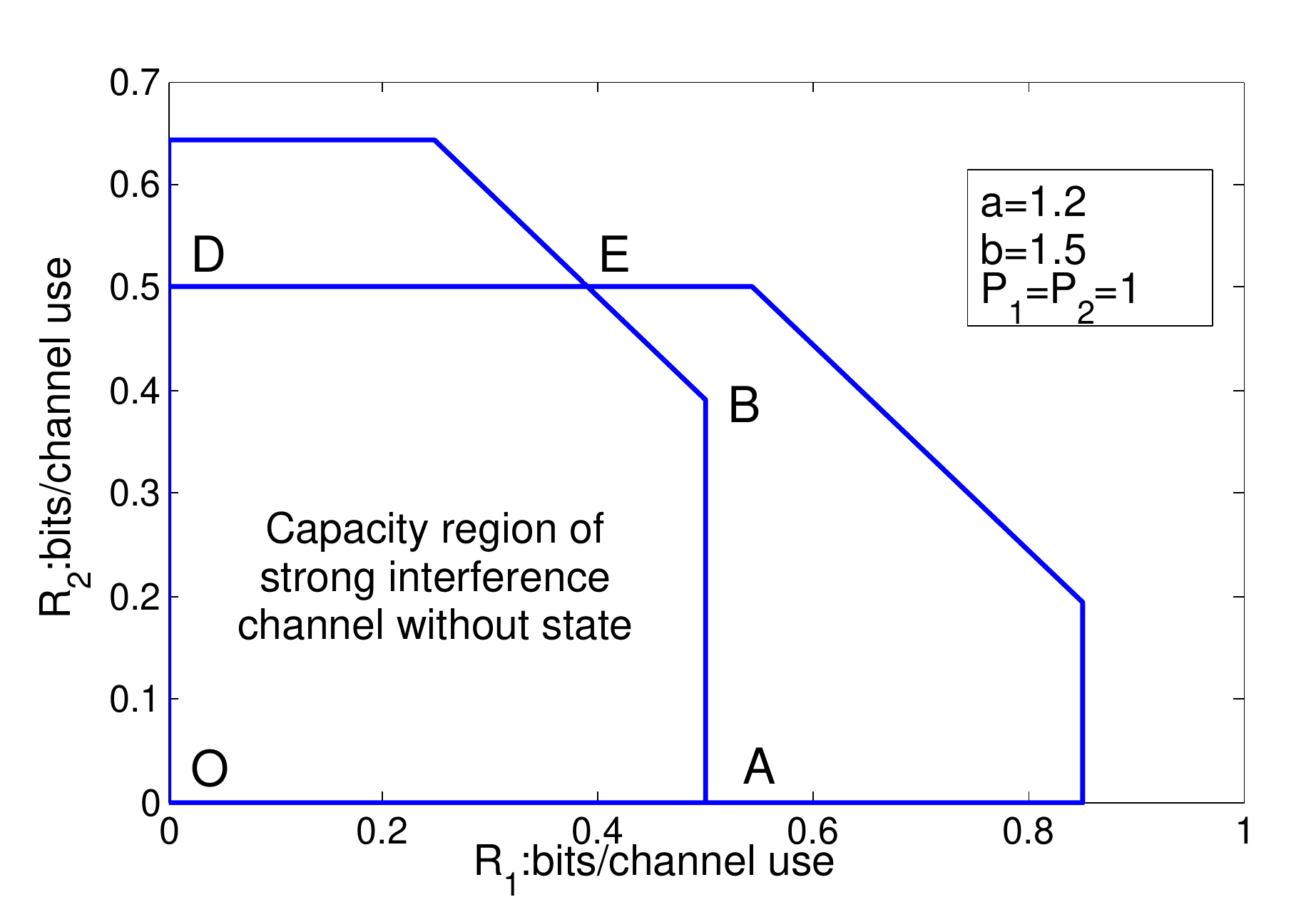}
			\caption{Capacity region of the strong IC without state}\label{fig:sreg}
		\end{figure}
	\vspace{5mm}	
		Such a region of is an intersection of the capacity regimes of two MACs, which is illustrated as the pentagon O-A-B-E-D-O in Fig. \ref{fig:sreg}.
		Our goal here is to study whether the points on sum-rate capacity boundary of the IC without state can be achieved. Such a problem has been studied in \cite{Duan16IT} for the channel with two receivers corrupted by the same but differently scaled state. Here, we generalize such a study to the situation when the two receivers are corrupted by two correlated states. 
		
		Since every point on this line of the sum-rate capacity can be achieved by rate splitting and successive cancellation in the case without state, for the state-dependent channel, we continue to adopt the idea of rate splitting and successive cancellation but using auxiliary random variables to incorporate dirty paper coding to further cancel state successively. More specifically, transmitter 1 splits its message $W_1$ into $W_{11}$ and $W_{12}$, and then encodes them into $U_1$ and $U_2$ respectively based on the Gel'fand-Pinsker binning scheme. Then transmitter 2 encodes its message $W_2$ into \textit{V}, based on the Gel'fand-Pinsker binning scheme.  The auxiliary random variables $U_1$, $U_2$, and $V$ are designed such that decoding of them at receiver 1 successively fully cancels the state corruption of $Y_1$ so that the sum capacity boundary (i.e., the line B-E) can be achieved if only decoding at receiver 1 is considered. Now further incorporating the decoding at receiver 2, if for any point on the line B-E, decoding of $V$ at receiver 2 does not cause further rate constraints, then such a point is achievable for the state-dependent IC.  
		
		\begin{proposition}\label{pps:strong_inner}
			For the state-dependent IC with states noncausally known at both transmitters, an achievable region consists of rate pairs $(R_1,R_2)$ satisfying:
			\begin{equation}
			\begin{aligned}\label{rate:pps_stronginner}
			R_{1}&\leqslant \min\{I(U_1;Y_1), I(U_1;Y_2)\}\\
			&+\min\{I(U_2;VY_1|U_1), I(U_2;VY_2|U_1)\}-I(U_1U_2;S_1)\\
			R_{2}&\leqslant \min\{I(V;Y_1|U_1), I(V;Y_2|U_1)\}-I(V;S_1)
			\end{aligned}
			\end{equation}
			for some distribution $P_{S_1S_2}P_{V|S_1}P_{X_2|VS_1}P_{U_1|S_1}P_{U_2|S_1U_1}P_{X_1|S_1U_1U_2}P_{Y_1|S_1X_1X_2}P_{Y_2|S_2X_2}$, where $U_1$, $U_2$ and $V$ are auxiliary random variables.
		\end{proposition}
		\begin{remark}
			This scheme can be generalized through further splitting the messages and changing the orders of decoding the messages at the two receivers. The achievable region can then be obtained by taking the convex hull of the union over all achievable regions for different scheme above.
		\end{remark}
		We note that although the above achievable rate region does not explicitly contain $S_2$, in fact $S_2$ implicitly affects the condition \eqref{eq:zscond} via $Y_2$. Furthermore, the correlation between $S_1$ and $S_2$ is also expected to affect the condition \eqref{eq:zscond} via $Y_2$, which is our major interest in the Gaussian case. 
		
		Based on Proposition \ref{pps:strong_inner}, we next characterize partial boundary of the capacity region for the state-dependent Gaussian IC. For the sake of technical convenience, we express the Gaussian model in a different form. In particular, we express $S_2$ as $S_2=cS_1+S_2^\prime$ where $c$ is a constant representing the level of correlation, and $S_1$ is independent from $S_2^\prime$ and $S_2^\prime\sim \mathcal{N}(0, Q_2')$ with $Q_2=c^2Q_1+Q_2^\prime$. Thus, without loss of generality, the channel model can be expressed in the following equivalent form that is more convenient for analysis here.
		\begin{subequations}
			\begin{flalign}
			Y_1&=X_1+ aX_2+S_1+N_1\\
			Y_2&=bX_1+X_2+cS_1+S_2^\prime+N_2.
			\end{flalign}
		\end{subequations}
		We next show that we can design a scheme to achieve the partial boundary of the capacity region for the IC without state. We note that the rate on the sum-capacity boundary can be characterized by
		\begin{equation}\label{rate:capIC}
		\begin{aligned}
		&R_1=\frac{1}{2}\log\left(1+\frac{P_1'}{a^2P_2+P_1''+1}\right) + \frac{1}{2}\log\left(1+P_1''\right),\\
		&R_2=\frac{1}{2}\log\left(1+\frac{a^2P_2}{P_1''+1}\right),
		\end{aligned}
		\end{equation}
		for some $P_1^{\prime}$, $P_1^{\prime\prime}\geqslant 0$, and $P_1^{\prime}+P_1^{\prime\prime}\leqslant P_1$.
		\begin{theorem} \label{thr:SPointCapa}
			Any rate point in \eqref{rate:capIC} can be achieved by the state-dependent IC if the channel parameters satisfy the following conditions
			\begin{equation}\label{cond:thr_stronginner}
			\begin{aligned}
			I(U_1;Y_2)-I(U_1;S_1)&\leqslant\frac{1}{2}\log\left(1+\frac{P_1'}{a^2P_2+P_1''+1}\right)\\
			I(U_2;VY_2|U_1)-I(U_2;S_1|U_1)&\leqslant \frac{1}{2}\log\left(1+P_1''\right)\\
			I(V;Y_2|U_1)-I(V;S_1)&\leqslant \frac{1}{2}\log\left(1+\frac{a^2P_2}{P_1''+1}\right),
			\end{aligned}
			\end{equation}
			where the mutual information terms are calculated by setting $U_1=X_1^\prime+\alpha_1S_1$, $U_2=X_1^{\prime\prime}+\alpha_2S_1$, $V=aX_2+\beta S_1$ and $X_1=X_1^\prime+X_1^{\prime\prime}$ where $X_1^\prime$, $X_1^{\prime\prime}$ and $X_2$ are Gaussian variables with mean zero and variances $P_1^\prime$, $P_1^{\prime\prime}$ and $P_2$, and $\alpha_1$,$\alpha_2$ and $\beta$ are given by
			\begin{equation} \label{eq:z_condition}
			\begin{aligned}
			\alpha_1&=\frac{P_1^\prime}{P_1+a^2P_2+1} \ \ \\\ \alpha_2&=\frac{P_1^{\prime\prime}}{P_1+a^2P_2+1}\\
			\beta&=\frac{a^2P_2}{P_1+a^2P_2+1}.
			\end{aligned}
			\end{equation}
		\end{theorem}
		\begin{proof}
			Theorem \ref{thr:SPointCapa} follows from Proposition \ref{pps:strong_inner} by choosing the auxiliary random variables $U_1$, $U_2$ and $V$ as in the statement of the theorem. In particular, $U_1$ is first decoded by receiver 1, and is designed to cancel the state in $Y_1$ treating all other variables as noise. Then, $V$ is decoded by receiver 1, and is designed to cancel the state in $Y_1^{\prime}=Y_1-U_1=X_1''+aX_2+(c-\alpha_1)S_1+N_1$. Finally, $U_2$ is designed to cancel the state in $Y_1''=Y_1'-V=X_1''+(c-\alpha_1-\beta)S_1+N_1$. In order to satisfy the state cancellation requirements, $\alpha_1$, $\alpha_2$ and $\beta$ should satisfy
			\begin{flalign}
			&\alpha_1=\frac{P_1'}{P_1+a^2P_2+1}, \label{eq:sdirty_con1}\quad \quad\\
			&				\frac{\alpha_2}{1-\alpha_1}=\frac{P_1^{\prime\prime}}{P_1^{\prime\prime}+1}, \label{eq:sdirty_con2}\\
			&				\frac{\beta}{1-\alpha_1}=\frac{a^2P_2}{P_1^{\prime\prime}+a^2P_2+1},\label{eq:sdirty_con3}
			\end{flalign}
			which yields \eqref{eq:z_condition}. Substituting these choices of the random variables and the coefficients into Proposition \ref{pps:strong_inner},  \eqref{rate:pps_stronginner} becomes
			
			\begin{equation}
			\begin{aligned}
			R_{1}&\leqslant \min\left\{ I(U_1;Y_2)-I(U_1;S_1),\frac{1}{2}\log\left(1+\frac{P_1'}{a^2P_2+P_1''+1}\right)\right\}\\
			&+\min\left\{I(U_2;VY_2|U_1)-I(U_2;S_1|U_1), \frac{1}{2}\log\left(1+P_1''\right)\right\}\\
			R_{2}&\leqslant \min\left\{I(V;Y_2|U_1)-I(V;S_1), \frac{1}{2}\log\left(1+\frac{a^2P_2}{P_1''+1}\right)\right\}.
			\end{aligned}
			\end{equation}
			Hence, if the condition \eqref{cond:thr_stronginner} is satisfied, the points on the sum capacity boundary \eqref{rate:capIC} can be achieved.
		\end{proof}

			\begin{figure}[H]
				\centering
				\includegraphics[height=2.5in,width=4in]{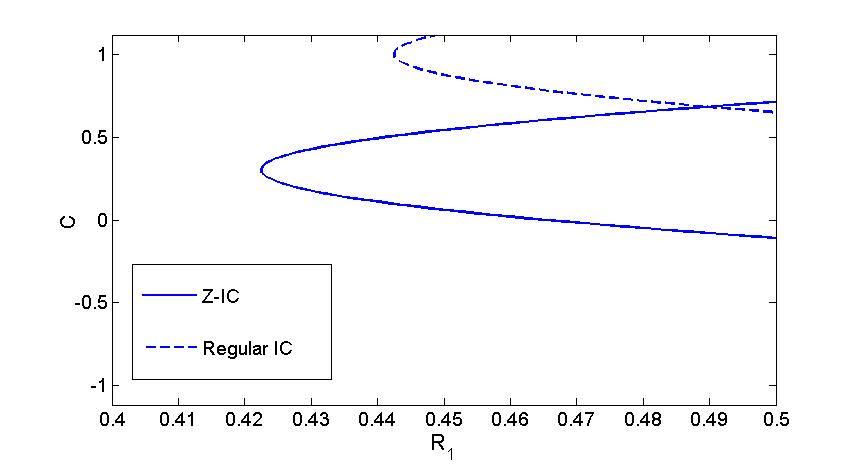}
				\caption{Ranges of $c$ under which points on sum capacity boundary of the strong IC and Z-IC without state can be achieved by the state-dependent IC and Z-IC.}\label{fig:sac}
			\end{figure}
			
			
		 In Fig.~\ref{fig:sac}, we plot the ranges of $c$ under which points, parameterized by $R_1$ on the sum capacity boundary of the strong Z-IC without state, can be achieved by the state-dependent Z-IC following Theorem \ref{thr:ZSPointCapa}. It can be seen that as correlation between the two states (represented by $c$) increases, initially more points on the sum capacity boundary are achieved and then  less points are achieved as $c$ is above a certain threshold. Thus, higher correlation does not guarantee more capability of achieving the sum capacity boundary. This is because in our scheme both $U_i$ and $V$ are specially designed for $Y_1$ based on dirty paper coding. At receiver 2, such design of $V$ initially approximates better the dirty paper coding design for $Y_2$ as $c$ becomes large, but then becomes worse as $c$ continues to increase, and hence decoding of $V$ at receiver 2 initially gets better and then becomes less capable, which consequently determines variation of achievability of the sum capacity boundary. 
			
			Fig. \ref{fig:scapa} also plots the same parameter range for the state-dependent IC as characterized by Theorem \ref{thr:SPointCapa}. It is clear that the state-dependent IC achieves a smaller line segment on the sum-capacity (i.e., smaller range of $R_1$). This is reasonable, because Theorem \ref{thr:SPointCapa} for the IC requires more conditions than Theorem \ref{thr:ZSPointCapa} for the Z-IC. Fig. \ref{fig:scapa} also demonstrates that large value of c(i.e., higher correlation between the states) is required for the IC to achieve the sum capacity than the Z-IC. This is because the dirty paper coding is designed with respect to receiver 1. High correlation between states helps such design to be more effective to cancel that states at receiver 2 as well.
			
			\section{Weak Interference Regime}\label{sec:icweak}
			In this section, we study the state-dependent IC and ZIC in the weak interference regime. The channel parameters for the IC in this regime satisfy  $|a(1+b^2P_1)|+|b(1+a^2P_2)|\leqslant1$, which reduces to $a\leqslant1$ for the Z-IC. It has been shown in \cite{Shang09,Anna09,Mota09}, for the weak IC without state and in \cite{Sason04} for the weak Z-IC that the sum capacity can be achieved by treating interference as noise. It was further shown in \cite{Duan16IT} that for the IC and Z-IC with the same but differently scaled state at two receivers, independent dirty paper coding at the two transmitters to cancel the states treating interference as noise achieves the same sum capacity. We node that such a scheme is also achievable in our model, which thus yields the following Corollary.
			
			\begin{corollary}(A direct result following \cite{Duan16IT})\label{thr:weak_IC_ZIC}
				For the state-dependent IC with states noncausally known at both transmitters, if $|a(1+b^2P_1)|+|b(1+a^2P_2)|\leqslant1$,  the sum capacity is given by
				\begin{flalign}\nn
				C_{sum} = \frac{1}{2}\log\left(1+\frac{P_1}{a^2P_2+1}\right)+ \frac{1}{2}\log\left(1+\frac{P_2}{b^2P_1+1}\right).
				\end{flalign}
			\end{corollary}
			
			It can be seen that the sum capacity achieving scheme does not depend on the correlation of the states, and hence, in the weak regime, the sum capacity is not affected by the correlation of the states.
			
\newpage
			\section{Technical Proofs}
			\subsection{Proof of Proposition \ref{pps:IC inner}}\label{apx:IC inner}
			We use random codes and fix the following joint distribution:
			$$P_{S_1S_2UVX_1X_2Y_1Y_2}=P_{S_1S_2}P_{U|S_1S_2}P_{X_1|US_1S_2}P_{V|S_1S_2}P_{X_2|VS_1S_2}P_{Y_1Y_2|S_1S_2X_1X_2}$$

			\begin{enumerate}
				\item Codebook Generation:
				\begin{itemize}
					\item Generate $2^{n(R_1+\tR_1)}$ codewords $U^n(w_1,l_1)$ with i.i.d.\ components based on $P_U$. Index these codewords by $w_1=1, \ldots, 2^{nR_1}, l_1 = 1, \ldots, 2^{n\tR_1}$.
					\item Generate $2^{n(R_2+\tR_2)}$ codewords $V^n(w_2,l_2)$ with i.i.d.\ components based on $P_V$. Index these codewords by $w_2=1, \ldots, 2^{nR_2}, l_2 = 1, \ldots, 2^{n\tR_2}$.
				\end{itemize}
				\item Encoding:
				\begin{itemize}	
					\item Transmitter 1: Given $(s^n_1,s^n_2)$ and $w_1$, choose a $u^n(w_1,\tl_1)$ such that $$(u^n(w_1,\tl_1),s_1^n,s_2^n) \in T^n_\epsilon(P_{S_1S_2U})$$ Otherwise, set $\tl_1=1$. It can be shown that for large n, such $u^n$ exists with high probability if 
					\begin{equation}
					\tR_1>I(U;S_1S_2)。 \label{eq:pps1}
					\end{equation}
					Then Generate $x^n_1$ with i.i.d. component based on $P_{X_1|US_1S_2}$ for transmission.					
					\item Transmitter 2: Given $(s^n_1,s^n_2)$ and $w_2$, choose a $v^n(w_2,\tl_2)$ such that $$(v^n(w_2,\tl_2),s_1^n,s_2^n) \in T^n_\epsilon(P_{S_1S_2V})$$
					Otherwise, set $\tl_2=1$. It can be shown that for large n, such $v^n$ exists with high probability if 	
					\begin{equation}
					\tR_2>I(V;S_1S_2).
					\end{equation}
					Then Generate $x^n_2$ with i.i.d. component based on $P_{X_2|US_1S_2}$ for transmission
				\end{itemize}

				\item Decoding:
				\begin{itemize}
					\item Decoder 1: Given $y^n_1$, find $(\hw_2, \hl_2)$ such that $$(v^n(w_2,\hl_2),y^n_1) \in T^n_\epsilon(P_{VY_1}).$$  If no or more than one such pair $(\hw_2, \hl_2)$ can be found, declare an error. It is easy to show that for sufficiently large n, we can correctly find such pair with high probability if 
					\begin{equation}
					R_2+\tR_2\leqslant I(V;Y_1).
					\end{equation}
					After decoding $v^n$, find the unique pair $(\hw_1, \hl_1)$ such that $$(u^n(\hw_1, \hl_1),v^n(w_2,\hl_2),y^n_1) \in T^n_\epsilon(P_{VUY_1})$$
					If no or more than one such pairs can be found, declare an error.  It is easy to show that for sufficiently large n, we can correctly find such pair with high probability if 
					\begin{equation}
					R_1+\tR_1\leqslant I(U;VY_1)
					\end{equation}
					
					\item Decoder 2: Given $y^n_2$, find $(\hw_1, \hl_1)$ such that $$(v^n(w_1,\hl_1),y^n_2) \in T^n_\epsilon(P_{UY_2}).$$  If no or more than one such pair $(\hw_1, \hl_1)$ can be found, declare an error. It is easy to show that for sufficiently large n, we can correctly find such pair with high probability if 
					\begin{equation}
					R_1+\tR_1\leqslant I(U;Y_2)
					\end{equation}
					After decoding $u^n$, find the unique pair $(\hw_2, \hl_2)$ such that $$(v^n(\hw_2, \hl_2),u^n(w_1,\hl_1),y^n_2) \in T^n_\epsilon(P_{VUY_2})$$
					If no or more than one such pairs can be found, declare an error.  It is easy to show that for sufficiently large n, we can correctly find such pair with high probability if 
					\begin{equation}
					R_2+\tR_2\leqslant I(V;UY_2)\label{eq:pps2}
					\end{equation}
				\end{itemize}
			\end{enumerate}
			Proposition \ref{pps:IC inner} is thus proved by combining \eqref{eq:pps1}-\eqref{eq:pps2}

			\subsection{Proof of Proposition \ref{pps:strong_inner}}\label{apx:strong_inner}
			We use random codes and fix the following joint distribution:
			$$P_{S_1S_2U_1U_2VX_1X_2Y_1Y_2}=P_{S_1S_2}P_{V|S_1}P_{X_2|VS_1}P_{U_1|S_1}P_{U_2|S_1U_1}P_{X_1|U_1U_2S_1}P_{Y_1|S_1X_1X_2}P_{Y_2|S_2X_2}$$
			
			\begin{enumerate}
				\item Codebook Generation:
				\begin{itemize}
					\item Generate $2^{n(R_{11}+\tR_{11})}$ codewords $U_1^n(w_{11},l_{11})$ with i.i.d.\ components based on $P_{U_1}$. Index these codewords by $w_{11}=1, \ldots, 2^{nR_{11}}, l_{11} = 1, \ldots, 2^{n\tR_{11}}$.
					\item For each $u_1^n(w_{11},l_{11})$, generate $2^{n(R_{12}+\tR_{12})}$ codewords $U_2^n(w_{11},l_{11},w_{12},l_{12})$ with i.i.d.\ components based on $P_{U_2|U_1}$. Index these codewords by $w_{12}=1, \ldots, 2^{nR_{12}}, l_{12} = 1, \ldots, 2^{n\tR_{12}}$.
					\item Generate $2^{n(R_2+\tR_2)}$ codewords $V^n(w_2,l_2)$ with i.i.d.\ components based on $P_V$. Index these codewords by $w_2=1, \ldots, 2^{nR_2}, v = 1, \ldots, 2^{n\tR_2}$.
				\end{itemize}
				\item Encoding:
				\begin{itemize}
					\item Transmitter 1: Given $s^n_1$ and $w_{11}$, choose a $u_1^n(w_{11},\tl_{11})$ such that $$(u^n(w_{11},\tl_{11}),s_1^n) \in T^n_\epsilon(P_{S_1U_{11}}).$$
					Otherwise, set $\tl_{11}=1$. It can be shown that for large n, such $u_1^n$ exists with high probability if 
					\begin{equation}
					\tR_{11}>I(U_1;S_1). \label{eq:pps4-1}
					\end{equation}
					Given $w_{11}$, $\tl_{11}$, $w_{12}$ and $s_1^n$, choose a $u_2^n(w_{11},\tl_{11},w_{12},\tl_{12})$ such that 
					$$(u_1^n(w_{11},\tl_{11}),u_2^n(w_{11},\tl_{11},w_{12},\tl_{12}),s_1^n) \in T^n_\epsilon(P_{S_1U_1U_2})$$
					Otherwise, set $\tl_{12}=1$. It can be shown that for large $n$, such $u_2^n$ exists with high probability if
					\begin{equation}
					\tR_{12}>I(U_2;S_1|U_1). \label{eq:pps4-2}
					\end{equation}
					Given $u_1^n(w_{11},\tl_{11})$, $u_2^n(w_{11},\tl_{11},w_{12},\tl_{12})$ and $s_1^n$, generate $x_1^n$ with i.i.d. components based on $P_{X_1|S_1U_1U_2}$
					
					\item Transmitter 2: Given $s^n_1$ and $w_2$, choose a $v^n(w_2,\tl_2)$ such that $$(v^n(w_2,\tl_2),s_1^n) \in T^n_\epsilon(P_{S_1V})$$
					Otherwise, set $\tl_2=1$. It can be shown that for large n, such $v^n$ exists with high probability if 	  
					\begin{equation}
					\tR_2>I(V;S_1).
					\end{equation}
					Then Generate $x^n_2$ with i.i.d. component based on $P_{X_2|VS_1}$ for transmission
				\end{itemize}
				
				\item Decoding:
				\begin{itemize}
					\item Decoder 1: Given $y^n_1$, find $(\hw_{11}, \hl_{11})$ such that $$(u_1^n(\hw_{11}, \hl_{11}),y^n_1) \in T^n_\epsilon(P_{U_1Y_1}).$$  If no or more than one such pair $(\hw_{11}, \hl_{11})$ can be found, declare an error. It is easy to show that for sufficiently large n, we can correctly find such pair with high probability if 
					\begin{equation}
					R_{11}+\tR_{11}\leqslant I(U_1;Y_1).
					\end{equation}
					After decoding $u_1^n$, find the unique pair $(\hw_{2}, \hl_{2})$ such that $$(u_1^n(\hw_{11}, \hl_{11}),v^n(w_2,\hl_2),y^n_1) \in T^n_\epsilon(P_{VU_1Y_1})$$
					If no or more than one such pairs can be found, declare an error.  It is easy to show that for sufficiently large n, we can correctly find such pair with high probability if 
					\begin{equation}
					R_2+\tR_2\leqslant I(V;Y_1|U_1)
					\end{equation}
					
					After successively decoding $v^n$, find the unique tuple $(w_{11},\tl_{11},w_{12},\tl_{12})$ such that
					$$((u_1^n(\hw_{11}, \hl_{11}),v^n(w_2,\hl_2),u_2^n(w_{11},\tl_{11},w_{12},\tl_{12}),y^n_1) \in T^n_\epsilon(P_{VU_1U_2Y_1}))$$
					If no or more than one such pairs can be found, declare an error.  It is easy to show that for sufficiently large n, we can correctly find such pair with high probability if 
					\begin{equation}
					R_{12}+\tR_{12}\leqslant I(U_2;VY_1|U_1)
					\end{equation}

					\item Decoder 2:  Given $y^n_2$, find $(\hw_{11}, \hl_{11})$ such that $$(u_1^n(\hw_{11}, \hl_{11}),y^n_2) \in T^n_\epsilon(P_{U_1Y_1}).$$  If no or more than one such pair $(\hw_{11}, \hl_{11})$ can be found, declare an error. It is easy to show that for sufficiently large n, we can correctly find such pair with high probability if 
					\begin{equation}
					R_{11}+\tR_{11}\leqslant I(U_1;Y_2).
					\end{equation}
					After decoding $u_1^n$, find the unique pair $(\hw_{2}, \hl_{2})$ such that $$(u_1^n(\hw_{11}, \hl_{11}),v^n(w_2,\hl_2),y^n_2) \in T^n_\epsilon(P_{VU_1Y_2})$$
					If no or more than one such pairs can be found, declare an error.  It is easy to show that for sufficiently large n, we can correctly find such pair with high probability if 
					\begin{equation}
					R_2+\tR_2\leqslant I(V;Y_2|U_1)
					\end{equation}
					
					After successively decoding $v^n$, find the unique tuple $(w_{11},\tl_{11},w_{12},\tl_{12})$ such that
					$$((u_1^n(\hw_{11}, \hl_{11}),v^n(w_2,\hl_2),u_2^n(w_{11},\tl_{11},w_{12},\tl_{12}),y^n_1) \in T^n_\epsilon(P_{VU_1U_2Y_2}))$$
					If no or more than one such pairs can be found, declare an error.  It is easy to show that for sufficiently large n, we can correctly find such pair with high probability if 
					\begin{equation}
					R_{12}+\tR_{12}\leqslant I(U_2;VY_2|U_1)
					\end{equation}
				\end{itemize}
				The corresponding achievable region is thus characterized by
				\begin{flalign}
				R_{11}&\leqslant \min\{I(U_1;Y_1), I(U_1;Y_2)\}-I(U_1;S_1)\label{eq:pps4-3}\\
				R_{12}&\leqslant \min\{I(U_2;VY_1|U_1), I(U_2;VY_2|U_1)\}-I(U_2;S_1|U_1)\label{eq:pps4-4}\\
				R_{2}&\leqslant \min\{I(V;Y_1|U_1), I(V;Y_2|U_1)\}-I(V;S_1)\label{eq:pps4-5}
				\end{flalign}
			\end{enumerate}
			Proposition \ref{pps:strong_inner} is thus proved by combining \cref{eq:pps4-1,eq:pps4-2} with \cref{eq:pps4-3,eq:pps4-4,eq:pps4-5}.
			
	\chapter{Conclusions and Future Work}\label{sec:conclusion}
	
	In this chapter, we first summarize our results in this dissertation, and then discuss some future directions.

	\section{Summary of Dissertation}
	In Chapter \ref{cha:helper}, we studied the state-dependent MAC with a helper. Our achievable scheme is based on integration of state subtraction and single-bin dirty paper coding. By analyzing the corresponding lower bound on the capacity, and comparing to the upper bounds, we characterized the capacity for various channel parameters. We anticipate that our way of analyzing the lower bound and characterizing the capacity can be applied to characterizing the capacity for other state-dependent networks. We further point out closely related problems of state masking \cite{Merhav07}, state amplification \cite{Kim08}, assisted interference suppression \cite{Grover10,Chou12}, which have a similar goal of minimizing the impact of the state on the output. It will be interesting to explore if the understanding here can shed any insight on these problems.		
	
	In Chapter \ref{chap:z2state}, we studied the state-dependent Gaussian Z-IC with receivers corrupted by two {\em correlated} states which are noncausally known at transmitters. We characterized the conditions on the channel parameters under which the state-dependent Z-IC achieves the capacity region or sum-capacity of the corresponding channel without state for the very strong regime, strong regime and weak regime.
	
	
	In Chapter \ref{cha:2state}, we studied the state-dependent Gaussian IC with receivers corrupted by two {\em correlated} states which are noncausally known at transmitters. We developed a new scheme that can simultaneously cancel the two states and interference, and analyze the impact of correlation between states on the achievability of the capacity region. Our comparison between the IC and the Z-IC suggests that the IC benefits more if the correlation between the states increases. We anticipate that the state cancellation schemes we developed here can be useful for studying other state-dependent models.
	
	\section{Future Works}

	In this dissertation, we characterized the capacity region for two-types of state-dependent models under various channel parameters. However, there are still channel parameters under which the channel capacity is still unknown. For the state-dependent MAC channel with a helper, the full capacity region remains unknown due to the lack of tighter inner and/or outer bounds. In the recent work \cite{Yang17}, a new outer bound is introduced which results in an improvement on the capacity characterization for this channel. Such a technique can be useful to characterize the full capacity of the state-dependent MAC channel with a helper model. For the state-dependent IC with correlated states noncausally known at the transmitters, we studied only the cases, in which both interferences are strong or weak. It is interesting to study the cases in which one interference is strong and one is weak.
	
	As mentioned in Chapter \ref{chap:Introduction}, there are two primary goals of information theory. The first is the development of the fundamental theoretical limits on the achievable performance when communicating a given information source over a given communications channel using coding schemes from a prescribed class. The second goal is the development of practical coding schemes, e.g. structured encoder(s) and decoder(s), that provide performance reasonably good in comparison with the optimal performance given by the theory. As our theoretical results on the capacity of these models fulfill the first goal for the channels of interest, the development of practical coding schemes becomes a natural direction to fulfill the second goal. 
	
	The lattice coding is an attractive practical coding scheme to deal with state. The proof in \cite{Erez2005} showed that the capacity loss due to lattice coding is limited by its "shaping gain", which is very small. Thus, it would be interesting as a future direction to apply the lattice coding strategy to our models to develop more practical coding schemes.


\bibliographystyle{IEEEtran}
\bibliography{Thesis_Yunhao}

\begin{thebibliography}{10}
\providecommand{\url}[1]{#1}
\csname url@samestyle\endcsname
\providecommand{\newblock}{\relax}
\providecommand{\bibinfo}[2]{#2}
\providecommand{\BIBentrySTDinterwordspacing}{\spaceskip=0pt\relax}
\providecommand{\BIBentryALTinterwordstretchfactor}{4}
\providecommand{\BIBentryALTinterwordspacing}{\spaceskip=\fontdimen2\font plus
\BIBentryALTinterwordstretchfactor\fontdimen3\font minus
  \fontdimen4\font\relax}
\providecommand{\BIBforeignlanguage}[2]{{%
\expandafter\ifx\csname l@#1\endcsname\relax
\typeout{** WARNING: IEEEtran.bst: No hyphenation pattern has been}%
\typeout{** loaded for the language `#1'. Using the pattern for}%
\typeout{** the default language instead.}%
\else
\language=\csname l@#1\endcsname
\fi
#2}}
\providecommand{\BIBdecl}{\relax}
\BIBdecl

\bibitem{Shannon1948}
C.~E. Shannon, ``A mathematical theory of communication,'' \emph{The Bell
  System Technical Journal}, vol.~27, pp. 379--423, 623--656, 1948.

\bibitem{Shannon59}
------, ``{Coding theorems for a discrete source with a fidelity criterion},''
  in \emph{IRE Nat. Conv. Rec., Pt. 4}, 1959, pp. 142--163.

\bibitem{Csiszar98}
\BIBentryALTinterwordspacing
I.~Csisz{\'{a}}r, ``The method of types,'' \emph{{IEEE} Trans. Information
  Theory}, vol.~44, no.~6, pp. 2505--2523, 1998. [Online]. Available:
  \url{https://doi.org/10.1109/18.720546}
\BIBentrySTDinterwordspacing

\bibitem{Mcmillan53}
\BIBentryALTinterwordspacing
B.~McMillan, ``The basic theorems of information theory,'' \emph{The Annals of
  Mathematical Statistics}, vol.~24, no.~2, pp. 196--219, 1953. [Online].
  Available: \url{http://www.jstor.org/stable/2236328}
\BIBentrySTDinterwordspacing

\bibitem{breiman1957}
\BIBentryALTinterwordspacing
L.~Breiman, ``The individual ergodic theorem of information theory,''
  \emph{Ann. Math. Statist.}, vol.~28, no.~3, pp. 809--811, 1957. [Online].
  Available: \url{http://dx.doi.org/10.1214/aoms/1177706899}
\BIBentrySTDinterwordspacing

\bibitem{Gray1990}
R.~M. Gray, \emph{Entropy and Information Theory}.\hskip 1em plus 0.5em minus
  0.4em\relax Springer-Verlag, Inc. New York, NY, USA, 1990.

\bibitem{Han81}
T.~S. Han and K.~Kobayashi, ``A new achievable rate region for the interference
  channel,'' \emph{IEEE Trans. Inform. Theory}, vol.~27, no.~1, pp. 49--60,
  January 1981.

\bibitem{Saito13}
Y.~Saito, A.~Benjebbour, Y.~Kishiyama, and T.~Nakamura, ``System-level
  performance evaluation of downlink non-orthogonal multiple access {(NOMA)},''
  in \emph{Proc. IEEE Intl. Symp. Personal, Indoor and Mobile Radio
  Communications (PIMRC)}, London, U.K., Sep. 2013.

\bibitem{Li14}
Q.~Li, H.~Niu, A.~Papathanassiou, and G.~Wu, ``5{G} network capacity: {K}ey
  elements and technologies,'' \emph{IEEE Veh. Technol. Mag.}, vol.~9, pp.
  71--78, 2014.

\bibitem{Xu15}
P.~Xu, Z.~Ding, X.~Dai, and H.~V. Poor, ``{NOMA}: {A}n information theoretic
  perspective,'' arXiv http://arxiv.org/abs/1504.07751, 2015.

\bibitem{Mitola:IPC:99}
J.~Mitola, ``Cognitive radio: {M}aking software radios more personal,''
  \emph{IEEE Personal Communications}, vol.~6, pp. 13--18, Aug. 1999.

\bibitem{Haykin:JSAC:05}
S.~Haykin, ``Cognitive radio: {B}rain-empowered wireless communications,''
  \emph{IEEE Journal on Selected Areas in Communications}, vol.~23, pp.
  201--220, Feb. 2005.

\bibitem{Akyi06}
I.~F. Akyildiz, W.-Y. Lee, M.~C. Vuran, and S.~Mohanty, ``Next
  generation/dynamics spectrum sccess/cognitive radio wireless networks: {A}
  survey,'' \emph{Computer Networks}, vol.~50, pp. 2127--2159, 2006.

\bibitem{Ashton09}
K.~Ashton, ``That 'internet of things' thing,'' \emph{RFiD Journal}, 2009.

\bibitem{Atzori10}
L.~Atzori, A.~Iera, and G.~Morabiton, ``The internet of things: {A} survey,''
  \emph{Computer Networks}, vol.~54, pp. 2787--2805, 2010.

\bibitem{Gubbi13}
J.~Gubbi, R.~Buyya, S.~Marusic, and M.~Palaniswami, ``Internet of things
  {(IoT)}: {A} vision, architectural elements, and future directions,''
  \emph{Future Generation Computer Systems}, vol.~29, pp. 1645--1660, 2013.

\bibitem{Shannon1958}
\BIBentryALTinterwordspacing
C.~E. Shannon, ``Channels with side information at the transmitter,'' \emph{IBM
  J. Res. Dev.}, vol.~2, no.~4, pp. 289--293, Oct. 1958. [Online]. Available:
  \url{http://dx.doi.org/10.1147/rd.24.0289}
\BIBentrySTDinterwordspacing

\bibitem{Gelf80}
S.~Gel'fand and M.~Pinsker, ``Coding for channels with ramdom parameters,''
  \emph{Probl. Contr. Inf. Theory}, vol.~9, no.~1, pp. 19--31, January 1980.

\bibitem{Costa83}
M.~H.~M. Costa, ``Writing on dirty paper,'' \emph{IEEE Trans. Inform. Theory},
  vol.~29, no.~3, pp. 439--441, May 1983.

\bibitem{Kim05}
S.~Sigurj\'{o}nsson and Y.-H. Kim, ``On multiple user channels with state
  information at the transmitters,'' in \emph{Proc.\ IEEE International
  Symposium on Information Theory (ISIT)}, Adelaide, Australia, Sep. 2005.

\bibitem{Lapidoth13_a}
A.~Lapidoth and Y.~Steinberg, ``The multiple-access channel with causal side
  information: Common state,'' \emph{IEEE Trans. Inform. Theory}, vol.~59,
  no.~1, pp. 32--50, Jan. 2013.

\bibitem{Zaidi13}
A.~Zaidi, P.~Piantanida, and S.~{Shamai (Shitz)}, ``Capacity region of
  cooperative multiple access channel with states,'' \emph{IEEE Trans. Inform.
  Theory}, vol.~59, no.~10, pp. 6153--6174, October 2013.

\bibitem{Lapidoth13_b}
A.~Lapidoth and Y.~Steinberg, ``The multiple-access channel with causal side
  information: Double state,'' \emph{IEEE Trans. Inform. Theory}, vol.~59,
  no.~3, pp. 1379--1393, Mar. 2013.

\bibitem{Li13}
M.~Li, O.~Simeone, and A.~Yener, ``Multiple access channels with state causally
  known at transmitters,'' \emph{IEEE Trans. Inform. Theory}, vol.~59, no.~3,
  pp. 1394--1404, Mar. 2013.

\bibitem{Steinberg05_2}
Y.~Steinberg and S.~{Shamai (Shitz)}, ``Achievable rates for the broadcast
  channel with states known at the transmitter,'' in \emph{Proc. IEEE Int.
  Symp. Information Theory (ISIT)}, Adelaide, Australia, July 2005.

\bibitem{Khos09}
R.~Khosravi-Farsani, B.~Akhbari, M.~Mirmohseni, and M.~Aref, ``Cooperative
  relay-broadcast channels with causal channel state information,'' in
  \emph{Proc.\ IEEE International Symposium on Information Theory (ISIT)},
  Seoul, Korea, Jul. 2009, pp. 1174--1178.

\bibitem{Lapidoth11}
A.~Lapidoth and L.~Wang, ``The state-dependent semideterministic broadcast
  channel,'' \emph{IEEE Trans. Inform. Theory}, vol.~59, no.~4, pp. 2242--2251,
  Apr. 2013.

\bibitem{Duan14ITW}
R.~Duan, Y.~Liang, and S.~{Shamai (Shitz)}, ``Dirty interference cancellation
  for {G}aussian broadcast channels,'' in \emph{Proc. IEEE Information Theory
  Workshop (ITW) on Information Theory for Wireless Networks}, Hobart,
  Tasmania, Australia, Nov. 2014.

\bibitem{Aref09}
B.~Akhbari, M.~Mirmohseni, and M.~R. Aref, ``Compress-and-forward strategy for
  the relay channel with non-causal state information,'' in \emph{Proc.\ IEEE
  International Symposium on Information Theory (ISIT)}, Seoul, Korea, Jul.
  2009.

\bibitem{Zaidi10}
A.~Zaidi, S.~P. Kotagiri, J.~N. Laneman, and L.~Vandendorpe, ``Cooperative
  relaying with state available noncausally at the relay,'' \emph{IEEE Trans.
  Inform. Theory}, vol.~56, no.~5, pp. 2272--2298, May 2010.

\bibitem{Zaidi11}
A.~Zaidi, S.~{Shamai (Shitz)}, P.~Piantanida, and L.~Vandendorpe, ``Bounds on
  the capacity of the relay channel with noncausal state at source,''
  \emph{IEEE Trans. Inform. Theory}, vol.~59, no.~5, pp. 2639--2672, May 2013.

\bibitem{Zhang11a}
L.~Zhang, S.~Cui, and J.~Jiang, ``Gaussian interference channel with state
  information,'' in \emph{Proc. IEEE Wireless Communications and Networking
  Conference}, March 2011, pp. 1960--1965.

\bibitem{Zhang11b}
L.~Zhang, T.~Liu, and S.~Cui, ``Symmetric {G}aussian interference channel with
  state information,'' in \emph{Proc. 49th Annual Allerton Conference on
  Communication, Control, and Computing}, Sep. 2011.

\bibitem{Ghas13}
S.~Ghasemi-Goojani and H.~Behroozi, ``On the achievable rate-regions for
  state-dependent {G}aussian interference channel,'' Available at
  http://arxiv.org/abs/1301.5535, submitted in January 2013.

\bibitem{Duan16IT}
R.~Duan, Y.~Liang, and S.~{Shamai (Shitz)}, ``State-dependent gaussian
  interference channels: Can state be fully canceled?'' \emph{IEEE Trans.
  Inform. Theory}, vol.~62, pp. 1957 -- 1970, Apr. 2016.

\bibitem{Duan13ITW}
R.~Duan, Y.~Liang, A.Khisti, and S.~{Shamai (Shitz)}, ``State-dependent
  {G}aussian {Z}-channel with mismatched side-information and interference,''
  in \emph{Proc. IEEE Information Theory Workshop (ITW) on Information Theory
  for Wireless Networks}, Sevilla, Spain, Sep. 2013.

\bibitem{Ghas14}
S.~Ghasemi-Goojani and H.~Behroozi, ``State-dependent {G}aussian
  {Z}-interference channel: New results,'' in \emph{Proc.\ International
  Symposium on Information Theory and its Applications (ISITA)}, Victoria,
  Australia, Oct. 2014, pp. 468--472.

\bibitem{Fehri2015Z}
H.~Fehri and H.~K. Ghomash, ``Z-interference channel with side information at
  the transmitters,'' \emph{AEU - International Journal of Electronics and
  Communications}, vol.~69, no.~9, pp. 1167--1180, 2015.

\bibitem{Haji13}
S.~Hajizadeh, M.~Monemizadeh, and E.~Bahmani, ``State-dependent {Z} channel,''
  in \emph{Proc.\ Conference on Information Sciences and Systems (CISS)},
  Princeton, NJ, USA., Mar. 2014.

\bibitem{Somekh08}
A.~Somekh-Baruch, S.~{Shamai (Shitz)}, and S.~Verd\'{u}, ``Cognitive
  interference channels with state information,'' in \emph{Proc. IEEE Int.
  Symp. Information Theory (ISIT)}, Toronto, Canada, July 2008.

\bibitem{Kazemi13ISIT}
M.~Kazemi and A.~Vosoughi, ``On the capacity of the state-dependent cognitive
  interference channel,'' in \emph{Proc.\ IEEE International Symposium on
  Information Theory (ISIT)}, Istanbul, Turkey, Jul. 2013.

\bibitem{Duan12ISIT}
R.~Duan and Y.~Liang, ``Gaussian cognitive interference channels with state,''
  in \emph{Proc.\ IEEE International Symposium on Information Theory (ISIT)},
  Boston, MA, Jul. 2012.

\bibitem{Duan15IT}
------, ``Bounds and capacity theorems for cognitive interference channels with
  state,'' \emph{IEEE Trans. Inform. Theory}, vol.~61, no.~1, pp. 280--304,
  Jan. 2015.

\bibitem{Mallik08}
S.~Mallik and R.~Koetter, ``Helpers for cleaning dirty papers,'' in \emph{Proc.
  IEEE International ITG Conference on Source and Channel Coding (SCC)}, Ulm,
  Germany, Jan 2008.

\bibitem{Laneman08}
S.~P. Kotagiri and J.~N. Laneman, ``Multiaccess channels with state known to
  some encoders and independent messages,'' \emph{EURASIP {J}ournal on
  {W}ireless {C}ommunications and {N}etworking}, 2008.

\bibitem{Zaidi09}
A.~Zaidi, S.~P. Kotagiri, J.~N. Laneman, and L.~Vandendorpe, ``Multiaccess
  channels with state known to one encoder: {A}nother case of degraded message
  sets,'' in \emph{Proc.\ IEEE International Symposium on Information Theory
  (ISIT)}, Seoul, Korea, Jul. 2009.

\bibitem{Somekh08MAC}
A.~Somekh-Baruch, S.~{Shamai (Shitz)}, and S.~Verd\'{u}, ``Cooperative
  multiple-access encoding with states available at one transmitter,''
  \emph{IEEE Trans. Inform. Theory}, vol.~54, no.~10, pp. 4448--4469, October
  2008.

\bibitem{Phil11}
T.~Philosof, R.~Zamir, U.~Erez, and A.~J. Khisti, ``Lattice strategies for the
  dirty multiple access channel,'' \emph{IEEE Trans. Inform. Theory}, vol.~57,
  no.~8, pp. 5006--5035, August 2011.

\bibitem{Duan14TIT}
R.~Duan, Y.~Liang, A.Khisti, and S.~{Shamai (Shitz)}, ``Parallel {G}aussian
  networks with a common state-cognitive helper,'' \emph{IEEE Trans. Inform.
  Theory}, vol.~61, no.~12, pp. 6680--6699, December 2015.

\bibitem{Yunhao16IT}
Y.~Sun, R.~Duan, Y.~Liang, A.~Khisti, and S.~S. Shitz, ``Capacity
  characterization for state-dependent gaussian channel with a helper,''
  \emph{IEEE Transactions on Information Theory}, vol.~62, no.~12, pp.
  7123--7134, Dec 2016.

\bibitem{Kim04}
Y.-H. Kim, A.~Sutivong, and S.~Sigurj\'{o}nsson, ``Multiple user writing on
  dirty paper,'' in \emph{Proc.\ IEEE International Symposium on Information
  Theory (ISIT)}, Chicago, Illinois, Jul. 2004.

\bibitem{Gelf84}
S.~Gel'fand and M.~Pinsker, ``On {G}aussian channels with random parameters,''
  in \emph{Proc.\ IEEE International Symposium on Information Theory (ISIT)},
  Tashkent, USSR, Sep. 1984.

\bibitem{Phil08}
T.~Philosof and R.~Zamir, ``The rate loss of single letter characterization for
  the ``dirty'' multiple access channel,'' in \emph{Proc. IEEE Information
  Theory Workshop (ITW) on Information Theory for Wireless Networks}, May 2008,
  pp. 31--35.

\bibitem{Yunhao16ISIT}
Y.~Sun, R.~Duan, Y.~Liang, A.~Khisti, and S.~S. (Shitz), ``Helper-assisted
  state cancellation for multiple access channels,'' in \emph{Proc.\ IEEE
  International Symposium on Information Theory (ISIT)}, Barcelona, Spain, Jul.
  2016.

\bibitem{Yunhao17ISIT}
Y.~Sun, R.~Duan, Y.~Liang, and S.~S. (Shitz), ``State-dependent z-interference
  channel with correlated states,'' in \emph{Proc.\ IEEE International
  Symposium on Information Theory (ISIT)}, Aachen, Germany, Jun. 2017.

\bibitem{Yunhao17IT}
Y.~Sun, R.~Duan, Y.~Liang, and S.~{Shamai (Shitz)}, ``State-dependent
  interference channel with correlated states,'' submitted to {\em IEEE Trans.
  Inform. Theory}, October 2017.

\bibitem{Duan14ISITA}
R.~Duan, Y.~Liang, and S.~{Shamai (Shitz)}, ``Dirty interference cancelation
  for multiple access channels,'' in \emph{Proc.\ International Symposium on
  Information Theory and its Applications (ISITA)}, Melbourne, Australia, Oct.
  2014.

\bibitem{Sason04}
I.~Sason, ``On achievable rate regions for the {G}aussian interference
  channel,'' \emph{IEEE Trans. Inform. Theory}, vol.~50, no.~6, pp. 1345--1356,
  June 2004.

\bibitem{Sato81}
H.~Sato, ``The capacity of the {G}aussian interference channel under strong
  interference,'' \emph{IEEE Trans. Inform. Theory}, vol.~27, no.~6, pp.
  786--788, Nov. 1981.

\bibitem{Shang09}
X.~Shang, G.~Kramer, and B.~Chen, ``A new outer bound and the
  noisy-interference sum-rate capacity for {G}aussian interference channels,''
  \emph{IEEE Trans. Inform. Theory}, vol.~55, no.~2, pp. 689--699, February
  2009.

\bibitem{Anna09}
V.~S. Annapureddy and V.~V. Veeravalli, ``Gaussian interference networks: Sum
  capacity in the low interference regime and new outer bounds on the capacity
  region,'' \emph{IEEE Trans. Inform. Theory}, vol.~55, no.~7, pp. 3032--3050,
  July 2009.

\bibitem{Mota09}
A.~S. Motahari and A.~K. Khandani, ``Capacity bounds for the {G}aussian
  interference channel,'' \emph{IEEE Trans. Inform. Theory}, vol.~55, no.~2,
  pp. 620--643, February 2009.

\bibitem{Merhav07}
N.~Merhav and S.~{Shamai (Shitz)}, ``Information rates subject to state
  masking,'' \emph{IEEE Trans. Inform. Theory}, vol.~53, no.~6, pp. 2254--2261,
  June 2007.

\bibitem{Kim08}
Y.-H. Kim, A.~Sutivong, and T.~M. Cover, ``State amplification,'' \emph{IEEE
  Trans. Inform. Theory}, vol.~54, no.~5, pp. 1850--1859, May 2008.

\bibitem{Grover10}
P.~Grover and A.~Sahai, ``Vector {W}itsenhausen counterexample as assisted
  interference suppression,'' \emph{International Journal of Systems, Control
  and Communications}, 2010.

\bibitem{Chou12}
C.~Choudhuri and U.~Mitra, ``On {W}itsenhausen's counterexample: {T}he
  asymptotic vector case,'' in \emph{Proc.\ IEEE Information Theory Workshop
  (ITW) on Information Theory for Wireless Networks}, 2012.

\bibitem{Yang17}
W.~Yang, Y.~Liang, S.~Shamai, and H.~V. Poor, ``State-dependent gaussian
  multiple access channels: New outer bounds and capacity results,'' in
  \emph{Proc. IEEE Int. Symp. Information Theory (ISIT)}, Aachen, Germany, Jul.
  2017.

\bibitem{Erez2005}
U.~Erez, S.~{Shamai (Shitz)}, and R.~Zamir, ``Capacity and lattice strategies
  for canceling known interference,'' \emph{IEEE Trans. Inform. Theory},
  vol.~51, no.~11, pp. 3820--3833, Nov. 2005.

\end{thebibliography}
\begin{vita}
	
\end{vita}
\end{document}